\documentclass{amsart}
\usepackage{mleftright}
\usepackage{hyperref}

\usepackage[T1]{fontenc}
\usepackage{microtype}

\usepackage{amssymb}
\usepackage{esint}
\usepackage{etoolbox}
\usepackage{mathtools}
\usepackage{slashed}

\newtheorem{theorem}{Theorem}[section]
\newtheorem{lemma}[theorem]{Lemma}
\newtheorem{proposition}[theorem]{Proposition}
\newtheorem{corollary}[theorem]{Corollary}

\theoremstyle{definition}
\newtheorem{definition}[theorem]{Definition}
\newtheorem{example}[theorem]{Example}

\theoremstyle{remark}
\newtheorem{remark}[theorem]{Remark}

\numberwithin{equation}{section}

\let\epsilon\varepsilon
\let\phi\varphi

\newcommand{\bc}{\mathbf{c}}
\newcommand{\bo}{\mathbf{0}}
\newcommand{\bw}{\mathbf{w}}
\newcommand{\bx}{\mathbf{x}}
\newcommand{\by}{\mathbf{y}}
\newcommand{\bz}{\mathbf{z}}

\newcommand{\bA}{\mathbb{A}}
\newcommand{\bC}{\mathbb{C}}
\newcommand{\bN}{\mathbb{N}}
\newcommand{\bQ}{\mathbb{Q}}
\newcommand{\bR}{\mathbb{R}}
\newcommand{\bT}{\mathbb{T}}
\newcommand{\bZ}{\mathbb{Z}}

\newcommand{\cA}{\mathcal{A}}
\newcommand{\cB}{\mathcal{B}}
\newcommand{\cE}{\mathcal{E}}
\newcommand{\cF}{\mathcal{F}}
\newcommand{\cH}{\mathcal{H}}
\newcommand{\cS}{\mathcal{S}}
\newcommand{\cT}{\mathcal{T}}

\newcommand{\fg}{\mathfrak{g}}

\newcommand{\Cstar}{{\ensuremath{C^\ast}}}
\newcommand{\G}{\ensuremath{G}}
\newcommand{\GA}{\ensuremath{G}-\ensuremath{\cA}}
\newcommand{\GCstar}{\ensuremath{G}-\ensuremath{C^\ast}}
\newcommand{\Gstar}{\ensuremath{G}-\ensuremath{\ast}}
\newcommand{\Star}{\ensuremath{\ast}}

\newcommand{\inj}{\hookrightarrow}
\newcommand{\iso}{\overset\sim\to}
\newcommand{\surj}{\twoheadrightarrow}

\DeclareMathOperator{\Aut}{Aut}
\DeclareMathOperator{\di}{d}
\DeclareMathOperator{\Der}{Der}
\DeclareMathOperator{\Dom}{Dom}
\DeclareMathOperator{\End}{End}
\DeclareMathOperator{\GL}{GL}
\DeclareMathOperator{\Hom}{Hom}
\DeclareMathOperator{\id}{Id}
\DeclareMathOperator{\im}{Im}
\DeclareMathOperator{\Ker}{Ker}
\DeclareMathOperator{\Lie}{Lie}
\DeclareMathOperator{\Ran}{Ran}
\DeclareMathOperator{\Skew}{Skew}
\DeclareMathOperator{\Span}{Span}
\DeclareMathOperator{\SO}{SO}
\DeclareMathOperator{\SU}{SU}
\DeclareMathOperator{\Tr}{Tr}
\DeclareMathOperator{\U}{U}

\DeclarePairedDelimiterX\abs[1]\lvert\rvert{
	\ifblank{#1}{\:\cdot\:}{#1}
}
\DeclarePairedDelimiterX\norm[1]\lVert\rVert{
	\ifblank{#1}{\:\cdot\:}{#1}
}
\DeclarePairedDelimiterX\hp[2](){
	\ifblank{#1}{\ifblank{#2}{\:\cdot\:,\:\cdot\:}{\:\cdot\:,#2}}{\ifblank{#2}{#1,\:\cdot\:}{#1,#2}}
}
\DeclarePairedDelimiterX\ip[2]\langle\rangle{
	\ifblank{#1}{\ifblank{#2}{\:\cdot\:,\:\cdot\:}{\:\cdot\:,#2}}{\ifblank{#2}{#1,\:\cdot\:}{#1,#2}}
}

\newcommand{\dual}[1]{\widehat{#1}}
\newcommand{\e}[1]{e\!\left({#1}\right)}
\newcommand{\iot}[1]{{\iota\!\left({#1}\right)}}

\newcommand{\alphast}{\left(\alpha_\ast\right)}

\providecommand\given{}
\newcommand\setSymbol[1][]{\nonscript\:#1\vert\nonscript\:\allowbreak}
\DeclarePairedDelimiterX\set[1]\{\}{%
\renewcommand\given{\setSymbol[\delimsize]}
#1
}
\newcommand{\inv}[1]{{#1}^{-1}}
\newcommand{\rest}[2]{\left.{#1}\right|_{#2}}

\newcommand\dif{\mathop{}\!\mathrm{d}}

\newcommand{\cat}[1]{\mathbf{#1}}
\newcommand{\op}{\mathrm{op}}
\newcommand{\nd}{\mathrm{nd}}
\newcommand{\p}[1]{{{#1}^\prime}}
\newcommand{\du}{\mathrm{d}}
\newcommand{\fin}{\mathrm{fin}}

\allowdisplaybreaks

\begin{document}

\title[Connes--Landi deformations]{A reconstruction theorem for Connes--Landi deformations of commutative spectral triples}

\author{Branimir \'Ca\'ci\'c}
\address{Department of Mathematics, Texas A\&M University, College Station, TX 77843-3368}
\curraddr{Department of Mathematics \& Statistics, University of New Brunswick, PO Box 4400, Fredericton, NB, E3B 5A3, Canada}
\email{bcacic@unb.ca}
%\thanks{}



%\date{\today}

\dedicatory{Dedicated to Marc Rieffel on the occasion of his 75th birthday}

\begin{abstract}
We formulate and prove an extension of Connes's reconstruction theorem for commutative spectral triples to so-called Connes--Landi or isospectral deformations of commutative spectral triples along the action of a compact Abelian Lie group \(G\), also known as toric noncommutative manifolds. In particular, we propose an abstract definition for such spectral triples, where noncommutativity is entirely governed by a deformation parameter sitting in the second group cohomology of the Pontrjagin dual of \(G\), and then show that such spectral triples are well-behaved under further Connes--Landi deformation, thereby allowing for both quantisation from and dequantisation to \(G\)-equivariant abstract commutative spectral triples. We then use a refinement of the Connes--Dubois-Violette splitting homomorphism to conclude that suitable Connes--Landi deformations of commutative spectral triples by a rational deformation parameter are almost-commutative in the general, topologically non-trivial sense.
\end{abstract}

\maketitle

\section{Introduction}\label{intro}

Just as the \(2\)-torus can be deformed along its translation action on itself to obtain the noncommutative \(2\)-tori, whether as \Cstar-algebras, Fr\'echet pre-\Cstar-algebras, or spectral triples, so too can more general smooth manifolds be deformed along an action of an Abelian Lie group to yield noncommutative \Cstar-algebras, Fr\'echet pre-\Cstar-algebras, or spectral triples. In the case of \Cstar-algebras or Fr\'echet pre-\Cstar-algebras, this process is Rieffel's \emph{strict deformation quantisation}~\cite{Rieffel}, whilst in the case of spectral triples and compact Abelian Lie groups, this process is Connes and Landi's \emph{isospectral deformation}~\cite{CL}, which, following Yamashita~\cite{Ya}, we call \emph{Connes--Landi deformation}. In fact, as was first observed by Sitarz~\cite{Sitarz} and V\'arilly~\cite{Varilly}, Connes--Landi deformation can be viewed as none other than the adaptation to spectral triples of strict deformation quantisation along the action of a compact Abelian Lie group.

In this paper, we formulate and prove an extension of Connes's reconstruction theorem for commutative spectral triples~\cite{Con13} to spectral triples that, \emph{a posteriori}, are Connes--Landi deformations along the action of a compact Abelian Lie group \(G\) of spectral triples of the form \((C^\infty(X),L^2(X,E),D)\), where \(X\) is a compact oriented Riemannian \(G\)-manifold and \(D\) is a \G-invariant essentially self-adjoint Dirac-type operator on a \G-equivariant Hermitian vector bundle \(E \to X\), i.e.,  \emph{toric noncommutative manifolds}. More precisely, we propose a suitable abstract definition of \emph{\(\theta\)-commutative spectral triples}, which closely resembles Connes's abstract definition of commutative spectral triple~\cite{Con96,Con13} except for the specification of deformation parameter \(\theta\) in the second group cohomology \(H^2(\dual{G},\bT)\) of the Pontrjagin dual \(\dual{G}\) of \(G\), which completely governs the failure of commutativity; in particular, \(0\)-commutative spectral triples are just \(G\)-equivariant abstract commutative spectral triples. Then, we show that the Connes--Landi deformation of a \(\theta\)-commutative spectral triple by \(\theta^\prime \in H^2(\dual{G},\bT)\) is itself \((\theta+\theta^\prime)\)-commutative, thereby facilitating both quantisation from and dequantisation to \(G\)-equivariant commutative spectral triples, to which we can apply Connes's result. 

In addition to extending Connes's reconstruction theorem, we also clarify a number of aspects of the general theory of Connes--Landi deformation. In particular, we use a refinement of the Connes--Dubois-Violette splitting homomorphism~\cite{CDV} to show that for \(\theta \in H^2(\dual{G},\bT)\) \emph{rational}, viz, of finite order in the group \(H^2(\dual{G},\bT)\), sufficiently well-behaved \(\theta\)-commutative spectral triples are almost-commutative in the general, topologically non-trivial sense proposed by the author~\cite{Ca12,Ca13} and studied by Boeijink and Van Suijlekom~\cite{BVS} and by Boeijink and Van den Dungen~\cite{BVD}. This generalises the now-folkloric example of rational noncommutative \(2\)-tori~\cite{HKS}.

We begin in \S \ref{sec:1} by reviewing Rieffel's theory of strict deformation quantisation  of Fr\'echet pre-\Cstar-algebras in the case of the action of a compact Abelian Lie group~\cite[pp.\ 19-22]{Rieffel}. In particular, we give a detailed, constructive account of the deformation of \G-equivariant finitely generated projective modules over \G-equivariant Fr\'echet pre-\Cstar-algebras, generalising existing results on the deformation of \G-equivariant vector bundles over \G-manifolds~\cite{CDV,BLVS}. In fact, we obtain an explicit formula for the projection onto the deformation of a \G-equivariant finitely generated projective module corresponding to a given \(G\)-invariant projection onto the original module, generalising the concrete examples studied by Connes and Landi~\cite[\S\S II--III]{CL} and by Landi and Van Suijlekom~\cite{LVS05}.

Next, in \S \ref{sec:2}, we recall the general theory of \emph{Connes--Landi deformations} or \emph{isospectral deformations}, first defined by Connes and Landi for \(\bT^2\)-actions on concrete commutative spectral triples~\cite{CL} and then extended by Yamashita to arbitrary \(\bT^2\)-equivariant spectral triples~\cite{Ya}. As Sitarz~\cite{Sitarz} and V\'arilly~\cite{Varilly} first showed, this amounts to a simultaneous strict deformation quantisation of the algebra \(\cA\) of a \(G\)-equivariant spectral triple \((\cA,H,D)\) and of its \(G\)-equivariant representation on the Hilbert space \(H\). In particular, we clarify the role of the group \(H^2(\dual{G},\bT)\) in parametrizing Connes--Landi deformations of a fixed \(G\)-equivariant spectral triple up to \(G\)-equivariant unitary equivalence, and then completely generalise the isomorphisms amongst the Morita--Rieffel equivalences of smooth noncommutative \(n\)-tori parametrized by the densely-defined \(\SO(n,n\vert\bZ)\)-action on the universal cover \(\Skew(n,\bR) \cong \bR^{n(n-1)/2}\) of \(H^2(\bZ^n,\bT) \cong \bT^{n(n-1)/2}\), as introduced by Rieffel and Schwarz~\cite{RS} and studied by Elliott and Li~\cite{EL}.

At last, in \S \ref{sec:3}, we formulate and prove our extension of Connes's reconstruction theorem for commutative spectral triples~\cite[Theorem 1.1]{Con13} to Connes--Landi deformations of \G-equivariant commutative spectral triples. First, by analogy with Connes's abstract definition of commutative spectral triple~\cite{Con96,Con13}, we propose an abstract definition of spectral triples that, \emph{a posteriori}, will be Connes--Landi deformations of \G-equivariant concrete commutative spectral triples, where noncommutativity is entirely governed by a deformation parameter \(\theta \in H^2(\dual{G},\bT)\) through its  associated alternating bicharacter \(\iot{\theta} \in \Hom(\wedge^2\dual{G},\bT)\).

\begin{definition}
Let \((\cA,H,D)\) be a \(G\)-equivariant regular spectral triple, let \(\theta \in H^2(\dual{G},\bT)\), and let \(p \in \bN\). We shall call \((\cA,H,D)\) a \(p\)-dimensional \emph{\(\theta\)-commutative} spectral triple if the following conditions all hold:
\begin{enumerate}\setcounter{enumi}{-1}
	\item \emph{Order zero}: The algebra \(\cA\) is \emph{\(\theta\)-commutative}, viz, \[\forall\bx,\by\in\dual{G}, \; \forall a_\bx \in \cA_\bx, \; \forall b_\by \in \cA_\by, \quad b_\by a_\bx = \e{\iot{\theta}(\bx,\by)} a_\bx b_\by,\] so that the \G-equivariant \Star-representation \(L : \cA \to B(H)\) of \(\cA\) can be deformed to a \G-equivariant \Star-homomorphism  \(R :  \cA^\op \to B(H)\),  such that for all \(a\), \(b \in \cA\), \([L(a),R(b)] = 0\).
	\item \emph{Dimension}: The eigenvalues \(\set{\lambda_n}_{n \in \bN}\) of \((D^2+1)^{-1/2}\), counted with multiplicity and arranged in decreasing order, satisfy \(\lambda_n = O(n^{-1/p})\) as \(n \to +\infty\).
	\item \emph{Order one}: For all \(a\), \(b\ \in \cA\), \([[D,L(a)],R(b)] = 0\).
	%\item \emph{Regularity}: For all \(a \in \cA\), \(a\), \([D,a] \in \cap_k \Dom (\ad \ds{D})^k\).
	\item \emph{Orientability}: Define \(\epsilon_\theta : (\cA^\fin)^{\otimes(p+1)} \to (\cA^\fin)^{\otimes(p+1)}\) by
\begin{multline*}
	\epsilon_\theta(a_0 \otimes a_1 \otimes \cdots \otimes a_p)\\ \coloneqq \frac{1}{p!} \sum_{\pi \in S_p} \exp\left(2\pi i \mkern-18mu \sum_{\substack{i < j\\ \pi(i)>\pi(j)}} \mkern-18mu \iot{\theta}(\bx_{\pi(i)},\bx_{\pi(j)})\right)(-1)^\pi a_0 \otimes a_{\pi(1)} \otimes \cdots \otimes a_{\pi(p)}
\end{multline*}
for isotypic elements \(a_0 \in \cA_{\bx_0}\), \ldots, \(a_p \in \cA_{\bx_p}\) of \(\cA\), and say that \(\bc \in (\cA^\fin)^{\otimes(p+1)}\) is \emph{\(\theta\)-antisymmetric} if \(\epsilon_\theta(\bc) = \bc\). Define \(\pi_D : \cA^{\otimes(p+1)} \to B(H)\) by
\[
	\forall a_0,\; a_1, \dotsc, a_p \in \cA, \quad \pi_D(a_0 \otimes a_1 \otimes \cdots \otimes a_p) \coloneqq L(a_0)[D,L(a_1)] \cdots [D,L(a_p)].
\]
There exists a \(G\)-invariant \(\theta\)-antisymmetric \(\bc \in (\cA^\fin)^{\otimes(p+1)}\), such that \(\chi := \pi_D(\bc)\) is a self-adjoint unitary, satisfying
\[
	\forall a \in \cA, \quad L(a)\chi = \chi L(a), \quad [D,L(a)] \chi = (-1)^{p+1} \chi [D,L(a)].
\]
	\item \emph{Finiteness and absolute continuity}: The subspace \(\cH_\infty := \cap_k \Dom \abs{D}^k\) defines a \(G\)-equivariant finitely-generated projective\ right \(\cA\)-module admitting a \(G\)-equivariant Hermitian metric \(\hp{}{}\), such that
\[
	\forall \xi, \; \eta \in \cH_\infty \quad \ip{\xi}{\eta} = \Tr_\omega\left( R\left(\hp{\xi}{\eta}\right) (D^2+1)^{-p/2}\right),
\]
where \(\Tr_\omega\) is a fixed Dixmier trace.
	\item \emph{Strong regularity}: For all \(T \in \End_{\cA^\op}(\cH_\infty)\),  \(T \in \cap_k \Dom ([\abs{D},\cdot]^k)\).
\end{enumerate}
\end{definition}

\noindent Observe, in particular, that a \(0\)-commutative spectral triple is precisely a \(G\)-equivar\-i\-ant abstract commutative spectral triple.

Our main result is actually the following, which both guarantees that Connes--Landi deformations of \G-equivariant commutative spectral triples do, indeed, satisfy our proposed abstract definition, and yields the promised extension of Connes's reconstruction theorem by facilitating dequantisation of such spectral triples to commutative spectral triples.

\begin{theorem}\label{mainthm}
Let \((\cA,H,D)\) be a \(p\)-dimensional \(\theta\)-commutative spectral triple and let \(\theta^\prime \in H^2(\dual{G},\bT)\). Then the Connes--Landi deformation \((\cA_\theta,H,D)\) of \((\cA,H,D)\) by \(\theta^\prime\) defines a \(p\)-dimensional \((\theta+\theta^\prime)\)-commutative spectral triple.
\end{theorem}

This result, whose statement and proof require only our abstract definition and the general machinery of strict deformation quantisation and Connes--Landi deformation recalled in \S\S \ref{sec:1}--\ref{sec:2}, yields an abstract generalisation of the seminal results of Connes and Landi~\cite[Theorem 6]{CL} and of Connes and Dubois-Violette~\cite[Theorem 9]{CDV} on deformations of concrete, viz, explicitly differential-geometric, commutative spectral triples along torus actions.

Now, if \((\cA,H,D)\) is \(\theta\)-commutative, then \((\cA_{-\theta},H,D)\) is \(0\)-commutative and hence commutative \emph{simpliciter} by Theorem~\ref{mainthm}. Thus, Connes's reconstruction theorem for commutative spectral triples and Theorem~\ref{mainthm} immediately yield our extension of the reconstruction theorem.

\begin{theorem}
Let \((\cA,H,D)\) be a \(p\)-dimensional \(\theta\)-commutative spectral triple. There exist a compact \(p\)-dimensional oriented Riemannian \(G\)-manifold \(X\), a \(G\)-equivariant Hermitian vector bundle \(E \to X\), and a \(G\)-invariant essentially self-adjoint Dirac-type operator \(D_E\) on \(E\)  such that
\[
	(\cA,H,D) \cong (C^\infty(X)_\theta, L^2(X,E),D_E).
\]
\end{theorem}

Finally, in \S \ref{sec:4}, we study the structure of \emph{rational} Connes--Landi deformations, viz, Connes--Landi deformations by \(\theta \in H^2(\dual{G},\bT)\) of finite order---if \(G = \bT^2\), then \(\theta \in H^2(\bZ^2,\bT) \cong \bR/\bZ\) is rational if and only if \(\theta \in \bQ/\bZ\). We begin by formulating and proving a refinement of Connes and Dubois-Violette's splitting homomorphism~\cite[Theorem 6]{CDV} that allows us to effect Connes--Landi deformation by \(\theta \in H^2(\dual{G},\bT)\) using only the subgroup \(
	\im\theta \coloneqq \overline{\set{\iot{\theta}(\bx,\cdot) \given \bx \in \dual{G}}}
\)
of \(G\), which is finite if and only if \(\theta\) has finite order. Then, we use our extended reconstruction theorem and our refined splitting homomorphism to conclude that for \(\theta \in H^2(\dual{G},\bT)\) of finite order, any \(\theta\)-commutative spectral triple is almost-commutative, in the general, topologically non-trivial sense first proposed by the author~\cite{Ca12,Ca13} and studied by Boeijink and Van Suijlekom~\cite{BVS} and by Boeijink and Van den Dungen~\cite{BVD}, whenever the finite group \(\im\theta\) acts freely and properly on the underlying manifold. This last result not only recovers H{\o}egh-Krohn and Skjelbred's now-folkloric characterisation of rational noncommutative \(2\)-tori~\cite{HKS}, but also has immediate implications for the computation of the spectral action on rational \(\theta\)-commutative spectral triples.

An extremely preliminary version of some results appeared in the author's doctoral dissertation~\cite[Chapter 7]{mythesis}, which is now completely superseded by this paper.

%
%We conclude \S \ref{CLsection} by stating and proving a discrete analogue of Connes and Dubois-Violette's splitting homomorphism~\cite{CDV}*{Theorem 6}, which can also be viewed as an adaptation-\emph{cum}-generalisation of results by Olesen, Pedersen, and Takesaki for faithful ergodic \(W^\ast\)- and \Cstar-dynamical systems~\cite{OPT}*{Theorems 5.12, 6.3}. In short, we canonically associate to a class \(\theta \in H^2(\dual{G},\bT)\) a \emph{discrete} group \(\Gamma_\theta\) and a class \(\theta_\nd \in H^2(\Gamma_\theta,\bT)\), such that the twisted group algebra \(\cS(\Gamma_\theta,\theta_\nd)\) is \emph{simple}, and then realise a Connes--Landi deformation \((\cA_\theta,H,D)\) as the noncommutative balanced Cartesian product over \(\Gamma_\theta\) of the undeformed spectral triple \((\cA,H,D)\) with the metrically trivial spectral triple \((\cS(\Gamma_\theta,\theta_\nd),L^2(\Gamma_\theta),0)\). We will use this discrete splitting homomorphism in \S \ref{reconstructsection} to study rational Connes--Landi deformations of commutative spectral triples (see below).

\subsection*{Update} 
The map \(\Theta_\ast\) of Lemma~\ref{voldeform} as given in the published version is ill-defined: in general, it will map the \emph{algebraic} tensor product \(\cA^{\otimes(p+1)}\) to the \emph{topological} tensor product \(\cA_\Theta^{\widehat{\otimes}(p+1)}\).
Nonetheless, we do obtain a map \(\Theta_\ast : (\cA^\fin)^{\otimes(p+1)} \to (\cA_\Theta^\fin)^{\otimes(p+1)}\), where, given a Fr\'{e}chet \(G\)-\(\ast\)-algebra \(\cB\), we set \(\cB^\fin := \bigoplus_{\bx \in \dual{G}}^{\mathrm{alg}} \cB_\bx\).
This motivates us to correct Definition~\ref{def:thetacomm}.\ref{item:orient}, the orientability condition of our definition of abstract \(\theta\)-commutative spectral triple, to require that the orientation cycle \(\bc\) lie in \((\cA^\fin)^{\otimes(p+1)}\).
Beyond the obvious superficial corrections required by this change, we have also added a short appendix where we prove that a \(G\)-equivariant concrete commutative spectral triple still defines a \(0\)-commutative spectral triple (Proposition~\ref{prop:appx}).

The author is currently supported by NSERC Discovery Grant RGPIN-2024-04467.

\section{Strict deformation quantisation}\label{sec:1}

\subsection{Group-theoretic preliminaries} 

We begin our review of strict deformation quantisation by fixing notation related to compact Abelian Lie groups and their Pontrjagin duals, and then recalling basic facts about the second group cohomology of discrete Abelian groups.

First, let us denote the additive and multiplicative avatars of the circle group by \(\bT \coloneqq \bR/\bZ\) and \(\U(1) \coloneqq \set{\zeta \in \bC \given \abs{\zeta} = 1}\), respectively, which are isomorphic via the complex exponential \(e : \bT \iso U(1)\) defined by \(e(t) \coloneqq e^{2\pi i t}\). For \(N \in \bN\), we denote the flat \(N\)-torus \(\bR^N/\bZ^N\) by \(\bT^N\), so that \(\exp : \bR^N = \Lie(\bT^N) \to \bT^N = \bR^N/\bZ^N\) is the quotient map; unless stated otherwise, \(\bT^N\) will be endowed with the flat Riemannian metric induced by the standard inner product on \(\mathbb{R}^N\).

Next, let \(G\) be a compact Abelian Lie group with Lie algebra \(\fg := \Lie(G)\). Fix a bi-invariant Riemannian metric \(g\) on \(G\), let \(\ip{}{} := g_0(\cdot,\cdot)\) denote the induced real inner product on \(\fg = T_0 G\), and define the \emph{Casimir element} \(\Omega_G \in S^2 \fg\) for \(G\) by \(\Omega_G :=  -\sum_{k=1}^N v_k \otimes v_k \), where \(\set{v_1,\dotsc,v_N}\) is any orthonormal basis of \(\fg\).

Now, let \(\dual{G} := \Hom_{\text{cts}}(G,\bT) = \Hom_{C^\infty}(G,\bT)\) denote the Pontrjagin dual of \(G\), which, therefore, is a finitely-generated discrete Abelian group; by abuse of notation, let \(\ip{}{} : \dual{G} \times G \to \bT\) denote the pairing of \(\dual{G}\) with \(G\). Observe that for any \(\bx \in \dual{G}\), its differential \(\bx_\ast : \fg = \Lie(G) \to \Lie(\bT) = \bR\) defines an element of \(\fg^\ast \cong \fg\), satisfying \(\ip{\bx}{\exp(v)} =  \exp(\ip{\bx_\ast}{v})\) for all \(v \in \fg\). The assignment \(\bx \mapsto \bx_\ast\) defines a homomorphism \(\dual{G} \to \fg^\ast \cong \fg\), allowing us, in turn, to define a length \(\abs{} : \dual{G} \to [0,+\infty)\) on the discrete group \(\dual{G}\) by \(\abs{\bx} := \sqrt{\ip{\bx_\ast}{\bx_\ast}}\); one can then check that \(\dual{G}\) has rapid decay and polynomial growth of order \(\dim \fg\) with respect to \(\abs{}\).

Finally, let us recall some pertinent facts about the second group cohomology of discrete Abelian groups, which we will later apply to \(\dual{G}\). Let \(\Gamma\) be a discrete Abelian group. For future convenience, let \(B(\Gamma) \coloneqq \Hom(\Gamma^{\otimes 2},\bT)\) denote the Abelian group of all \(\bT\)-valued bicharacters on \(\Gamma\), let 
\(
	S(\Gamma) \coloneqq \set{\Theta \in B(\Gamma) \given \forall \bx, \by \in \Gamma, \; \Theta(\by,\bx) = \Theta(\bx,\by)}
\)
denote the subgroup of all \emph{symmetric} bicharacters, and let
\(
	A(\Gamma) \coloneqq \set{\Theta \in B(\Gamma) \given \forall \bx \in \Gamma, \; \Theta(\bx,\bx) = \bo}
\)
denote the subgroup of \emph{alternating} bicharacters.

Recall that a \emph{multiplier} or \emph{\(2\)-cocycle} on \(\Gamma\) is a map \(M : \Gamma \times \Gamma \to \bT\), satisfying \(M(\bx,\by+\bz) + M(\by,\bz) = M(\bx,\by) + M(\bx+\by,\bz)\)
for all \(\bx\), \(\by\), and \(\bz \in \Gamma\), and \(M(\bx,\bo) = M(\bo,\bx) = 0\) for all \(\bx \in \Gamma\), and let \(Z^2(\Gamma,\bT)\) denote the Abelian group of all multipliers on \(\Gamma\) together with pointwise addition; observe, in particular, that \(B(\Gamma) \leq Z^2(\Gamma,\bT)\). Observe, moreover, that \(Z^2(\Gamma,\bT)\) admits the involution \(M \mapsto M^t\) defined by \(M^t(\bx,\by) \coloneqq M(\by,\bx)\) for all \(\bx\), \(\by \in \Gamma\). If \(T : \Gamma \to \bT\) is a function such that \(T(0) = 0\), define the \emph{\(2\)-coboundary} \(\di T \in Z^2(\Gamma,\bT)\) by setting
\(
	\di T(\bx,\by) \coloneqq T(\bx) + T(\by) - T(\bx+\by)
\)
for all \(\bx\), \(\by \in \Gamma\), and let  \(B^2(\Gamma,\bT) \leq Z^2(\Gamma,\bT)\) denote the subgroup of all coboundaries in \(Z^2(\Gamma,\bT)\). At last, the \emph{second group cohomology} of \(\Gamma\) (with coefficients in \(\bT\)) is defined to be the Abelian group
\(
	H^2(\Gamma,\bT) \coloneqq Z^2(\Gamma,\bT)/B^2(\Gamma,\bT);
\)
in particular, we call \(M_1\), \(M_2 \in Z^2(\Gamma,\bT)\) \emph{cohomologous} if and only if \([M_1] = [M_2] \in H^2(\Gamma,\bT)\), if and only if \(M_1 - M_2 \in B^2(\Gamma,\bT)\). Actual computations, however, will depend on the following extremely useful result.

\begin{theorem}[{Kleppner \cite[\S 7]{Klepp}}]\label{kleppner}
Let \(\Gamma\) be a discrete Abelian group. Then \(B^2(\Gamma,\bT) = \set{M \in Z^2(\Gamma,\bT) \given M^t = M}\), so that the antisymmetrization operation \(Z^2(\Gamma,\bT) \to A(\Gamma)\) defined by \(M \mapsto M - M^t\) descends to an isomorphism \(\iota : H^2(\Gamma,\bT) \iso A(\Gamma)\). Moreover, every \(2\)-cocycle \(M \in Z^2(\Gamma,\bT)\) is cohomologous to a bicharacter \(\Theta \in B(\Gamma)\), yielding a natural isomorphism
\(
	H^2(\Gamma,\bT) \cong B(\Gamma)/S(\Gamma)
\).
\end{theorem}

\begin{remark}
The full statement of Kleppner's result, which applies to much larger class of locally compact Abelian topological groups, only yields injectivity of the map \(\iota\). However, in the case of discrete Abelian groups, Moore group cohomology coincides with ordinary group cohomology with coefficients in the trivial \(\Gamma\)-module \(\bT\), so that \(\iota\) is indeed surjective~\cite[Proof of Theorem 2]{Hughes}.
\end{remark}

\begin{remark}
If \(M \in Z^2(\Gamma,\bT)\), then \([M^t] = -[M]\). Hence, if \(\theta \in H^2(\Gamma,\bT)\), then \([\iot{\theta}] = 2\theta\).
\end{remark}

\begin{example}\label{torus1}
Consider \(\bZ^N = \dual{\bT^N}\). Since
\(
	A(\bZ^N) \cong \bT^{N(N-1)/2}
\)
it follows that \(\iota : H^2(\bZ^N,\bT) \iso A(\bZ^N) \cong \bT^{N(N-1)/2}\) Indeed, if \(\theta \in H^2(\bZ^N,\bT)\), and if \(\theta_{ij} \coloneqq \iot{\theta}(e_i,e_j)\) for \(\set{e_i}\) the standard ordered basis of \(\mathbb{R}^N\), then
\[
	\forall \bx,\; \by \in \Gamma, \quad \iot{\theta}(\bx,\by) = \sum_{1 \leq i < j \leq N} \theta_{ij}(x_i y_j - x_j y_i),
\] 
so that the image of \(\theta\) in \(\bT^{N(N-1)/2}\) is the \(\tfrac{N(N-1)}{2}\)-tuple \((\theta_{ij})_{1\leq i < j \leq N}\). 

Now, in the operator-algebraic and noncommutative-geometric literature, one finds the following constructions of bicharacter representatives for \(\theta \in H^2(\bZ^N,\bT)\).
\begin{enumerate}
	\item Define \(\Theta \in B(\bZ^N)\) by \(\Theta(\bx,\by) \coloneqq \sum_{1 \leq i < j \leq N} \theta_{ij} x_i y_j\).
	\item Define \(\Theta \in B(\bZ^N)\) by \(\Theta(\bx,\by) \coloneqq \sum_{1 \leq i < j \leq N} -\theta_{ij} x_j y_i\).
	\item Choose \((\tfrac{1}{2}\theta_{ij})_{1\leq i <  j \leq N} \in \bT^{N(N-1)/2} \) such that \((\theta_{ij}) = 2(\tfrac{1}{2}\theta_{ij})\), and define \(\Theta \in A(\bZ^N)\) by \(\Theta(\bx,\by) \coloneqq \sum_{1 \leq i < j \leq N} \tfrac{1}{2}\theta_{ij} (x_i y_j - x_j y_i)\).
\end{enumerate}
The first two constructions define right splittings of the short exact sequence
\[
	0 \to S(\bZ^N) \inj B(\bZ^N) \surj H^2(\bZ^N,\bT) \to 0
\]
of Abelian groups, whilst the third does not. Nonetheless, since alternating bicharacter representatives yield considerable algebraic simplifications in the context of strict deformation quantisation, the third construction is frequently used in the operator-algebraic literature.
\end{example}

\subsection{Peter--Weyl theory for  \texorpdfstring{\(G\)}{G}-equivariant Fr\'echet  \texorpdfstring{\Star}{*}-algebras and modules}\label{ssec:peterweyl}

Fix a compact Abelian Lie group \(G\). We simplify our later discussion of strict deformation quantisation and Connes--Landi deformation by reviewing the Peter--Weyl theory of \G-equivariant Fr\'echet \Star-algebras (cf.~\cite[\S 4]{AbEx}) and \G-equivariant Fr\'echet modules over such algebras. We begin by fixing terminology and notation.

\begin{definition}
A \emph{Fr\'echet \Star-algebra} is a Fr\'echet space \(\cA\) endowed with the structure of a \Star-algebra, such that the multiplication \(\cA \times \cA \to \cA\) and the involution \(\ast : \cA \to \cA\) are continuous. In particular, a \emph{Fr\'echet pre-\Cstar-algebra} is a Fr\'echet \Star-algebra \(\cA\) endowed with a continuous injective \Star-homomorphism \(\cA \inj A\) into a \Cstar-algebra \(A\), such that the image of \(\cA\) in \(A\) is dense and is closed under the holomorphic functional calculus.
\end{definition}

\begin{remark}
We consider only \emph{unital} Fr\'echet \Star-algebras.
\end{remark}

\begin{definition}
A \emph{Fr\'echet \Gstar-algebra} is a pair \((\cA,\alpha)\), where \(\cA\) is a unital Fr\'echet \Star-algebra and \(\alpha : G \to \Aut(\cA)\) is an action of \(G\) on \(\cA\) by isometric \Star-automorphisms that is strongly smooth, in the sense that \(t \mapsto \alpha_t(a)\) is smooth for each \(a \in \cA\). If \(\cA\) admits a \G-equivariant continuous injective \Star-homomorphism \(\cA \inj A\) into a unital \GCstar-algebra, such that the image of \(\cA\) in \(A\) is dense, \(G\)-invariant, and closed under the holomorphic functional calculus, then we call \((\cA,\alpha)\) a \emph{Fr\'echet pre-\GCstar-algebra}.
\end{definition}

\begin{example}
Let \(X\) be a compact \G-manifold with \G-action \(\sigma : G \mapsto \operatorname{Diff}(X)\), and define \(\alpha : G \to \Aut(C^\infty(X))\) by setting \(\alpha_t f \coloneqq f \circ \sigma_{-t}\) for all \(f \in C^\infty(X)\), \(t \in G\). Then \((C^\infty(X),\alpha)\) is a Fr\'echet pre-\GCstar-algebra with \Cstar-completion \(C(X)\). 
\end{example}

\begin{example}
The rapid decay algebra
\[
	\cS(\dual{G}) \coloneqq \set*{f : \dual{G} \to \bC \given \forall k \in \bN, \; \norm{f}_{k} \coloneqq \sup_{\bx \in \dual{G}} (1+4\pi^2\abs{\bx}^2)^k \abs{f(\bx)} < +\infty}
\]
of the discrete group \(\dual{G}\) with respect to the length \(\abs{}\), endowed with the seminorms \(\set{\norm{}_k}_{k \in \bN}\), the convolution product, \Star-operation, and \G-action defined by
\begin{gather*}
	\forall f, \; g \in \cS(\dual{G}), \; \forall \bx \in \dual{G}, \quad (f \star g)(\bx) \coloneqq \sum_{\by \in \dual{G}} f(\bx-\by)g(\by), \quad  f^\ast(\bx) \coloneqq \overline{f(-\bx)},\\
	\forall t \in G, \; \forall f \in \cS(\dual{G}), \; \forall \bx \in \dual{G}, \quad \dual{\alpha}_t(f)(\bx) \coloneqq \e{\ip{\bx}{t}}f(\bx),
\end{gather*}
is a Fr\'echet pre-\GCstar-algebra with \Cstar-completion \(C^\ast(\dual{G}) = C^\ast_r(\dual{G})\).
\end{example}

%\begin{example}
%On the one hand, \(C^\infty(G)\), endowed with the \G-action induced by the translation action on \(G\) itself, defines a Fr\'echet pre-\GCstar-algebra with \Cstar-completion \(C(G)\). On the other hand, the rapid decay algebra
%\[
%	\cS(\dual{G}) \coloneqq \set*{f : \dual{G} \to \bC \given \forall k \in \bN, \; \norm{f}_{k} \coloneqq \sup_{\bx \in \dual{G}} (1+4\pi^2\abs{\bx}^2)^k \abs{f(\bx)} < +\infty}
%\]
%of the discrete group \(\dual{G}\) with respect to the length \(\abs{}\), endowed with the seminorms
%\[
%	\forall k \in \bN, \; \forall f \in \cS(\dual{G}), \quad \norm{f}_k \coloneqq \sup_{\bx \in \dual{G}} (1+4\pi^2\abs{\bx}^2)^k \abs{f(\bx)},
%\]
%the convolution product and the \Star-operation
%\[
%	\forall f, \; g \in \cS(\dual{G}), \; \forall \bx \in \dual{G}, \quad (f \star g)(\bx) \coloneqq \sum_{\by \in \dual{G}} f(\bx-\by)g(\by), \quad  f^\ast(\bx) \coloneqq \overline{f(-\bx)},
%\]
%and the \G-action
%\[
%	\forall t \in G, \; \forall f \in \cS(\dual{G}), \; \forall \bx \in \dual{G}, \quad \dual{\alpha}_t(f)(\bx) \coloneqq \e{\ip{\bx}{t}}f(\bx),
%\]
%defines a Fr\'echet pre-\GCstar-algebra with \Cstar-completion \(C^\ast(\dual{G}) = C^\ast_r(\dual{G})\).
%\end{example}

For the remainder of this subsection, let \(\cA\) be a fixed Fr\'echet \Gstar-algebra \(\cA\) with \G-action \(\alpha : G \to \Aut(\cA)\), and let \(\set{\norm{}_k}_{k \in \bN}\) be its defining family of seminorms. Our main goal is to consistently decompose elements of \(\cA\) into rapidly decaying Fourier series of eigenvectors for the \(G\)-action.

First, let us enrich the Fr\'echet topology on \(\cA\) in order to facilitate discussion of the convergence of Fourier series. Since the action \(\alpha\) of \(G\) on \(\cA\) is strongly smooth, it induces an action \(\alpha_\ast : \fg \to \Der(\cA)\) of \(\fg\) on \(\cA\) by continuous \Star-derivations, where
\begin{equation}
	\forall X \in \fg, \; \forall a \in \cA, \quad \alphast_{X}(a) \coloneqq \lim_{s \to 0} \frac{1}{s}\left(\alpha_{\exp(sX)}(a) - a\right).
\end{equation}
In fact, since \(\fg\) is Abelian, \(\alpha_\ast\) extends to an representation \(S\fg \to B(\cA)\) of the symmetric tensor algebra \(S \fg\) on \(\cA\) by continuous \(\bC\)-linear operators, so that the Casimir element \(\Omega_G \in S^2 \fg\) of \(G\) acts as the continuous \(\bC\)-linear operator \(\Delta_G \coloneqq \alphast_{\Omega_G}\). We now use \(\Delta_G\) to define additional seminorms \(\set{\norm{}_{j,k}}_{j,k\in\bN}\) on \(\cA\) by
\begin{equation}\label{extraseminorms}
	\forall j, \; k \in \bN, \; \forall a \in \cA, \quad \norm{a}_{j,k} \coloneqq \norm{(1+\Delta_G)^j a}_k;
\end{equation}
since the \G-action on \(\cA\) was assumed to be strongly smooth and isometric, \(\cA\) remains complete in these additional seminorms.

\begin{definition}[{Yamashita~\cite[pp.\ 255--6]{Ya}}]\label{additional}
Let \((\cA,\alpha)\) be a Fr\'echet \Gstar-algebra. Then \((\cA^\infty,\alpha)\) is the Fr\'echet-\Gstar-algebra defined by \(\cA\) endowed with the larger set of seminorms \(\set{\norm{}_{j,k}}_{j,k \in \bN}\) and the same \G-action \(\alpha\). 
\end{definition}

Next, let us decompose \(\cA\) into eigenspaces or \emph{isotypic subspaces} for the \(G\)-action. For each \(\bx \in \dual{G}\), the \emph{\(\bx\)-isotypic subspace of \(\cA\)} is defined to be the closed subspace
\begin{equation}
	\cA_\bx \coloneqq \set{a \in A \given \forall t \in G, \; \alpha_t(a) = \e{\ip{\bx}{t}}a} \subset \cA.
\end{equation}
The isotypic subspaces of \(\cA\) are linearly independent and satisfy \(\cA_\bx \cA_\by \subseteq \cA_{\bx+\by}\) and \(\cA_\bx^\ast = \cA_{-\bx}\) for all \(\bx\), \(\by \in \dual{G}\), and their algebraic direct sum \(\oplus_{\bx \in \dual{G}}^{\text{alg}} \cA_\bx\) is dense in \(\cA\) by the Peter--Weyl theorem for topological vector spaces \cite[Theorem 5.7]{BtD}.

\begin{remark}
This \(\dual{G}\)-grading of \(\cA\) could be interpreted as a ``Fr\'echet \Star-algebraic bundle'' or ``smooth Fell bundle'' \(\set{\cA_\bx}_{\bx \in \dual{G}} \to \dual{G}\) over \(\dual{G}\) (cf.~\cite[Chapter VIII]{FD} \cite{Exel}).
\end{remark}

\begin{example}\label{canonicalexample0}
For each \(\bx \in \dual{G}\), \(C^\infty(G)_\bx = \bC U_\bx\), where \(U_\bx \coloneqq e \circ \bx : G \to U(1)\); in particular, if \(f \in C^\infty(G)_\bx\), then
\[
	f = \left(\int_G U_\bx^\ast(t)f(t) \dif t \right)U_\bx = \left(\int_G \e{-\ip{\bx}{t}} f(t) \dif t \right)U_\bx.
\]
\end{example}

We can define projection maps onto the isotypic subspaces as follows. For each \(\bx \in \dual{G}\), define a linear contraction \(F_\bx : \cA \surj \cA_\bx\) by
\begin{equation}
	\forall a \in \cA, \quad F_\bx(a) \coloneqq \int_G \e{-\ip{\bx}{t}}\alpha_t(a) \dif t,
\end{equation}
where \(\dif t\) denotes the normalised Haar measure on \(G\); in particular, \(F_\bo\) is the canonical conditional expectation \(\cA \surj \cA_\bo = \cA^G\). The maps \(F_\bx\) satisfy \(F_\bx F_\by = \delta_{\bx,\by} F_\by\) for all \(\bx\), \(\by \in \dual{G}\), and so define mutually orthogonal projections onto the isotypic subspaces, which piece together to define the Fourier transform on \(\cA\).

\begin{definition}
We define a \emph{section} of \(\cA\) to be a function \(s : \dual{G} \to \cA\) such that \(s(\bx) \in \cA_\bx\) for all \(\bx \in \dual{G}\). In particular, we define the \emph{Fourier transform} of \(a \in \cA\) to be the section \(\hat{a}\) of \(\cA\) constructed as follows:
\begin{equation}
	\forall \bx \in \dual{G}, \quad \widehat{a}(\bx) \coloneqq F_\bx(a) = \int_G \e{-\ip{\bx}{t}}\alpha_t(a)\dif t.
\end{equation}
\end{definition}

\begin{remark}
A section of \(\cA\) can be interpreted as a section of the ``smooth Fell bundle''  \(\set{\cA_\bx}_{\bx \in \dual{G}} \to \dual{G}\) over \(\dual{G}\).
\end{remark}

Finally, let us identify the range of the Fourier transform on \(\cA\), which will allow us to perform all computations in \(\cA\) in terms of rapidly decaying Fourier series. Let us call a section \(s\) of \(\cA\) \emph{rapidly decaying} if it satisfies
\begin{equation}\label{seminorms}
 \forall k, \; l \in \bN, \quad \norm{s}_{k,l} \coloneqq \sup_{\bx \in \dual{G}} (1+4\pi^2\abs{\bx}^2)^k \norm{s(\bx)}_l < +\infty.
\end{equation}
The  \(\bC\)-vector space \(\cS(\dual{G};\cA)\) of all rapidly decaying sections of \(\cA\) defines a unital \Star-algebra via the convolution product and \Star-operation
\[
	\forall s_1, \; s_2 \in \cS(\dual{G}; \cA), \; \forall \bx \in \dual{G}, \quad s_1 \star s_2(\bx) \coloneqq \sum_{\by \in \dual{G}} s_1(\bx-\by)s_2(\by), \quad s^\ast(\bx) \coloneqq s(-\bx)^\ast,
\]
where absolute convergence of each \(s_1s_2(\bx)\) is guaranteed by the rapid decay of \(s_1\) and \(s_2\). Since the isotypic subspaces of \(\cA\) are closed, the standard proof that \(\cS(\dual{G})\) is a Fr\'echet pre-\Cstar-algebra (cf.\ \cite{Jol}), \emph{mutatis mutandis}, together with the fact that the \(\cA_\bx\) are closed, implies that \(\cS(\dual{G})\), in fact, defines a Fr\'echet \Star-algebra.

\begin{proposition}
The unital \Star-algebra \(\cS(\dual{G};\cA)\), endowed with the seminorms \(\norm{}_{k,l}\) defined by Equation~\ref{seminorms} and the \G-action \(\hat{\alpha}\) defined by
\[
 \forall t \in G, \; \forall s \in \cS(\dual{G};\cA), \; \forall \bx \in \dual{G}, \quad \hat{\alpha}_t(s)(\bx) \coloneqq \e{\ip{\bx}{t}} s(\bx).
\]
is a Fr\'echet \Gstar-algebra.
\end{proposition}

\begin{remark}
If the Fr\'echet \Gstar-algebra \(\cA\) is viewed as a ``smooth Fell bundle,'' then \(\cS(\dual{G};\cA)\) can be viewed as its ``smooth cross-sectional algebra'' (cf.\ \cite[\S \textsc{viii}.5]{FD}).
\end{remark}

\begin{example}\label{canonicalexample1}
As we have already seen, \(C^\infty(G)_\bx = \bC U_\bx\) for each \(\bx \in \dual{G}\), where \(U_\bx \coloneqq e \circ \bx : G \to U(1)\) is unitary. Hence, the map
\((a_\bx U_\bx)_{\bx \in \dual{G}} \mapsto (a_\bx)_{\bx \in \dual{G}}\) defines a canonical isomorphism \(\cS(\dual{G};C^\infty(G)) \iso \cS(\dual{G})\). Thus, \(\cS(\dual{G};\cA)\), in general, can be interpreted as the rapid decay algebra of \(\dual{G}\) with coefficients in \(\cA\).
\end{example}

We now show that \(\cS(\dual{G};\cA)\) is the range of the Fourier transform on \(\cA\), which, therefore, defines a \(G\)-equivariant \Star-isomorphism \(\cA^\infty \to \cS(\dual{G};\cA)\).

\begin{proposition}\label{fourier}
Let \(\cA\) be a Fr\'echet \Gstar-algebra. The \emph{Fourier transform} \(a \mapsto \hat{a}\) yields a \G-equivariant topological \Star-isomorphism \(F : \cA^\infty \to \cS(\dual{G};\cA)\). Moreover, 
\(
	\inv{F} : s \mapsto \check{s} \coloneqq \sum_{\bx \in \dual{G}} s(\bx)
\)
with absolute convergence in \(\cA\).
\end{proposition}

\begin{example}
By Examples~\ref{canonicalexample0} and \ref{canonicalexample1}, the composition of the Fourier transform \(F : C^\infty(G) \iso \cS(\dual{G};C^\infty(G))\) with the canonical isomorphism \(\cS(\dual{G};C^\infty(G)) \iso \cS(\dual{G})\) is precisely the usual Fourier transform \(C^\infty(G) \iso \cS(\dual{G})\).
\end{example}

\begin{proof}[Proof of Proposition~\ref{fourier}]
On the one hand, since~\cite[Proof of Lemma 1]{Ya}
\[
\forall a \in \cA, \; \bx \in \dual{G} \quad F(\Delta_G a)(\bx) = 4\pi^2 \abs{\bx}^2 \widehat{a}(\bx),
\]
and since each \(F_\bx\) is contractive, it follows that
\[
	\forall a \in \cA, \; \forall k,l \in \bN, \quad \sup_{\bx \in \dual{G}} (1+4\pi^2 \abs{\bx}^2)^k \norm{\widehat{a}(\bx)}_l   \leq \norm{(1+\Delta_G)^ka}_l = \norm{a}_{k,l}
\]
and hence that \(F: \cA \to \cS(\dual{G};\cA)\) is well-defined and continuous; it is easy to check, then, that \(F\) is a \(G\)-equivariant \Star-homomorphism.

On the other hand, if \(K \coloneqq \lceil \tfrac{N}{2} \rceil\), then for all \(s \in \cS(\dual{G};\cA)\), 
\[
	\forall l \in \bN, \quad \sum_{\bx \in \dual{G}} \norm{s(\bx)}_l \leq \sum_{\bx \in \dual{G}} (1+4\pi^2 \abs{\bx}^2)^{-K} \sup_{\by \in \dual{G}} (1+4\pi^2\abs{\by}^2)^K \norm{s(\by)}_l< \infty,
\]
so that \(\check{s}\) is absolutely convergent in \(\cA\), and hence \(s \mapsto \check{s}\) yields a well-defined map \(E : \cS(\dual{G};\cA) \to \cA\); it is easy to check, then, that \(E\) is a \(G\)-equivariant \Star-homomorphism such that \(E \circ F = \id\) and \(F \circ E = \id\), so that by the bounded inverse theorem for Fr\'echet spaces, \(E = F^{-1}\) is continuous and \(F\), indeed, defines a \G-equivariant topological \Star-isomorphism.
\end{proof}

As a consequence of this machinery, we can freely write
\begin{equation}
	\forall a, \; b \in \cA, \quad a b = \sum_{\bx,\by\in\dual{G}} \hat{a}(\bx-\by)\hat{b}(\by), \quad a^\ast = \sum_{\bx \in \dual{G}} \hat{a}(-\bx)^\ast,
\end{equation}
with absolute convergence in \(\cA\).

Finally, let us sketch the analogous Peter--Weyl theory for \G-equivariant Fr\'echet modules. Once more, let us begin by fixing terminology and notation.

\begin{definition}
Let \(\cA\) be a Fr\'echet \Star-algebra. A (right) \emph{Fr\'echet \(\cA\)-module} is a Fr\'echet space \(\cE\) together with the algebraic structure of right \(\cA\)-module, such that the right \(\cA\)-module structure \(\cE \times \cA \to \cE\), \((\xi,a) \mapsto \xi \lhd a\)
is continuous. In particular, a \emph{Hermitian Fr\'echet \(\cA\)-module} is a Fr\'echet \(\cA\)-module \(\cE\) endowed with a \emph{Hermitian metric}, viz, a continuous map \(\hp{}{} : \cE \times \cE \to \cA\) satisfying the following properties:
\begin{enumerate}
	\item For all \(\xi, \eta_1, \eta_2 \in \cE\), \(a_1, a_2 \in \cA\), \(\hp{\xi}{a_1 \lhd \eta_1 + a_2 \lhd \eta_2} =  \hp{\xi}{\eta_1}a_1 + \hp{\xi}{\eta_2}a_2\).
	\item For all \(\xi\), \(\eta \in \cE\), \(\hp{\eta}{\xi} = \hp{\xi}{\eta}^\ast\).
	\item For all \(\xi \in \cE\), \(\hp{\xi}{\xi} \geq 0\); moreover, \(\hp{\xi}{\xi} = 0\) if and only if \(\xi = 0\).
\end{enumerate}
\end{definition}

\begin{definition}
Let \((\cA,\alpha)\) be a Fr\'echet \Gstar-algebra. A \emph{Fr\'echet \GA-module} is a Fr\'echet \(\cA\)-module \(\cE\), together with a smooth, isometric action \(U : G \to B(\cE)\) of \(G\) on \(\cE\) by continuous \(\bC\)-linear operators, such that
\[
	\forall \xi \in \cE, \; \forall a \in \cA, \; \forall t \in G, \quad U_t(\xi \lhd a) = (U_t \xi) \lhd \alpha_t(a).
\]
In particular, a \emph{Hermitian Fr\'echet \GA-module} is Fr\'echet \GA-module \(\cE\) together with a Hermitian metric \(\hp{}{} : \cE \times \cE \to \cA\), such that
\[
	\forall \xi, \; \eta \in \cE, \; \forall t \in G, \quad \hp{U_t\xi}{U_t\eta} = \alpha_t\left(\hp{\xi}{\eta}\right).
\]
\end{definition}

\begin{example}
Let \(X\) be a compact \(G\)-manifold and let \(E \to X\) be a \G-equivariant Hermitian vector bundle. Then the Fr\'echet space \(C^\infty(X,E)\) of smooth global sections of \(E\) defines a Hermitian Fr\'echet \G-\(C^\infty(X)\)-module.
\end{example}

Again, we are free to enrich the Fr\'echet topology on a Fr\'echet \GA-module \(\cE\) with the additional seminorms defined by Equation~\ref{seminorms}, \emph{mutatis mutandis}, so that \(\cE\) remains a Fr\'echet \G-\(\cA^\infty\)-module.

Now, just as in the case of Fr\'echet \Gstar-algebras, a (Hermitian) Fr\'echet \GA-module \(\cE\) admits a Peter--Weyl decomposition into isotypic components, so that one can define a (Hermitian) Fr\'echet \G-\(\cS(\dual{G};\cA)\)-module \(\cS(\dual{G},\cE)\) by
\[
	\cS(\dual{G},\cE) \coloneqq \set*{s : \dual{G} \to \cE \given \begin{cases} \forall \bx \in \dual{G}, \; s(\bx) \in \cE_\bx,\\ \forall j, \; k \in \bN, \; \sup_{\bx \in \dual{G}} (1+4\pi^2\abs{\bx}^2)^j \norm{s(\bx)}_k < +\infty \end{cases}},
\]
endowed with the seminorms, right \(\cS(\dual{G};\cA)\)-module structure, and \(G\)-action
\begin{gather*}
	\forall j,\; k \in \bN, \; \forall s \in \cS(\dual{G},\cE), \quad \norm{s}_{j,k} \coloneqq  \sup_{\bx \in \dual{G}} (1+4\pi^2\abs{\bx}^2)^j \norm{s(\bx)}_k,\\
	\forall s \in \cS(\dual{G};\cE), \; \forall \hat{a} \in \cS(\dual{G};\cA), \; \forall \bx \in \dual{G}, \quad \left(s \lhd \hat{a}\right)(\bx) \coloneqq \sum_{\by \in \dual{G}} s(\bx-\by) \lhd \hat{a}(\by),\\
	\forall t \in G, \; \forall s \in \cS(\dual{G};\cE), \; \forall \bx \in \dual{G}, \quad \left(\hat{U}_t s\right)(\bx) \coloneqq \e{\ip{\bx}{t}} U_t s(\bx),
\end{gather*}
and, if \(\cE\) is Hermitian, the Hermitian metric
\[
	\forall s_1, \; s_2 \in \cS(\dual{G};\cE), \quad \hp{s_1}{s_2} \coloneqq \sum_{\bx, \by \in \dual{G}} \hp{s_1(-\bx+\by)}{s_2(\by)};
\]
as a result, one can again construct a Fourier transform \(F : \cE \to \cS(\dual{G};\cE)\), \(\xi \mapsto \hat{\xi}\), which defines a \G-equivariant topological isomorphism satisfying \(F(\xi \lhd a) = F(\xi) \lhd F(a)\) for all \(\xi \in \cA\) and \(a \in \cA\), and, if \(\cE\) is Hermitian, \(\hp{\hat{\xi}}{\hat{\eta}} = \hp{\xi}{\eta}\) for all \(\xi\), \(\eta \in \cE\), with inverse Fourier transform \(s \mapsto \check{s} = \sum_{\bx \in \dual{G}} s(\bx)\). As a result, we can freely write
\begin{gather}
	\forall \xi \in \cE, \; \forall a \in \cA, \quad \xi \lhd a = \sum_{\bx,\by\in\dual{G}} \hat{\xi}(\bx+\by)\lhd\hat{a}(\by),\\
	\forall \xi, \; \eta \in \cE, \quad \hp{\xi}{\eta} = \sum_{\bx,\by\in\dual{G}} \hp{\hat{\xi}(-\bx+\by)}{\hat{\eta}(\by)},
\end{gather}
with absolute convergence in \(\cE\) and \(\cA\), respectively. Analogous results hold for left Fr\'echet \GA-modules. 

\subsection{Deformation of Fr\'echet \texorpdfstring{\(G\)-\(\ast\)}{G-*}-algebras}

We now turn to strict deformation quantisation of Fr\'echet \Gstar-algebras, expanding on Rieffel's account~\cite[pp.\ 19--22]{Rieffel} in the spirit of Abadie and Exel's Fell bundle-theoretic approach~\cite{AbEx}.

Fix a Fr\'echet \Gstar-alge\-bra \((\cA,\alpha)\). Let \(\Theta \in Z^2(\dual{G},\bT)\). Recall that the \emph{twisted rapid decay algebra} \(\cS(\dual{G},\Theta)\) is defined to be \(\cS(\dual{G})\) as Fr\'echet space, endowed with the same \(G\)-action but with the twisted convolution product and \Star-operation
\begin{gather*}
	\forall s_1, \; s_2 \in \cS(\dual{G}), \; \forall \bx \in \dual{G}, \quad s_1 \star_\Theta s_2(\bx) \coloneqq \sum_{\by \in \dual{G}} \e{-\Theta(\bx-\by,\by)}s_1(\bx-\by)s_2(\by),\\
	\forall s \in \cS(\dual{G}), \; \forall \bx \in \dual{G}, \quad s_1^{\ast_\Theta}(\bx) \coloneqq \e{-\Theta(-\bx,-\bx)} \overline{s_1(-\bx)};
\end{gather*}
one can then check that \(\cS(\dual{G},\Theta)\) remains a Fr\'echet pre-\GCstar-algebra (see Chatterji's appendix to~\cite{Mathai}). By analogy, one can define a unital \Star-algebra \(\cS(\dual{G},\Theta;\cA)\) by endowing the Fr\'echet space space \(\cS(\dual{G};\cA)\) with the same \(G\)-action but with the deformed convolution product and \Star-operation
\begin{gather*}
	\forall s_1, \; s_2 \in \cS(\dual{G},\Theta;\cA), \; \forall \bx \in \dual{G}, \enskip s_1 \star_\Theta s_2(\bx) \coloneqq  \sum_{\by \in \dual{G}} \e{-\Theta(\bx-\by,\by)}s_1(\bx-\by)s_2(\by),\\
	\forall s \in \cS(\dual{G},\Theta;\cA), \forall \bx \in \dual{G}, \quad s^{\ast_\Theta}(\bx) \coloneqq \e{-\Theta(-\bx,-\bx)}s(-\bx)^\ast,
\end{gather*}
where absolute convergence of each \((s_1 \star_\Theta s_2)(\bx)\) is guaranteed by the rapid decay of \(s_1\) and \(s_2\). The proof that \(\cS(\dual{G},\Theta)\) is a Fr\'echet \Star-algebra, \emph{mutatis mutandis}, together with the fact that the isotypic subspaces \(\cA_\bx\) are closed, shows that \(\cS(\dual{G},\Theta;\cA)\) is a Fr\'echet \Gstar-algebra. Moreover, since \(\cS(\dual{G},\Theta;\cA) = \cS(\dual{G};\cA)\) as \(G\)-equivariant Fr\'echet spaces, the Fourier transform on \(\cA\) still defines a \Star-preserving \G-equivariant topological isomorphism \(F_\Theta  : \cA^\infty \to \cS(\dual{G},\Theta;\cA)\) of Fr\'echet spaces. Thus, we can deform \(\cA\) simply by pulling back the deformation of \(\cS(\dual{G};\cA)\) to \(\cA\).

\begin{definition}
Let \((\cA,\alpha)\) be a Fr\'echet \Gstar-algebra. The \emph{strict deformation quantisation} of \((\cA,\alpha)\) by \(\Theta \in Z^2(\dual{G},\bT)\) is the Fr\'echet \Gstar-algebra \((\cA_\Theta,\alpha)\), where \(\cA_\Theta = \cA\) as Fr\'echet \(G\)-spaces with the deformed multiplication and \Star-operation
\begin{gather}
	\forall a, \; b \in \cA, \quad a \star_\Theta b \coloneqq \inv{F}(\hat{a} \star_\Theta \hat{b}) = \sum_{\bx,\by\in\dual{G}} \e{\Theta(\bx-\by,\by)} \hat{a}(\bx-\by)\hat{b}(\by)\\
	\forall a \in \cA, \quad a^{\ast_\Theta} \coloneqq \inv{F}(\hat{a}^{\ast_\Theta}) = \sum_{\bx \in \dual{G}} \e{\Theta(-\bx,-\bx)}\hat{a}(-\bx)^\ast.
\end{gather}
\end{definition}

\begin{remark}
If \(\Theta \in A(\dual{G})\) is an alternating bicharacter, then the deformed \Star-operation is identical to the undeformed \Star-operation.
\end{remark}

The mathematical-physical justification for this construction and its name is given by the following result of Rieffel's, which we only state for the relevant case of compact Abelian Lie groups acting on compact manifolds.

\begin{theorem}[Rieffel~\cite{Rieffel}]
Let \(X\) be a compact \G-manifold and let \(P \in \wedge^2 \fg\). On the one hand, let \(\{\cdot,\cdot\}\) denote the Poisson bracket on \(C^\infty(X)\) induced by the Poisson bi-vector field \((\alpha_\ast \otimes \alpha_\ast)(P)\), and on the other hand, for \(\hbar \in \bR\), let \(\Theta_\hbar \in A(\dual{G})\) be defined by \(\Theta_\hbar(\bx,\by) \coloneqq \exp( \tfrac{1}{2}\hbar\ip{P}{\bx_\ast \otimes \by_\ast})\) for \(\bx\), \(\by \in \dual{G}\). Finally, for \(\hbar \in \bR\), let \(C(X_\hbar)\) denote the \Cstar-completion of \(C^\infty(X)_{\Theta_{\hbar}}\) with multiplication \(\star_\hbar\) and \Cstar-norm \(\norm{}_\hbar\). Then \(\set{C(X_\hbar)}_{\hbar \in \bR}\) is a continuous field of \Cstar-algebras such that
\[
	\forall f, \; g \in C^\infty(X), \quad \lim_{\hbar \to 0} \norm*{\frac{1}{i\hbar}(f \star_\hbar g - g \star_\hbar f) - \{f,g\}}_\hbar = 0.
\]
\end{theorem}

More generally, if \(X\) is a compact \G-manifold and \(\Theta \in Z^2(\dual{G},\bT)\), then one is free to form the deformed algebra \(C^\infty(X)_\Theta\). The most familiar example of this construction is surely the operator algebraist's noncommutative torus.

\begin{example}\label{nctorusex}
Let \(\theta \in \wedge^2 \mathbb{R}^N\) and define \(\Theta \in A(\bZ^N)\) by setting \(\Theta(\bx,\by) \coloneqq \exp(\tfrac{1}{2}\theta(\bx_\ast,\by_\ast))
\) for \(\bx\), \(\by \in \bZ^N\). Then \(C^\infty(\bT^N_\theta) \coloneqq C^\infty(\bT^N)_\Theta\) is the noncommutative \(N\)-torus with deformation parameter \(\theta\), satisfying the commutation relations
\[
	\forall \bx, \; \by \in \bZ^N, \quad U_\by \star_\Theta U_\bx = \exp(2\pi i \theta(\bx_\ast,\by_\ast)) U_\bx \star_\Theta U_\by.
\]
\end{example}

Finally, the basic properties of strict deformation quantisation can be summarised as follows; in our context, the rapid decay of Fourier coefficients renders all convergence issues moot, so that the proving these properties reduces to a mechanical, purely algebraic fact-check.

\begin{proposition}[{Brain--Landi--Van Suijlekom~\cite[\S 2]{BLVS}}]\label{isom}
Let \(\cat{Alg}_G\) denote the category of Fr\'echet \Gstar-algebras and \G-equivariant continuous \Star-homomorphisms. For every \(\Theta \in Z^2(\dual{G},\bT)\), the assignment
\begin{gather*}
	\operatorname{Obj}(\cat{Alg}_G) \ni (\cA,\alpha) \mapsto (\cA_\Theta,\alpha),\\
	\operatorname{Mor}(\cat{Alg}_G) \ni \left(f : (\cA,\alpha) \to (\cB,\beta)\right) \mapsto \left(f : (\cA_\Theta,\alpha) \to (\cB_\Theta,\beta)\right),
\end{gather*}
defines an isomorphism of categories \(\cat{R}_\Theta : \cat{Alg}_G \iso \cat{Alg}_G\). Moreover, the functors \(\cat{R}_\Theta\) satisfy the following properties.
\begin{enumerate}
	\item Let \(\id_{\cat{Alg}_G}\) denote the identity functor on \(\cat{Alg}_G\). Then \(\cat{R}_0 = \id_{\cat{Alg}_G}\).
	\item For any \(\Theta\), \(\p{\Theta} \in Z^2(\dual{G},\bT)\), \(\cat{R}_\Theta \circ \cat{R}_\p{\Theta} = \cat{R}_{\Theta + \p{\Theta}}\).
%	\item Let \(\cat{Op} : \cat{Alg}_G \to \cat{Alg}_G\) denote the opposite algebra functor \((\cA,\alpha) \mapsto (\cA^\op,\alpha)\). Then for any \(\Theta \in Z^2(\dual{G},\bT)\), \(\cat{Op} \circ \cat{R}_\Theta = \cat{R}_{\Theta^t} \circ \cat{Op}\).
\end{enumerate}
Thus, \(\Theta \mapsto \cat{R}_\Theta\) defines a strict action of the group \(Z^2(\dual{G},\bT)\) on the category \(\cat{Alg}_G\).
%such that the involution \(\Theta \mapsto \Theta^t\) of \(Z^2(\dual{G},\bT)\) is implemented by the involutive endofunctor \(\cat{Op}\) of \(\cat{Alg}_G\).
\end{proposition}

For a overview of strict deformation quantisation, in full generality, from a thoroughly functorial perspective, consult Brain, Landi, and Van Suijlekom~\cite[\S 2]{BLVS}.

\subsection{Deformation of  \texorpdfstring{\G}{G}-equivariant Fr\'echet modules}\label{moduledeformsec}

For the remainder of this subsection, let \((\cA,\alpha)\) be a fixed Fr\'echet \Gstar-algebra for \(G\) a fixed compact Abelian Lie group. We now discuss the deformation of Fr\'echet \GA-modules to Fr\'echet \G-\(\cA_\Theta\)-modules for \(\Theta \in Z^2(\dual{G},\bT)\) with a particular emphasis on the finitely generated projective case. In what follows, all modules are right modules unless stated otherwise.

Let \(\cE\) be a (Hermitian) Fr\'echet \GA-module and let \(\Theta \in Z^2(\dual{G},\bT)\). Recall that the Fourier transform \(F : \cE \to \cS(\dual{G};\cE)\) defines a (Hermitian metric-preserving) \(G\)-equivariant bicontinuous linear map intertwining the \(\cA\)-module structure on \(\cE\) with the \(\cS(\dual{G};\cA)\)-module structure on \(\cS(\dual{G};\cE)\). Now, it is easy to deform the Fr\'echet \G-\(\cS(\dual{G};\cA)\)-module \(\cS(\dual{G};\cE)\) into a Fr\'echet \G-\(\cS(\dual{G},\Theta;\cA)\)-module \(\cS(\dual{G},\Theta;\cE)\) by defining \(\cS(\dual{G},\Theta;\cE) \coloneqq \cS(\dual{G};\cE)\) as a \G-equivariant Fr\'echet space, except endowed with the right \(\cS(\dual{G},\Theta;\cA)\)-module structure
\begin{multline*}
	\forall s \in \cS(\dual{G},\Theta;\cE), \; \forall \hat{a} \in \cS(\dual{G},\Theta;\cA), \; \forall \bx \in \dual{G}, \\ \left(s \lhd_\Theta \hat{a}\right)(\bx) \coloneqq \sum_{\bx,\by\in\dual{G}} \e{-\Theta(\bx-\by,\by)} s(\bx-\by) \lhd \hat{a}(\by),
\end{multline*}
and, if \(\cE\) is Hermitian, the Hermitian metric
\[
	\forall s_1, \; s_2 \in \cS(\dual{G},\Theta;\cE), \quad \hp{s_1}{s_2}_\Theta \coloneqq \sum_{\bx,\by\in\dual{G}}  \e{-\Theta(\bx-\by,\bx)}  \hp{s_1(-\bx+\by)}{s_1(\by)}.
\]
Thus, we can deform \(\cE\) to a \(\cA_\Theta\)-module simply by pulling back the deformation of \(\cS(\dual{G},\Theta;\cE)\) to \(\cE\).

\begin{definition}
Let \((\cE,U)\) be a Fr\'echet \GA-module. The \emph{deformation} of \(\cE\) by \(\Theta \in Z^2(\dual{G},\bT)\) is the Fr\'echet \G-\(\cA_\Theta\)-module \((\cE_\Theta,U)\), where \(\cE_\Theta \coloneqq \cE\) as \G-equivariant Fr\'echet spaces, endowed with the right \(\cA_\Theta\)-module structure
\begin{equation}
	\forall \xi \in \cE, \; \forall a \in \cA_\Theta, \quad \xi \lhd_\Theta a \coloneqq \inv{F}(\hat{\xi} \lhd_\Theta \hat{a}) = \sum_{\bx,\by\in\dual{G}} \e{-\Theta(\bx-\by,\by)} \hat{\xi}(\bx-\by) \lhd \hat{a}(\by);
\end{equation}
moreover, if \(\cE\) is Hermitian, then we can endow \(\cE_\Theta\) with the \G-equivariant \(\cA_\Theta\)-valued Hermitian metric
\begin{equation}\label{hermitian}
	\forall \xi, \; \eta \in \cE_\Theta, \quad \hp{\xi}{\eta}_\Theta \coloneqq \hp{\hat{\xi}}{\hat{\eta}}_\Theta = \sum_{\bx,\by\in\dual{G}}  \e{-\Theta(\bx-\by,\bx)}  \hp{\hat{\xi}(-\bx+\by)}{\hat{\eta}(\by)}.
\end{equation}
\end{definition}

\begin{remark}
The analogous construction can just as easily be made in the case of \emph{left} (Hermitian) Fr\'echet \GA-modules.
\end{remark}

The basic properties of deformations of Fr\'echet \GA-modules can be summarised as follows; again, in our context, the rapid decay of Fourier coefficients reduces the proof of these properties to a mechanical, purely algebraic fact-check.

\begin{proposition}[{Brain--Landi--Van Suijlekom~\cite[\S 2]{BLVS}}]\label{modfunctor}
Let \((\cA,\alpha)\) be a Fr\'echet \Gstar-algebra and let \(\cat{Mod}_G(\cA)\) denote the category of right Fr\'echet \GA-modules and \G-equivariant continuous \(\cA\)-module homomorphisms. For every \(\Theta \in Z^2(\dual{G},\bT)\), the assignment
\begin{gather*}
	\operatorname{Obj}(\cat{Mod}_G(\cA)) \ni (\cE,U) \mapsto (\cE_\Theta,U)\\
	\operatorname{Mor}(\cat{Mod}_G(\cA)) \ni (f : (\cE,U) \to (\cF,V)) \mapsto (f : (\cE_\Theta,U) \mapsto (\cF_\Theta,V)),
\end{gather*}
defines an isomorphism of categories \(\cat{R}_{\cA,\Theta} : \cat{Mod}_G(\cA) \iso \cat{Mod}_G(\cA_\Theta)\). Moreover, the functors \(\cat{R}_{\cA,\Theta}\) satisfy the following properties.
\begin{enumerate}\setcounter{enumi}{-1}
	\item Let \(\cA\) and \(\cB\) be Fr\'echet \Gstar-algebras, let \(\phi : \cA \to \cB\) be a \G-equivariant continuous \Star-homomorphism, and let \(\phi^\ast : \cat{Mod}_G(\cB) \to \cat{Mod}_G(\cA)\) denote the pullback functor. For any \(\Theta \in Z^2(\dual{G},\bT)\), \(\phi^\ast \circ \cat{R}_{\cB,\Theta} = \cat{R}_{\cA,\Theta} \circ \phi^\ast\).
	\item Let \(\id_{\cat{Mod}_G(\cA)}\) denote the identity on \(\cat{Mod}_G(\cA)\). Then \(\cat{R}_{\cA,0} = \id_{\cat{Mod}_G(\cA)}\).
	\item For any \(\Theta\), \(\p{\Theta} \in Z^2(\dual{G},\bT)\), \(\cat{R}_{\cA_\Theta,\p{\Theta}} \circ \cat{R}_{\cA,\Theta} = \cat{R}_{\cA,\Theta+\p{\Theta}}\).
%	\item Let \(\cat{Op}^X_\cA : \cat{Mod}^X_G(\cA) \to \cat{Op}^{\op(X)}_{\cA^\op}\) denote the opposite module functor, where \(\op(L) = R\) and \(\op(R) = L\), so that, for instance, \(\cat{Op}^R_\cA : \cE_{\cA} \mapsto {}_{\cA^\op}\cE^\op\). Then for any \(\Theta \in Z^2(\dual{G},\bT)\), \(\cat{Op}^X_{\cA_\Theta} \circ \cat{R}^X_{\cA,\Theta} = \cat{R}^{\op{X}}_{\cA^\op,\Theta^t} \circ \cat{Op}^X_\cA\).
\end{enumerate}
Analogous results hold for left Fr\'echet \GA-modules and \G-equivariant continuous \(\cA\)-module homomorphisms and for left or right Hermitian Fr\'echet \GA-modules and \G-equivariant unitary isomorphisms, viz, bicontinuous unitary \(\cA\)-module homomorphisms.
\end{proposition}

Again, for an overview of the deformation of equivariant Fr\'echet modules, in full generality, from a thoroughly functorial perspective, consult Brain, Landi, and Van Suijlekom~\cite[\S 2]{BLVS}. 

Let us now focus our attention on finitely generated projective modules. From now on, we assume that \((\cA,\alpha)\) is a unital Fr\'echet pre-\GCstar-algebra with \G-equivar\-i\-ant \Cstar-completion the unital \GCstar-algebra \((A,\alpha)\); this will always be the case in the context of Connes--Landi deformation. In addition, we shall only consider deformation along bicharacters; in light of Theorem~\ref{kleppner} and our later discussion in \S~\ref{cohomology}, we can do so without any loss of generality.

Now, recall that a finitely generated projective right \(\cA\)-module \(\cE\) admits a canonical Fr\'echet topology, which is induced by any inclusion of \(\cE\) as a complementable subspace of \(\cA^N\) for some \(N \in \bN\). One can therefore make the following definition.

\begin{definition}[{Julg~\cite[Definition 2.1]{Julg}}]\label{fgpmoddef}
We define a \emph{f.g.p.\ \GA-module} to be a finitely generated projective right \(\cA\)-module \(\cE\) together with a strongly continuous, isometric \G-action \(U : G \to B(\cE)\), such that
\begin{equation}
	\forall a \in \cA, \; \forall \xi \in \cE, \; \forall t \in G, \quad U_t(\xi \lhd a) = U_t(\xi) \lhd \alpha_t(a).
\end{equation}
In particular, we define a \emph{Hermitian f.g.p.\ \GA-module} to be a f.g.p.\ \GA-module together with a Hermitian metric \( \hp{}{} : \cE \times \cE \to \cA\) satisfying
\begin{equation}
	\forall \xi,\; \eta \in \cE, \; \forall t \in G, \quad \hp{U_t\xi}{U_t\eta} = \alpha_t(\hp{\xi}{\eta}).
\end{equation}
\end{definition}

\begin{remark}
Observe that an f.g.p.\ \GA-module is, in particular, a Fr\'echet \(\cA\)-module, but is not, \emph{a priori}, a Fr\'echet \GA-module.
\end{remark}

Under our additional assumption on \(\cA\), one can prove the following crucial characterisation of f.g.p.\ \G-\(\cA\)-modules as \(G\)-equivariant direct summands of finitely generated free \G-\(\cA\)-modules, which will immediately imply that f.g.p.\ \GA-modules are, in particular, Fr\'echet \GA-modules.

\begin{lemma}[{Julg~\cite[Lemma 4.2]{Julg}, cf.\ Segal~\cite[Proposition 2.4]{Segal}}]\label{fgpmod}
Let \(\cE\) be an f.g.p.\ \GA-module. There exists a \G-equivariant isomorphism \(\psi : \cE \iso e(V \otimes \cA)\), where \(V\) is a finite-dimensional representation of \(G\) and \(e \in (B(V) \otimes \cA)^G\) is an idempotent. If, in addition, \(\cE\) is Hermitian, then one can take the representation \(V\) to be unitary, the idempotent \(E\) to be an orthogonal projection, and the isomorphism \(\psi : \cE \iso e(V \otimes \cA)\) to be unitary, in the sense that
\[
	\forall v, \; \p{v} \in V, \; \forall a, \p{a} \in \cA, \quad \hp{\inv{\psi}(v\otimes a)}{\inv{\psi}(\p{v}\otimes\p{a})} = \ip{v}{\p{v}}a^\ast \p{a}.
\]
\end{lemma}

\begin{proof}
The proof of the \Cstar-algebraic case \cite[\S 11.2]{Blackadar} \cite[\S 2.2]{Phillips} carries over mostly unchanged; we shall sketch the proof with an emphasis on necessary changes. 

To begin, let \(\iota : \cE \inj \cA^N\) be an inclusion of \(\cE\) as a complementable subspace of \(\cA^N\) for some \(N \in \bN\), thereby defining the Fr\'echet topology on  \(\cE\), and let \(\set{\xi_1,\dotsc,\xi_N}\) be a set of generators for \(E\) \emph{qua} \(\cA\)-module.
Let \(E \coloneqq \cE \otimes_\cA A\), topologised as a complementable subspace of \(A^N\) via the \G-equivariant extension \(\iota \otimes \id : E \coloneqq \cE \otimes_\cA A \inj \cA^N \otimes_\cA A = A^N\) of \(\iota\) from \(\cE\) to its \(G\)-equivariant closure \(E\); in particular, \(E\) defines a \G-equivariant Hilbert \(A\)-module, which, again, is generated by \(\set{\xi_1,\dotsc,\xi_N} \in \cE \subset E\).

Now, use the convolution action of \(L^1(G)\) on the \(G\)-invariant dense subspace \(\cE\) on \(E\) find \(\eta_1,\dotsc,\eta_N \in \cE\) with \(\norm{\xi_k - \eta_k} < N^{-1}\) in \(E\) for each \(k\), such that
\(
	V \coloneqq \Span_\bC \cup_{k=1}^N L^1(G)\eta_k
\)
is finite-dimensional and \G-invariant. It suffices to show that \(\set{\eta_1,\dotsc,\eta_N}\) also generates \(\cE\) as an \(\cA\)-module. Indeed, if that is the case, then we can define a \G-equivariant surjective \(\cA\)-module homomorphism \(\phi : V \otimes \cA \to \cE\) by setting
\(
	 \phi(v \otimes a) \coloneqq v \lhd a
\)
for \(v \in V\) and \(a \in \cA\). Since \(\cE\) is projective, \(\phi\) admits a right inverse \(\psi_0 : \cE \to V \otimes \cA\), which we can average over \(G\) to obtain a \G-equivariant right inverse \(\psi : \cE \to V \otimes \cA\); at last, it follows that \(\psi : \cE \iso e(V \otimes \cA)\) defines a \G-equivariant topological isomorphism of \(\cA\)-modules with inverse \(\inv{\psi} = \rest{\phi}{e(V \otimes \cA)}\), where \(e \coloneqq \phi \circ \psi \in B(V) \otimes \cA\) is the desired \(G\)-invariant idempotent.

Let us conclude the above argument by showing that \(\set{\eta_1,\dotsc,\eta_N}\) also generates \(\cE\). Define homomorphisms \(T\), \(\p{T} \in \Hom_\cA(\cA^N,\cE) \subset \Hom_{A}(A^N,E)\) by setting
\[
	\forall (a_1,\dotsc,a_N) \in \cA^N, \quad T(a_1,\dotsc,a_N) \coloneqq \sum_{k=1}^N \xi_k \lhd a_k, \quad \p{T}(a_1,\dots,a_N) \coloneqq \sum_{k=1}^N \eta_k \lhd a_k,
\]
where, by construction, \(T\) defines a continuous left inverse for \(\iota\). It clearly suffices to show that \(\p{T}\) is surjective. Now, by our assumption on \(\set{\eta_1,\dotsc,\eta_N}\), it follows that \(\norm{T-\p{T}}_{\Hom_A(A^N,E)} < 1\), so that
\[
	\norm{((\id - \iota \circ T) + \iota \circ \p{T})-\id}_{M_N(A)} = \norm{\iota \circ T - \iota \circ \p{T}}_{M_N(A)} \leq \norm{T-\p{T}}_{\Hom_A(A^N,E)} < 1,
\]
and hence, \((\id - \iota \circ T) + \iota \circ \p{T}\) is invertible in \(M_N(A)\); as a result, since \(M_N(\cA) \subset M_N(A)\) is closed under the holomorphic functional calculus, \((\id - \iota \circ T) + \iota \circ \p{T}\) is actually invertible in \(M_N(\cA)\). But now, since \(\cA^N = \Ran(\id - \iota \circ T) \oplus \Ran(\iota)\), it follows that \(\Ran(\iota \circ \p{T}) = \Ran(\iota)\), so that by injectivity of \(\iota\), \(\p{T}\) is indeed surjective.

Finally, suppose that \(\cE\) is Hermitian with Hermitian metric \(\hp{}{}\). Choose a \G-invariant inner product \(\ip{}{}\) on \(V\), which induces a Hermitian metric on \(V \otimes A\) by setting
\(
	\hp{v \otimes a}{\p{v}\otimes\p{a}}^\prime \coloneqq \ip{v}{\p{v}} a^\ast \p{a},
\)
for all \(v\), \(\p{v} \in V\) and \(a\), \(\p{a} \in \cA\), and construct a \G-invariant orthogonal projection \(p\), such that \(e(V \otimes \cA) = p(V \otimes \cA)\) \cite[Theorem 3.8]{GBVF}. Endow \(\cE\) with the \G-invariant Hermitian metric \(\psi^\ast \hp{}{}\), and let \(P \in \End_\cA(\cE) \cong p(B(V) \otimes \cA)p\) be the \G-invariant strictly positive operator, such that
\(
	\hp{\xi}{\eta} = \psi^\ast \hp{\xi}{P\eta}^\prime.
\) for all \(\xi\), \(\eta \in \cE\). Snce \(\End_\cA(\cE) \cong p(B(V) \otimes \cA)p\) is closed under the holomorphic functional calculus, \(P\) admits a strictly positive \G-invariant square root \(P^{1/2} \in \End_\cA(\cE)^G\) of \(P\), which, in turn, yields the desired \G-equivariant unitary isomorphism \(\psi \circ P^{-1/2} : \cE \iso p(V \otimes \cA)\).
\end{proof}

\begin{remark}
By Lemma~\ref{fgpmod}, any f.g.p. \GA-module \(\cE \cong e(V \otimes \cA)\) admits a \G-invariant Hermitian metric; choose a \G-invariant inner product on \(V\) and endow \(V \otimes \cA\) with the corresponding \G-invariant Hermitian metric, which one can just pull back to \(\cE\). More importantly, any f.g.p. \GA-module \(\cE \cong e(V \otimes \cA)\) is a Fr\'echet \GA-module, since the \G-action on \(e(V \otimes \cA)\) is manifestly strongly smooth.
\end{remark}

Since an f.g.p.\ \GA-module \(\cE\) is necessarily a Fr\'echet \GA-module, we may deform it by any \(\Theta \in B(\Gamma)\) to obtain a Fr\'echet \G-\(\cA_\Theta\)-module \(\cE_\Theta\); as it turns out, \(\cE_\Theta\) still defines a f.g.p.\ \G-\(\cA_\Theta\)-module, which we show by constructing an explicit \(G\)-invariant projection \(p_\Theta\) onto \(\cE_\Theta\) from a given a \(G\)-invariant projection \(p\) onto \(\cE\).
 
\begin{theorem}[{cf.\ Connes--Dubois-Violette~\cite[Proposition 5]{CDV}, Brain--Landi--Van Suijlekom~\cite[Proposition 3.2.1]{BLVS}}]\label{moduledeform}
Let \(\cE\) be a f.g.p.\ \GA-module; without loss of generality, take \(\cE\) to be Hermitian. Let \(\pi : G \to U(V)\) be any finite-dimensional unitary \G-representation and let \(p \in B(V) \otimes \cA\) be any \G-invariant orthogonal projection such that \(\cE \cong p(V \otimes \cA)\) as Hermitian Fr\'echet \GA-modules; let \((V,\pi) = (\bC v_1,e \circ \bx_1) \oplus \cdots (\bC v_N,e \circ \bx_N)\) be any orthogonal decomposition of \(\pi\) into irreducible unitary representations, where \(\bx_1,\dotsc,\bx_N \in \dual{G}\). For any \(\Theta \in B(\dual{G})\), we have that \(\cE_\Theta \cong p_\Theta (V \otimes \cA_\Theta)\) as Hermitian Fr\'echet \GA-modules, where \(p_\Theta \in B(V) \otimes \cA_\Theta\) is the \G-invariant orthogonal projection defined by
\begin{equation}\label{projformula}
	p_\Theta \coloneqq \sum_{j,k=1}^N  \left( v_j \otimes \ip{v_k}{} \right) \otimes \e{-\Theta(\bx_j,\bx_j-\bx_k)} \hp{v_j \otimes 1}{p(v_k \otimes 1)};
\end{equation}
in particular, \(\cE_\Theta\) defines a f.g.p.\ \G-\(\cA_\Theta\)-module.
\end{theorem}

\begin{remark}
In the case where \(\cA = C^\infty(X)\) and \(\cE = C^\infty(X,E)\) for \(X\) a compact \G-manifold and \(E \to X\) a \G-equivariant Hermitian vector bundle, Connes and Dubois-Violette~\cite[Proposition 5]{CDV} give an indirect proof that \(\cE_\Theta\) is finitely generated projective using Julg's characterization of equivariant \(K\)-theory for \Cstar-algebras~\cite{Julg}, whilst Brain, Landi, and Van Suijlekom~\cite[Proposition 3.2.1]{BLVS} give a completely different constructive proof of this fact that essentially depends on \G-equivariantly realising \(E\) as a vector bundle associated to a \G-equivariant principal bundle over \(X\).
\end{remark}

\begin{proof}[Proof of Theorem~\ref{moduledeform}]
Let \(\Theta \in \dual{G}\). In what follows, for \(\bx \in \dual{G}\), let \(\bC_\bx\) denote \(\bC\) endowed with unitary \G-representation \(e \circ \bx : G \to \U(1)\). 

First, suppose that \(\cE\) is free of rank \(1\), which, by Lemma~\ref{fgpmod}, we identify with \(\bC_\bx \otimes \cA\) for some \(\bx \in \dual{G}\). Since \(\cE_{\bx+\by} = \bC \otimes \cA_{\by}\) for all \(\by \in \dual{G}\), it follows that
\begin{gather*}
	(1 \otimes a) \lhd_\Theta b = \sum_{\by,\bz \in \dual{G}} \e{-\Theta(\bx+\by-\bz,\bz)} 1 \otimes \hat{a}(\by-\bz)\hat{b}(\bz) = 1 \otimes a \star_\Theta \alpha_{-\Theta(\bx,\cdot)}(b),\\
	\hp{1 \otimes a}{1 \otimes b}_\Theta = \sum_{\by,\bz \in \dual{G}} \e{-\Theta(-\bx+\by-\bz,\by)} \hat{a}(-\by+\bz)^\ast \hat{b}(\bz) = \alpha_{\Theta(\bx,\cdot)}(a^{\ast_\Theta} \star_\Theta b),
\end{gather*}
for all \(a\), \(b \in \cA\), and hence that \(\id \otimes \alpha_{\Theta(\bx,\cdot)}\) defines a \G-equivariant unitary isomorphism \(\cE_\Theta \iso \bC_\bx \otimes \cA_\Theta\) of Hermitian f.g.p.\ \G-\(\cA_\Theta\)-modules.

Let us now turn to the general case. Let \(\pi : G \to U(V)\) be any finite-dimensional unitary \G-representation and let \(p \in B(V) \otimes \cA\) be any \G-invariant orthogonal projection such that \(\cE \cong p(V \otimes \cA)\) as Hermitian Fr\'echet \GA-modules; let \((V,\pi) = (\bC v_1,e \circ \bx_1) \oplus \cdots (\bC v_N,e \circ \bx_N)\) be any orthogonal decomposition of \(\pi\) into irreducible unitary representations, where \(\bx_1,\dotsc,\bx_N \in \dual{G}\). On the one hand, since \(p\) is \G-invariant, it follows that \(\cE_\Theta \cong p(V \otimes \cA)_\Theta\), where \(p\) still defines a \G-invariant orthogonal projection in \(\End_{\cA_\Theta}((V \otimes \cA)_\Theta)\). On the other hand, by our discussion above, \(S \coloneqq \oplus_{k=1}^N (\id \otimes \alpha_{\Theta(\bx_k,\cdot)})\) defines a \G-equivariant unitary isomorphism \((V \otimes \cA)_\Theta \iso V \otimes \cA_\Theta\) of Hermitian f.g.p.\ \G-\(\cA_\Theta\)-modules. Thus, \(S\) restricts to a \G-equivariant unitary isomorphism
\(
	\cE_\Theta \cong p(V \otimes \cA)_\Theta \iso p_\Theta(V \otimes \cA_\Theta),
\)
of Hermitian Fr\'echet \G-\(\cA_\Theta\)-modules, where the \G-invariant orthogonal projection \(p_\Theta \coloneqq S p S^{-1} \in B(V) \otimes \cA_\Theta\) can be explicitly constructed in terms of \(p\):
\begin{align*}
	p_\Theta
	 &= \sum_{j,k=1}^N \left( v_j \otimes \ip{v_k}{} \right) \otimes \hp{v_j \otimes 1}{SpS^{-1}(v_k \otimes 1)}\\
	&= \sum_{j,k=1}^N \left( v_j \otimes \ip{v_k}{} \right)\otimes \hp{v_j \otimes 1}{p (v_k \otimes 1)}_\Theta\\
	&= \sum_{j,k=1}^N \left( v_j \otimes \ip{v_k}{} \right) \otimes \e{-\Theta(\bx_j,\bx_k-\bx_j)} \hp{v_j \otimes 1}{p (v_k \otimes 1)}. \qedhere
\end{align*}
\end{proof}

\begin{example}[Landi--Van Suijlekom~\cite{LVS05}, cf.\ Connes--Landi~\cite{CL}]
Let \(\theta \in \bR\) and define \(\Theta \in A(\bZ^2)\) by \(\Theta(\bx,\by) \coloneqq \exp(\tfrac{1}{2}\theta(x_1y_2-x_2 y_1))\) for \(\bx\), \(\by \in \bZ^2\). Recall that the Hopf fibration \(\SU(2) \to S^7 \to S^4\) is \(\bT^2\)-equivariant for the trivial \(\bT^2\)-action on \(\SU(2)\) and certain \(\bT^2\)-actions on \(S^7\) and \(S^4\), and hence deforms to a nontrivial quantum principal \(\SU(2)\)-bundle \(C^\infty(S^7_\theta) \hookleftarrow C^\infty(S^7_\theta)^{\SU(2)} \cong C^\infty(S^4_\theta)\), where \(C^\infty(S^7_\theta) \coloneqq C^\infty(S^7)_\Theta\) and \(C^\infty(S^4_\Theta) \coloneqq C^\infty(S^4)_\Theta\)~\cite[\S 5]{LVS05}.

Now, let \(n \in \bN\), and let \(V^{(n)} \coloneqq S^n \bC^2\) be the corresponding irreducible representation of \(\SU(2)\). On the one hand, by equivariance of the Hopf fibration, \(V^{(n)}\) gives rise to a \(\bT^2\)-equivariant Hermitian vector bundle \(E^{(n)} \coloneqq S^7 \times_{\SU(2)} V^{(n)}\), such that \(C^\infty(S^4,E^{(n)}) \cong p^{(n)} C^\infty(S^4)^{4^n}\) for an explicit \(\bT^2\)-invariant projection \(p^{(n)} \in M_{4^n}(C^\infty(S^4))\)~\cite[Appendix A]{LVS05}. On the other hand, Landi and Van Suijlekom~\cite{LVS05} use the deformed Hopf fibration to construct an associated \(C^\infty(S^4_\theta)\)-module \(\cE^{(n)}_\theta \coloneqq C^\infty(S^7_\theta) \boxtimes_{\SU(2)} V^{(n)}\), such that \(\cE^{(n)}_\theta \cong p^{(n)}_\theta C^\infty(S^4_\theta)\) for an explicit projection \(p^{(n)}_\theta \in M_{4^n}(C^\infty(S^4_\theta))\), with \(p^{(1)}_\theta\) recovering the instanton projection constructed by Connes and Landi~\cite[\S\S II--III]{CL}. By \(\bT^2\)-equivariance of all constructions, it not only follows that \(\cE^{(n)}_\theta = C^\infty(S^4,E^{(n)})_\Theta\), but also that \(p^{(n)}_\theta = (p^{(n)})_\Theta\) is precisely given by Equation~\ref{projformula} as applied to \(p^{(n)}\).
\end{example}

Finally, let us observe that if \(\cE\) is a Hermitian f.g.p.\ \G-\(\cA\)-module, then the algebra \(\End_{\cA_\Theta}(\cE_\Theta)\) can naturally be identified with \(\End_\cA(\cE)_\Theta\).

\begin{corollary}[{cf.\ Brain--Landi--Van Suijlekom~\cite[Proposition 4.2]{BLVS}}]\label{enddeform}
Let \(\cE\) be a Hermitian f.g.p. \G-\(\cA\)-module. Then for all \(\Theta \in B(\dual{G})\), the map \(\pi_\Theta : \End_\cA(\cE) \to \End_{\cA_\Theta}(\cE_\Theta)\) defined by
\begin{equation}\label{enddeformeq}
	\forall T \in \End_{\cA}(\cE), \; \forall \xi \in \cE, \enskip \pi_\Theta(T) \xi := \sum_{\bx\in\dual{G}}\sum_{\by \in \dual{G}} \e{-\Theta(\bx-\by,\by)} \hat{T}(\bx-\by) \hat{\xi}(\by),
\end{equation}
is a \G-equivariant topological \Star-isomorphism of Fr\'echet \Gstar-algebras.
\end{corollary}

\begin{proof}
By Lemma~\ref{fgpmod}, \(\End_\cA(\cE)\) is a Fr\'echet pre-\GCstar-algebra, so that \(\cE\) is simultaneously a left Fr\'echet \G-\(\End_\cA(\cE)\)-module and a Hermitian right \G-\(\cA\)-module. Thus, \(\cE\) can be simultaneously deformed to \(\cE_\Theta\) as both a left Fr\'echet \G-\(\End_\cA(\cE)_\Theta\)-module and a Hermitian right \G-\(\cA_\Theta\)-module, where the \(\End_\cA(\cE)_\Theta\)-module structure on \(\cE_\Theta\) is given by
\[
	\forall T \in \End_\cA(\cE)_\Theta, \; \forall \xi \in \cE_\Theta, \quad T \rhd_\Theta \xi \coloneqq \sum_{\bx,\by\in\dual{G}} \e{-\Theta(\bx,\by)}\hat{T}(\bx-\by)\hat{\xi}(\by),
\]
i.e., by \(T \rhd_\Theta \xi = \pi_\Theta(T)\xi\). Then, one can check on isotopic elements that \(T \mapsto (T \rhd_\Theta \cdot)\) defines an injective \G-equivariant \Star-homomorphism \(\pi_\Theta : \End_\cA(\cE)_\Theta \inj \End_{\cA_\Theta}(\cA_\Theta)\). 

Now, by Theorem~\ref{moduledeform}, \(\End_{\cA_\Theta}(\cE_\Theta)\) is also a Fr\'echet \Gstar-algebra, so that \(\cE_\Theta\) is simultaneously a left Fr\'echet \G-\(\End_{\cA_\Theta}(\cE_\Theta)\)-module and a Hermitian right \G-\(\cA_\Theta\)-module; in particular, \(\pi_\Theta\) is continuous. Thus, \(\cE_\Theta\) can be simultaneously deformed to \(\cE = (\cE_\Theta)_{-\Theta}\) as both a left Fr\'echet \G-\(\End_{\cA_\Theta}(\cE_\Theta)_{-\Theta}\)-module and as a Hermitian right \GA-module, where the  \(\End_{\cA_\Theta}(\cE_\Theta)_{-\Theta}\)-module structure is now given by
\[
	\forall T \in \End_{\cA_\Theta}(\cE_\Theta)_{-\Theta}, \; \forall \xi \in \cE_\Theta, \quad S \rhd_{-\Theta} \xi \coloneqq \sum_{\bx,\by\in\dual{G}} \e{\Theta(\bx,\by)}\hat{T}(\bx-\by)\hat{\xi}(\by).
\]
Again, one can check on isotypic elements that \(S \mapsto (S \rhd_{-\Theta})\) defines an injective \G-equivariant continuous \Star-homomorphism \(\pi_{-\Theta} : \End_{\cA_\Theta}(\cE_\Theta)_{-\Theta} \inj \End_\cA(\cE)\); by checking, once more, on isotypic elements, one can show that \(\pi_{-\Theta} = \inv{\pi_\Theta}\).
\end{proof}

\section{Connes--Landi deformation}\label{sec:2}

\subsection{Deformation of  \texorpdfstring{\G}{G}-equivariant spectral triples}

We now review the definition and construction of Connes--Landi deformations, following Yamashita~\cite[\S 2]{Ya} and Landi and Van Suijlekom~\cite[\S 3]{LVS}. We begin by fixing notation and terminology.

\begin{definition}
A \emph{spectral triple} is a triple \((\cA,H,D)\), where \(\cA\) is a unital \Star-algebra, \(H\) is a Hilbert space endowed with a faithful \Star-representation \(L : \cA \to B(H)\) of \(\cA\) on \(H\), and \(D\) is a densely-defined self-adjoint operator with compact resolvent on \(H\), such that \([D,L(a)] \in B(H)\) for all \(a \in \cA\). If, in addition, 
\(
	L(\cA) + [D,L(\cA)] \subset \cap_k \Dom ([\abs{D},\cdot])^k,
\)
then it is called \emph{regular}.
\end{definition}

From the standpoint of strict deformation quantisation, an essential feature of a regular spectral triple \((\cA,H,D)\) is that the algebra \(\cA\) can be canonically completed to a Fr\'echet pre-\Cstar-algebra.

\begin{proposition}[{Rennie~\cite[Lemma 16]{Rennie}}]\label{rennie}
Let \((\cA,H,D)\) be regular and let
\begin{equation}
	\cA^\delta \coloneqq \set{a \in \cA^{\prime\prime} \given L(a), \; [D,L(a)] \in \cap_k \Dom \delta^k}, \quad \delta \coloneqq [\abs{D},\cdot],
\end{equation}
endowed with the seminorms \(\set{\norm{\cdot}_k}_{k\in\bN}\) defined by
\begin{equation}\label{seminorms2}
	\forall k \in \bN, \; a \in \cA^\delta, \quad \norm{a}_k \coloneqq \sum_{j=0}^k \left( \norm{\delta^j(L(a))} + \norm{\delta^j([D,L(a)])} \right)
\end{equation}
Then \(\cA^\delta\) is Fr\'echet pre-\Cstar-algebra and \((\cA^\delta,H,D)\) defines a a regular spectral triple.
\end{proposition}

In the sequel, if \((\cA,H,D)\) is a regular spectral triple, then we always complete \(\cA\) to \(\cA^\delta\), and hence view \(\cA = \cA^\delta\) as a Fr\'echet pre-\Cstar-algebra by Proposition~\ref{rennie}.

\begin{remark}
The submultiplicative seminorms defined in Equation~\ref{seminorms2}, first introduced by Rennie~\cite[p.\ 147]{Rennie}, are equivalent to those defined by Yamashita~\cite[eq.\ 3]{Ya} and by Connes~\cite[eq.\ 12]{Con13}, respectively.
\end{remark}

In light of Proposition~\ref{rennie}, we define \G-equivariant spectral triples as follows.

\begin{definition}
A \emph{\G-equivariant spectral triple} is a regular spectral triple \((\cA,H,D)\) together with a strongly smooth action \(\alpha : G \to \Aut(\cA)\) by isometric \Star-automor\-phisms and a strongly continuous representation \(U : G \to \U(H)\), such that
\[
	\forall t \in G, \; \forall a \in \cA, \quad U_t L(a) U_t^\ast = L(\alpha_t(a)), \quad U_t D U_t^\ast = D.
\]
\end{definition}

\begin{remark}
In terms of our earlier conventions, if \((\cA,H,D)\) is a \G-equivariant spectral triple, then we consider \(\cA = \cA^\delta\) to be topologised as \((\cA^\delta)^\infty\) as needed.
\end{remark}

\begin{example}
If \(X\) is a compact oriented Riemannian \G-manifold and \(D\) a \G-invariant self-adjoint Dirac-type operator on a \G-equivariant Hermitian vector bundle \(E \to X\), then \((C^\infty(X),L^2(X,E),D)\) is a \G-equivariant spectral triple.
\end{example}

Now, let \((\cA,H,D)\) be a \G-equivariant spectral triple. By the classical Peter--Weyl theorem, we can decompose \(H\) as a direct sum of isotypic subspaces,
\[
	H = \oplus_{\bx \in \dual{G}} H_\bx, \quad H_\bx \coloneqq \set{\xi \in H \given \forall t \in G, \; U_t \xi = \e{\hp{\bx}{t}} \xi};
\]
let \(H^{\mathrm{fin}} \coloneqq \oplus_{\bx \in \dual{G}}^{\mathrm{alg}} H_\bx\) be their algebraic direct sum, which defines a dense subspace  of \(H\), and for each \(\bx \in \dual{G}\), let \(P_\bx\) denote the orthogonal projection onto \(H_\bx\). If \(\Theta \in Z^2(\dual{G},\bT)\), then the obvious analogy with  Fr\'echet \GA-modules suggests a deformation of \(L : \cA \to B(H)\) to a \G-equivariant \Star-representation \(L_\Theta : \cA_\Theta \to B(H)\), such that \((\cA_\Theta,H,D)\) remains a spectral triple.

\begin{theorem}[{Yamashita~\cite[Proposition 5]{Ya}}]\label{yamashita}
Let \((\cA,H,D)\) be a \G-equivariant spectral triple, and let \(\Theta \in Z^2(\dual{G},\bT)\). Then \((\cA_\Theta,H,D)\), endowed with the left \Star-representation \(L_\Theta : \cA_\Theta \to B(H)\) given by
\begin{equation}\label{convolve}
	\forall a \in \cA_\Theta,\; \forall \xi \in H^{\mathrm{fin}}, \quad  L_\Theta(a) \xi \coloneqq \sum_{\bx \in \dual{G}} \e{-\Theta(\bx-\by,\by)} L(\hat{a}(\bx-\by)) P_\by \xi,
\end{equation}
defines a \G-equivariant spectral triple.
\end{theorem}

This result, then, motivates and justifies the following definition.

\begin{definition}[{Connes--Landi~\cite[\S 5]{CL}, Yamashita~\cite[\S 2]{Ya}}]
Let \((\cA,H,D)\) be a \G-equivariant spectral triple. The \emph{Connes--Landi deformation} of \((\cA,H,D)\) by \(\Theta \in Z^2(\dual{G},\bT)\) is the \G-equivariant spectral triple \((\cA_\Theta,H,D)\) of Proposition~\ref{yamashita}.
\end{definition}

\begin{example}[{cf.\ Connes--Dubois-Violette~\cite[\S 13]{CDV}, D\k{a}browski~\cite{Dabrowski}}]\label{torus4}
Consider \(G = \bT^N\) and \(\cA = C^\infty(\bT^N)\), where the \(\bT^N\)-action on \(C^\infty(\bT^N)\) is induced by the translation action on \(\bT^N\). Fix a spin structure on \(\bT^N\) and let \(\slashed{S} \to \bT^N\) and \(\slashed{D}\) denote the corresponding spinor bundle and Dirac operator, respectively; the translation action of \(\bT^N\) on itself lifts to an action of a double cover \(\widetilde{\bT^N}\) of \(\bT^N\) on \(\slashed{S}\), making \(\slashed{S}\) a \(\widetilde{\bT^N}\)-equivariant Hermitian vector bundle and \(\slashed{D}\) a \(\widetilde{\bT^N}\)-invariant Dirac operator. 

Now, let \(\theta \in \wedge^2\bR^N\), let \(\Theta \in A(\bZ^N)\) be defined by \(\Theta(\bx,\by) \coloneqq \exp(\tfrac{1}{2}\theta(\bx_\ast,\by_\ast))\) for \(\bx\), \(\by \in \bZ^N\), and let \(\widetilde{\Theta}\) denote the pullback of \(\Theta\) to a bicharacter on the Pontrjagin dual \(\widetilde{\bZ^N}\) of \(\widetilde{\bT^N}\). Then the Connes--Landi deformation
\[
	(C^\infty(\bT^N_\theta),L^2(\bT^N_\theta,\slashed{S}),\slashed{D}) \coloneqq (C^\infty(\bT^N)_{\widetilde{\Theta}},L^2(\bT^N,\slashed{S}),\slashed{D})
\]
 of \((C^\infty(\bT^N),L^2(\bT^N,\slashed{S}),\slashed{D})\) \emph{qua} \(\widetilde{\bT^N}\)-equivariant spectral triple by \(\widetilde{\Theta}\) is precisely the spectral triple of the noncommutative \(N\)-torus \(C^\infty(\bT^N)\) with deformation parameter \(\theta\), corresponding to the chosen spin structure for \(\bT^N\).
\end{example}

Let us now prove Proposition~\ref{yamashita}. Recall that a strongly continuous unitary representation \(U : G \to \U(H)\) of \(G\) on a Hilbert space \(H\) induces an action \(\beta : G \to \Aut(B(H))\) of \(G\) on \(B(H)\) by \Star-automorphisms, strongly continuous for the strong operator topology, defined by \(\beta_t(S) \coloneqq U_t S U_t^\ast\) for \(S \in B(H)\) and \(t \in G\). Thus, if
\[
	B^\infty(H) \coloneqq \set{S \in B(H) \given G \ni t \mapsto \beta_t(S) \; \text{is norm-smooth}},
\]
then \((B^\infty(H),\beta)\), endowed with the seminorms of Equation~\ref{extraseminorms}, defines a Fr\'echet pre-\GCstar-algebra with \G-equivariant \Cstar-closure \((B^c(H),\beta)\), where
\[
	B^c(H) \coloneqq \set{S \in B(H) \given G \ni t \mapsto \beta_t(S) \; \text{is norm-continuous}}.
\]
Now, if \((\cA,H,D)\) is a \G-equivariant spectral triple and if \(\Theta \in Z^2(\dual{G},\bT)\), then the \G-equivariant \Star-representation \(L : \cA \to B(H)\) actually defines a \G-equivariant isometric \Star-homomorphism \(L : \cA \inj B^\infty(H)\), which, therefore, also defines a \G-equivariant isometric \Star-homomorphism \(L : \cA_\Theta \inj B^\infty(H)_\Theta\). Thus, it suffices to construct a suitable \G-equivariant topological \Star-isomorphism \(\pi_\Theta : B^\infty(H)_\Theta \iso B^\infty(H)\), such that \(L_\Theta \coloneqq \pi_\Theta \circ L\) has the correct form.

We begin with a technical convenience to facilitate this construction.

\begin{lemma}\label{Ypsilon}
Let \((\cA,H,D)\) be a \G-equivariant spectral triple; let \(\Theta \in Z^2(\dual{G},\bT)\). For any \(\bx \in \dual{G}\), \(\Upsilon^\Theta_\bx \coloneqq \sum_{\by \in \dual{G}} \e{-\Theta(\bx,\by)}P_\by\) strongly converges to a \G-invariant unitary operator on \(H\) that commutes with \(D\).
\end{lemma}

\begin{proof}
Strong convergence of \(\sum_{\by \in \dual{G}} \e{-\Theta(\bx,\by)}P_\by \eqqcolon \Upsilon^\Theta_\bx\) follows from strong convergence of \(\sum_{\by \in \dual{G}} P_\by = \id_H\); since \(\Upsilon^\Theta_\bx\) is  diagonal with respect to the Peter--Weyl decomposition \(H = \oplus_{\by \in \dual{G}} H_\by\), acting as a unimodular constant on each isotypic subspace, it defines a \G-invariant unitary. Finally, since \(D\) is \G-invariant, the strongly continuous representation \(U : G \to \U(H)\) restricts to a strongly continuous representation \(U : G \to U(\Dom D)\), with respect to which the isotypic subspaces of \(\Dom D\) are simply \((\Dom D)_\by = \Dom D \cap H_\by\) for \(\by \in \dual{G}\); thus, \(\Upsilon^\Theta_\bx\) restricts to a \G-invariant unitary on \(\Dom D\), constant on each isotypic subspace, and, as such, commutes with \(D\).
\end{proof}

\begin{remark}
If \(\Theta\) is a bicharacter, then \(\bx \mapsto \Upsilon^\Theta_\bx = U_{-\Theta(\bx,\cdot)}\) defines a unitary representation of \(\dual{G}\) on \(H\); otherwise, the map \(\bx \mapsto \Upsilon^\Theta_\bx\) is not even projective.
\end{remark}

Given the unitaries \(\Upsilon^\Theta_\bx\), we can now construct the desired isomorphism \(\pi_\Theta : B^\infty(H)_\Theta \iso B^\infty(H)\), and hence proceed with the proof of Proposition~\ref{yamashita}.

\begin{lemma}\label{bigdeform}
Let \((\cA,H,D)\) be a \G-equivariant spectral triple; let \(\Theta \in Z^2(\dual{G},\bT)\). The map \(\pi_\Theta : B^\infty(H)_\Theta \to B^\infty(H)\) defined by
\begin{equation}\label{bigdeform:1}
	\forall S \in B^\infty(H)_\Theta, \quad \pi_\Theta(S) \coloneqq \inv{F}\left(\left(\hat{S}(\bx)\Upsilon_\bx^\Theta\right)_{\bx \in \dual{G}}\right) = \sum_{\bx \in \dual{G}} \hat{S}(\bx)\Upsilon^\Theta_\bx
\end{equation}
is a \G-equivariant topological \Star-isomorphism such that
\begin{gather}
	\forall S \in B^\infty(H)_\Theta, \; \forall \xi \in H^{\mathrm{fin}}, \quad \pi_\Theta(S)\xi = \sum_{\bx,\by\in\dual{G}} \e{-\Theta(\bx-\by,\by)}\hat{S}(\bx-\by)P_\by\xi,\label{bigdeform:2}\\
	\forall S \in B^\infty(H) \cap \Dom [\abs{D},\cdot], \quad [\abs{D},\pi_\Theta(S)] = \pi_\Theta([\abs{D},S]) \in B^\infty(H).\label{bigdeform:3}
\end{gather}
\end{lemma}

\begin{proof}
Since the operators \(\Upsilon^\Theta_\bx\) are all \G-invariant and unitary, if \(S \in B^\infty(H)\), then \(\left(\hat{S}(\bx)\Upsilon^\Theta_\bx\right)_{\bx \in \dual{G}} \in \cS(\dual{G};B^\infty(H))\) if and only if \(\hat{S} \in \cS(\dual{G};B^\infty(H))\), so that \(\pi_\Theta : B^\infty(H) \to B^\infty(H)\) is well-defined; from there, it is easy to check that \(\pi_\Theta\) defines a \G-equivariant \Star-preserving bicontinuous linear map with inverse \(\pi_{-\Theta}\). 

Now, by the proof of Lemma~\ref{Ypsilon}, \(D\) and \(\abs{D}\) define \G-equivariant bounded operators \(\Dom D \to H\). Let \(S \in B^\infty(H) \cap \Dom[\abs{D},\cdot]\). Then \(S\) restricts to an element of \(B^\infty(\Dom D)\), whilst \([\abs{D},S] \in B^\infty(H)\), so that one can simply check on the dense subspace \((\Dom D)^{\mathrm{fin}} = \oplus_{\by \in\dual{G}}^{\mathrm{alg}} (\Dom D \cap H_\by)\) of \(\Dom D\) that \([\abs{D},\pi_\Theta(S)] = \pi_\Theta([\abs{D},S])\) as bounded operators \(\Dom D \to H\), and hence, since \(\pi_\Theta([\abs{D},S]) \in B^\infty(H)\), as bounded operators on \(H\) itself.
\end{proof}

\begin{proof}[Proof of Proposition~\ref{yamashita}]
By Lemma~\ref{bigdeform} and \G-equivariance of \(L : \cA \inj B^\infty(H)\), we can safely write \(L_\Theta = \pi_\Theta \circ L : \cA_\Theta \inj B^\infty(H)\). Hence, by Lemma~\ref{bigdeform}, it suffices to prove that \([D,L_\Theta(a)] = \pi_\Theta([D,L(a)])\) for all \(a \in \cA\); however, this follows, \emph{mutatis mutandis}, from the proof of Equation~\ref{bigdeform:3}.
\end{proof}

Finally, let us record the basic properties of Connes--Landi deformation by analogy with Propositions~\ref{isom} and~\ref{modfunctor}. Before continuing, recall that two (regular) spectral triples \((\cA_1,H_1,D_1)\) and \((\cA_2,H_2,D_2)\) are called \emph{unitarily equivalent}, denoted \((\cA_1,H_1,D_1) \cong (\cA_2,H_2,D_2)\), if there exist a (topological) \Star-isomorphism \(\phi : \cA_1 \to \cA_2\) and a \G-equivariant unitary \(\Phi : H_1 \to H_2\) such that
\[
	\forall a \in \cA_1, \quad \Phi L_1(a) \Phi^\ast = L_2(\phi(a)); \quad \Phi D_1 \Phi^\ast = D_2;
\]
in which case we call \((\phi,\Phi)\) a \emph{unitary equivalence} from \((\cA_1,H_1,D_1)\) to \((\cA_2,H_2,D_2)\), written \((\phi,\Phi) : (\cA_1,H_1,D_1) \iso (\cA_2,H_2,D_2)\). In particular, if \((\cA_1,H_1,D_1)\) and \((\cA_2,H_2,D_2)\) are \G-equivariant spectral triples and if \(\phi\) and \(\Psi\) are \G-equivariant, then we shall say that \((\cA_1,H_1,D_1)\) and \((\cA_2,H_2,D_2)\) are \emph{\(G\)-equivalent}, denoted \((\cA_1,H_1,D_1) \cong_G (\cA_2,H_2,D_2)\), and call \((\phi,\Phi)\) a \emph{\G-equivalence}.

\begin{proposition}
Let \(\cat{SpTr}_G\) denote the category of \G-equivariant spectral triples and \G-equivalences, where the composition of two \G-equivalences \[(\phi,\Phi) : (\cA_1,H_1,D_1) \iso (\cA_2,H_2,D_2), \quad (\p{\phi},\p{\Phi}) : (\cA_2,H_2,D_2) \iso (\cA_3,H_3,D_3),\] is given by
\(
	(\p{\phi},\p{\Phi}) \circ (\phi,\Phi) \coloneqq (\p{\phi} \circ \phi,\p{\Phi} \circ \Phi) : (\cA_1,H_1,D_1) \to (\cA_3,H_3,D_3)
\).
For every \(\Theta \in Z^2(\dual{G},\bT)\), the assignment
\begin{align*}
	&\operatorname{Obj}(\cat{SpTr}_G) \ni (\cA,H,D) \mapsto (\cA_\Theta,H,D),\\
	&\operatorname{Mor}(\cat{SpTr}_G) \ni \left((\phi,\Phi) : (\cA_1,H_1,D_1) \iso (\cA_2,H_2,D_2)\right)\\ 
	&\quad\quad\quad\quad\quad\quad\quad\quad\quad \mapsto \left((\phi,\Phi) : ((\cA_1)_\Theta,H_1,D_1) \iso ((\cA_2)_\Theta,H_2,D_2)\right),\end{align*}
defines an isomorphism of categories \(\cat{CL}_\Theta : \cat{SpTr}_G \iso \cat{SpTr}_G\). Moreover, the functors \(\cat{CL}_\Theta\) satisfy the following properties.
\begin{enumerate}
	\item Let \(\id_{\cat{SpTr}_G}\) denote the identity functor on \(\cat{SpTr}_G\). Then \(\cat{CL}_0 = \id_{\cat{SpTr}_G}\).
	\item For any \(\Theta\), \(\p{\Theta} \in Z^2(\dual{G},\bT)\), \(\cat{CL}_{\p{\Theta}} \circ \cat{CL}_\Theta = \cat{CL}_{\Theta + \p{\Theta}}\).
\end{enumerate}
Thus, \(\Theta \mapsto \cat{CL}_\Theta\) defines a strict action of the group \(Z^2(\dual{G},\bT)\) on \(\cat{SpTr}_G\).
\end{proposition}

Given Proposition~\ref{isom} and the techniques used to prove Proposition~\ref{yamashita}, the proof of these properties reduces, again, to a mechanical, algebraic fact-check.

\subsection{Group cohomology and symmetries of Connes--Landi deformation}\label{cohomology}

We now show that the second group cohomology \(H^2(\dual{G},\bT)\) of the Pontrjagin dual \(\dual{G}\) of \(G\) parametrizes Connes--Landi deformations of a given \G-equivariant spectral triple \(\Sigma\), up to \emph{natural} \G-equivalence.

\begin{theorem}[{cf.\ Venselaar~\cite[Lemma 1]{Venselaar}}]\label{equivariant}
For any \(\Theta\), \(\p{\Theta} \in Z^2(\dual{G},\bT)\), the functors \(\cat{CL}_\Theta\) and \(\cat{CL}_{\p{\Theta}}\) are naturally isomorphic if and only if \(\Theta\) and \(\p{\Theta}\) are cohomologous.
\end{theorem}

\begin{proof}
Let \(\Theta\), \(\p{\Theta} \in Z^2(\dual{G},\bT)\). Let \(\theta \coloneqq [\Theta]\), \(\p{\theta} \coloneqq [\p{\Theta}]\), and recall that \(\iot{\theta} = \Theta - \Theta^t\), \(\iot{\p{\theta}} = \p{\Theta} - (\p{\Theta})^t.\) so that, by commutativity of \(C^\infty(G)\),
\[
	b_\by \star_\Theta a_\bx = \e{\iot{\theta}} a_\bx \star_\Theta b_\by, \quad b_\by \star_{\p{\Theta}} a_\bx = \e{\iot{\p{\theta}}} a_\bx \star_{\p{\Theta}} b_\by
\] 
for all isotypic \(a_\bx \in C^\infty(G)_\bx = \bC U_\bx\) and \(b_\by \in C^\infty(G)_\by = \bC U_\by\) in \(C^\infty(G)\).

First, suppose that \(\cat{CL}_\Theta\) and \(\cat{CL}_{\p{\Theta}}\) are naturally isomorphic. Let \(\cat{\Psi} : \cat{CL}_\Theta \iso \cat{CL}_{\p{\Theta}}\) be a natural isomorphism. In particular, \(\cat{\Psi}\) yields a \G-equivalence \((\psi,\Psi) : (C^\infty(G)_\Theta,L^2(G,\slashed{S}),\slashed{D}) \iso (C^\infty(G)_{\p{\Theta}},L^2(G,\slashed{S}),\slashed{D})\), where \(\slashed{S} \to G\) is the spinor bundle and \(\slashed{D}\) is the Dirac operator associated to the trivial spin structure.
%For all \(\bx\), \(\by \in \dual{G}\), on the one hand,
%\[
%	U_\by \star_\Theta U_\bx = \e{-\Theta(\by,\bx)} U_\by U_\bx = \e{-\Theta(\by,\bx)} U_\bx U_\by = \e{\iot{\theta}(\bx,\by)} U_\bx \star_\Theta U_\by,
%\]
%whilst on the other, by an identical computation,
%\[
%	\psi(U_\by) \star_\p{\Theta} \psi(U_\bx) = \e{\iot{\p{\theta}}(\bx,\by)} \psi(U_\bx) \star_\p{\Theta} \psi(U_\by)
%\]
%since \(\psi(U_\bx) \in C^\infty(G)_\bx\) and \(\psi(U_\by) \in C^\infty(G)_\by\) by \(G\)-equivar\-i\-ance of \(\psi\). Thus, 
Then
\begin{multline*}
	\e{\iot{\p{\theta}}(\bx,\by)} \psi(U_\bx) \star_\p{\Theta} \psi(U_\by)
 = \psi(U_\by) \star_\p{\Theta} \psi(U_\by) = \psi(U_\by \star_\Theta U_\bx)\\ = \psi(\e{\iot{\theta}(\bx,\by)} U_\bx \star_\Theta U_\by) = \e{\iot{\theta}(\bx,\by)} \Psi(U_\bx) \star_\p{\Theta} \psi(U_\by)
\end{multline*}
for all \(\bx\), \(\by \in \dual{G}\), so that \(\iot{\theta} = \iot{\p{\theta}}\), and hence, by Theorem~\ref{kleppner}, \(\theta = \p{\theta}\).

Now, suppose that \(\theta = \p{\theta}\), so that \(\p{\Theta} - \Theta = \dif T\) for some \(T : \dual{G} \to \bT\) such that \(T(\bo) = 0\). For any \(\Sigma = (\cA,H,D) \in \operatorname{Obj}(\cat{SpTr}_G)\), define \(\psi_\Sigma : \cA_\Theta \to \cA_{\p{\Theta}}\) by
\[
	\forall a \in \cA, \quad \psi_\Sigma(a) \coloneqq \inv{F}\left(\e{T(\bx)}\hat{a}_\bx)_{\bx \in \dual{G}}\right) = \sum_{\bx \in \dual{G}} \e{T(\bx)}a_\bx,
\]
and \(\Psi_\Sigma : H \to H\) by \(\Psi_\Sigma \coloneqq \sum_{\bx \in \dual{G}}\e{T(\bx)}P_\bx\); by construction, \(\psi_\Sigma\) defines a \G-equivariant bicontinuous linear map, whilst by the proof of Lemma~\ref{Ypsilon}, \(\Psi_\Sigma\) converges to a \G-invariant unitary that commutes with \(D\). On the one hand, since
\begin{gather*}
	T(\bx+\by)-\Theta(\bx,\by) = -\p{\Theta}(\bx,\by) + T(\bx) + T(\by), \\
	T(-\bx)-\Theta(\bx,\bx)  = -\Theta(-\bx,-\bx)-\Theta(\bx)
\end{gather*}
%\begin{gather*}
%\begin{multlined}
%	T(\bx+\by)-\Theta(\bx,\by) = -\di T(\bx,\by) + T(\bx) + T(\by) - \Theta(\bx,\by)\\ = -\p{\Theta}(\bx,\by) + T(\bx) + T(\by),
%\end{multlined}\\
%\begin{multlined}
%	T(-\bx)-\Theta(\bx,\bx) = -\di T(-\bx,-\bx) - T(\bx) + T(\bo) - \Theta(\bx,\bx)\\ = -\Theta(-\bx,-\bx)-\Theta(\bx)
%\end{multlined}
%\end{gather*}
for all \(\bx\), \(\by \in \dual{G}\), it follows that
%\begin{gather*}
%\begin{multlined}
%\psi_\cA(a_\bx \star_\Theta b_\by) = \e{T(\bx+\by)}\e{-\Theta(\bx,\by)}a_\bx b_\by\\ = \e{-\p{\Theta}(\bx,\by)} \e{T(\bx)}a_\bx \e{T(\by)}b_\by = \psi_\cA(a_\bx) \star_\p{\Theta} \psi_\cA(b_\by),
%\end{multlined}\\
%\psi_\cA(a_\bx^{\ast_\Theta}) = \e{T(-\bx)}\e{-\Theta(\bx,bx)}a_\bx^\ast = \e{-\Theta(-\bx,-\bx)}(\e{T(\bx)}a_\bx)^\ast = \psi_\cA(a_\bx)^{\ast_\p{\Theta}}
%\end{gather*}
\(\psi_\Sigma(a_\bx \star_\Theta b_\by) = \psi_\Sigma(a_\bx)\star_\p{\Theta}\psi_\Sigma(b_\by)\) and \(\psi_\Sigma(a_\bx^{\ast_\Theta}) = \psi_\Sigma(a_\bx)^{\ast_\p{\Theta}}\) for all isotypic elements \(a \in \cA_\bx\) and \(b_\by \in \cA_\by\) of \(\cA\), and hence that \(\psi_\Sigma\) defines a \G-equivariant topological \Star-isomorphism. On the other hand, since
\(
	T(\bx+\by) - \Theta(\bx,\by) - T(\by) = -\p{\Theta}(\bx,\by) + T(\bx)
\)
for all \(\bx\), \(\by \in \dual{G}\), one can check that \(\Psi_\Sigma L_\Theta(a_\bx) \Psi_\Sigma^\ast \xi_\by = L_{\p{\Theta}}(\psi_\Sigma(a_\bx))\xi_\by\) for all isotypic elements \(a_\bx \in \cA_\bx\) and \(\xi_\by \in H_\by\) of \(\cA\) and \(H\), respectively, so that \(\Psi_\Sigma\) spatially implements \(\psi_\Sigma\). Thus, \((\psi_\Sigma,\Psi_\Sigma)\) defines a \G-equivalence. Moreover, by the universal form of this construction, the collection \(\set{(\psi_\Sigma,\Psi_\Sigma) \given \Sigma \in \operatorname{Obj}(\cat{SpTr}_G)}\) of \G-equivalences defines a natural isomorphism \(\cat{\Psi} : \cat{CL}_\Theta \iso \cat{CL}_{\p{\Theta}}\).
\end{proof}

Thus, up to natural \G-equivalence, we can define the Connes--Landi deformation of \((\cA,H,D) \in \operatorname{Obj}(\cat{SpTr}_G)\) by \(\theta \in H^2(\dual{G},\bT)\) to be
\(
	(\cA_\theta,H,D) \coloneqq (\cA_\Theta,H,D)
\)
for any representative \(\Theta\) of \(\theta\); in fact, by Theorem~\ref{kleppner},  we can even take \(\Theta \in B(\dual{G})\).

\begin{remark}
\emph{Mutatis mutandis}, we have also shown that the functors \(\cat{R}_\Theta, \cat{R}_{\p{\Theta}} : \cat{Alg}_G \iso \cat{Alg}_G\) are naturally isomorphic if and only if \(\Theta\) and \(\p{\Theta}\) are cohomologous. Moreover, given \(T : \dual{G} \to \bT\) such that \(\dif T = \p{\Theta} - \Theta\), and given the corresponding explicit natural isomorphism \(\psi^T : \cat{R}_\Theta \iso \cat{R}_{\p{\Theta}}\), one can use the same construction, yet again, to obtain a natural isomorphism \(\Psi^T : \cat{R}_{\cA,\Theta} \iso (\psi^T_\cA)^\ast \circ \cat{R}_{\cA,\p{\Theta}}\) of functors \(\cat{Mod}_G(\cA) \iso \cat{Mod}_G(\cA)\). As a result, if \(\cA\) is a Fr\'echet pre-\GCstar-algebra and if \(\cE\) is an f.g.p.\ \GA-module, then for \emph{any} \(\Theta \in Z^2(\dual{G},\bT)\), \(\cE_\Theta\) is a f.g.p.\ \G-\(\cA_\Theta\)-module.
\end{remark}

\begin{example}\label{nctorusex2}
We continue from Example~\ref{nctorusex}. For simplicity, let \(\slashed{S} = \bT^N \times \bC^{2^{\lfloor N/2 \rfloor}}\) be the spinor bundle and let \(\slashed{D}\) be the spin Dirac operator on \(\bT^N\) corresponding to the trivial spin structure. Let \(\theta \in \wedge^2 \bR^N\) and define \(\Theta \in A(\bZ^N)\) by \(\Theta(\bx,\by) \coloneqq \exp(\tfrac{1}{2}\theta(\bx_\ast,\by_\ast))\) for \(\bx\), \(\by \in \dual{G}\), so that \[(C^\infty(\bT^N_\theta),L^2(\bT^N_\theta,\slashed{S}),\slashed{D}) \coloneqq (C^\infty(\bT^N)_\Theta, L^2(\bT^N,\slashed{S}),\slashed{D})\] is the smooth noncommutative \(N\)-torus with deformation parameter \(\theta\) and trivial spinor structure. Then \(\iot{[\Theta]}(\bx,\by) = \e{\theta(\bx_\ast,\by_\ast)}\) for all \(\bx\), \(\by \in \bZ^N\), so that  \([\Theta] = p(\theta)\) for the universal covering map 
\[
	p : \wedge^2 \bR^N \surj \bT^{N(N-1)/2} \cong H^2(\bZ^N,\bT), \quad 
	\theta \mapsto (\exp(\theta(e_i,e_j)))_{1 \leq i < j \leq N},
\]
where \(\set{e_i}\) denotes the standard ordered basis of \(\bR^N\), and hence, we have natural \G-equivalences
\(
	(C^\infty(\bT^N_\theta),L^2(\bT^N_\theta,\slashed{S}),\slashed{D}) \cong_G (C^\infty(\bT^N)_{\p{\Theta}},L^2(\bT^N,\slashed{S},\slashed{D})
\)
for \emph{every} representative \(\p{\Theta} \in Z^2(\dual{G},\bT)\) of \(p(\theta)\). In particular, we could have used any of the three constructions of Example~\ref{torus1} to obtain \(\p{\Theta} \in B(\bZ^N)\); for instance, the alternating bicharacter \(\Theta\) fits precisely into the third construction. Observe, moreover, that we can rewrite the defining commutation relations of \(C^\infty(\bT^N_\theta)\) as
\[
	\forall \bx, \; \by \in \bZ^N, \quad U_\by \star_{\Theta} U_\bx = \e{\iot{p(\theta)}(\bx,\by)} U_\bx \star_\Theta U_\by,
\]
which, as a result, manifestly depend only on the cohomology class \(p(\theta) = [\Theta]\) through its canonically associated alternating bicharacter \(\iot{p(\theta)}\). Thus, as is well known in the literature on \(W^\ast\)- and \Cstar-dynamical systems~\cite{AHK,OPT,MW}, the deformation parameter \(\theta\) of a noncommutative \(N\)-torus \(C^\infty(\bT^N)\) is most naturally viewed as an element of \(H^2(\bZ^N,\bT) \cong \bT^{N(N-1)/2}\).
\end{example}

Since the Connes--Landi deformations of a fixed \G-equivariant spectral triple \(\Sigma\) can be parametrized, up to natural \G-equivalence, by the compact Abelian Lie group \(H^2(\dual{G},\bT)\), it is natural to ask which symmetries of \(H^2(\dual{G},\bT)\) lift to unitary equivalences or even more general Morita--Rieffel equivalences amongst (\G-equivalence classes of) Connes--Landi deformations of \(\Sigma\). For the remainder of this section, we discuss symmetries induced by automorphisms of the Lie group \(G\).

Now, the action of \(\Aut(G)\) on \(G\) induces a right action of \(\Aut(G)\) on \(\dual{G}\) by \(\hat{\phi}(\bx) \coloneqq \bx \circ \phi\) for \(\phi \in \Aut(G)\) and \(\bx \in \dual{G}\), and hence, an action on \(Z^2(\dual{G},\bT)\) by
\[
	\forall \Theta \in Z^2(\dual{G},\bT), \; \forall \phi \in \Aut(G), \; \forall \bx, \by \in \dual{G}, \quad \hat{\phi}^\ast\Theta(\bx,\by) \coloneqq \Theta(\hat{\phi}(\bx),\hat{\phi}(\by)),
\]
which, at last, descends to an action on \(H^2(\dual{G},\bT)\) by \(\hat{\phi}^\ast [\Theta] \coloneqq [\hat{\phi}^\ast\Theta]\) for \(\phi \in \Aut(G)\) and \(\Theta \in Z^2(\dual{G},\bT)\). On the other hand, if \(\Sigma \coloneqq (\cA,H,D) \in \operatorname{Obj}(\cat{SpTr}_G)\), then certain automorphisms of \(G\) can be viewed as acting non-isometrically on \(\Sigma\).

\begin{definition}
Let \(\Sigma = (\cA,H,D)\) be a \G-equivariant spectral triple. We call \(\phi \in \Aut(G)\) \emph{implementable on \(\Sigma\)} if there exist \(\psi \in \Aut(\cA)\), \(\Psi \in \U(H)\) such that
\begin{gather*}
	\forall a \in \cA, \quad \Psi L(a)\Psi^\ast = L(\psi(a)),\\
	\forall t \in G, \quad \psi \circ \alpha_t \circ \inv{\psi} = \alpha_{\phi(t)}, \quad \Psi U_t \Psi^\ast = U_{\phi(t)},
\end{gather*}
in which case, we call \((\psi,\Psi)\) an \emph{implementation} of \(\phi\). We denote by \(\Aut_\Sigma(G)\) the subgroup of all implementable automorphisms in \(\Aut(G)\).
\end{definition}

The point, then, is that automorphisms of \(G\) implementable on a given \G-equivariant spectral triple \(\Sigma\) induce what can be viewed as non-equivariant, non-isometric unitary equivalences amongst the Connes--Landi deformations of \(\Sigma\).

\begin{proposition}[{cf.\ Venselaar~\cite[Corollary 2]{Venselaar}}]\label{nonequivariant}
Let \(\Sigma = (\cA,H,D)\) be a \G-equivariant spectral triple.  Let \(\Theta \in Z^2(\dual{G})\) and \(\phi \in \Aut_\Sigma(G)\). Any implementation \((\psi,\Psi)\) of \(\phi\) defines a unitary equivalence \((\psi,\Psi) : (\cA_{\hat{\phi}^\ast\Theta},H,D) \iso (\cA_\Theta,H,\Psi D \Psi^\ast)\).
\end{proposition}

\begin{proof}
By definition, \((\psi,\Psi) : (\cA,H,D) \iso (\cA,H,\Psi D \Psi^\ast)\) is a unitary equivalence. In particular, since \((\psi,\Psi)\) implements \(\phi\), it follows that that \(\psi(\cA_{\bx}) = \cA_{\hat{\phi}{}^{-1}(\bx)}\) and \(\Psi H_\bx = H_{\hat{\phi}{}^{-1}(\bx)}\) for all \(\bx \in \dual{G}\). Thus, for all isotypic \(a_\bx \in \cA_\bx\) and \(\xi_\by \in H_\by\),
\begin{multline*}
	\Psi L_{\hat{\phi}^\ast\Theta}(a_\bx)\xi_\by = \e{\hat{\phi}^\ast\Theta(\bx,\by)} \Psi L(a_\bx) \xi_\by = \e{\Theta(\hat{\phi}(\bx),\hat{\phi}(\by))} L(\psi(a_\bx)) \Psi \xi_\by\\ = L_\Theta(\psi(a_\bx)) \Psi \xi_\by,
\end{multline*}
so that, indeed, \(\Psi L_{\hat{\phi}^\ast\Theta}(\cdot)\Psi^\ast = L_\Theta \circ \psi\).
\end{proof}

Thus, by a very mild abuse of notation, we can say that if \(\Sigma = (\cA,H,D)\) is a \G-equivariant spectral triple, then
\[
	\forall \phi \in \Aut_\Sigma(G), \; \forall \theta \in H^2(\dual{G},\bT), \quad (\cA_{\hat{\phi}^\ast\theta},H,D) \cong (\cA_\theta,H,D_\phi),
\]
where, up to \G-equivalence, \((\cA,H,D_\phi) \coloneqq (\cA,H,\Psi D \Psi^\ast)\) for any implementation \((\psi,\Psi)\) of \(\phi\); in other words, \(\phi \in \Aut_\Sigma(G)\) acts by non-isometric diffeomorphisms on \(\set{\cA_\theta \given \theta \in H^2(\dual{G},\bT)}\) \emph{qua} family of noncommutative manifolds. 

Finally, let us observe that between the non-equivariant unitary equivalences of Proposition~\ref{nonequivariant} and the \G-equivalences of Theorem~\ref{equivariant}, we have recovered all the standard isomorphisms of smooth noncommutative tori.

\begin{example}[Rieffel--Schwarz~\cite{RS}, Elliott--Li~\cite{EL}]
We continue from Examples~\ref{nctorusex} and~\ref{nctorusex2}. For \(R\) a commutative ring, let \(\Skew(N,R)\) denote the \(R\)-module of skewsymmetric \(N \times N\) matrices with entries in \(R\). Let \(\set{e_1,\dotsc,e_N}\) denote the standard ordered basis of \(\bR^N\), and define an isomorphism \(\wedge^2 \bR^N \iso \Skew(N,\bR)\) by \(\Theta \mapsto (\Theta(e_i,e_j))\), so that the universal cover \(p : \wedge^2 \bR^N \surj \bT^{N(N-1)/2} \cong H^2(\bZ^N,\bT)\) is given by the composition \[\wedge^2 \bR^N \cong \Skew(N,\bR) \surj \Skew(N,\bR)/\Skew(N,\bZ) \cong \bT^{N(N-1)/2} \cong H^2(\bZ^N,\bT).\] Let \((C^\infty(\bT^N),L^2(\bT^N,\slashed{S}),\slashed{D})\) be the \(\bT^N\)-equivariant spectral triple of \(\bT^N\) endowed with the trivial spin structure, and recall that for any \(\theta \in \wedge^2 \bR^N \cong \Skew(N,\bR)\) and any representative \(\Theta\) of \(p(\theta)\),
\[
	(C^\infty(\bT^N_\theta),L^2(\bT^N_\theta,\slashed{S}),\slashed{D}) \cong_G (C^\infty(\bT^N)_{\Theta},L^2(\bT^N,\slashed{S}),\slashed{D}).
\]
On  the one hand, by Theorems~\ref{kleppner} and~\ref{equivariant}, for all \(\theta\), \(\p{\theta} \in \Skew(N,\bR)\),
\begin{equation}\label{equivariantex}
	(C^\infty(\bT^N_\theta),L^2(\bT^N_\theta,\slashed{S}),\slashed{D})  \cong_G (C^\infty(\bT^N_{\p{\theta}}),L^2(\bT^N_{\p{\theta}},\slashed{S}),\slashed{D})
\end{equation}
if and only if \(p(\theta) = p(\p{\theta})\), if and only if \(\p{\theta} - \theta \in \Skew(N,\bZ)\). On the other hand, since \(L^2(\bT^N,\slashed{S}) \cong L^2(\bT^N) \otimes \bC^{2^{\lfloor N/2 \rfloor}}\), each automorphism \(\phi \in \Aut(\bT^N) \cong \GL(N,\bZ)\) is implementable, viz, by \((\phi^\ast,\phi^\ast \otimes \id)\), so that by Proposition~\ref{nonequivariant}, for all \(\phi \in \GL(N,\bZ)\) and \(\theta \in \Skew(N,\bZ)\),
\begin{equation}\label{nonequivariantex}
	(C^\infty(\bT^N_{\phi \theta \phi^T }),L^2(\bT^N_{\phi \theta \phi^T},\slashed{S}),\slashed{D}) \cong (C^\infty(\bT^N_\theta),L^2(\bT^N_\theta,\slashed{S}),(\phi^\ast \otimes \id)\slashed{D}(\phi^\ast \otimes \id)^\ast)) 
\end{equation}
since \(p(\phi \theta \phi^T) = \hat{\phi}^\ast p(\theta)\).

Now, one can construct an action of the discrete group \(\SO(N,N\vert\bZ)\) on a certain dense subset \(\cT^0_N \subset \Skew(N,\bR) \cong \wedge^2 \bR^N\) of second category, such that the subgroups \(\GL(N,\bZ)\) and \(\Skew(N,\bR)\) of \(\SO(N,N \mid \bZ)\) act by
\[
	\forall \theta \in \cT^0_N, \; \forall \nu \in \Skew(N,\bR), \; \forall R \in \GL(N,\bZ), \quad \nu \cdot \theta \coloneqq \theta + \nu, \quad R \cdot \theta \coloneqq R \theta R^T.
\]
Rieffel and Schwarz~\cite{RS} constructed, for every \(\theta \in \cT^0_N\) and \(g \in \SO(N,N\vert\bZ)\), a Morita--Rieffel equivalence of the algebras \(C^\infty(\bT^N_\theta)\) and \(C^\infty(\bT^N_{g \cdot \theta})\), which is a \Star-isomorphism whenever \(g\) is in the subgroup generated by \(\GL(N,\bZ)\) and \(\Skew(N,\bZ)\); Elliott and Li~\cite{EL} then showed that for \(\theta\), \(\p{\theta} \in \cT^0_N\), \(C^\infty(\bT^N_\theta)\) and \(C^\infty(\bT^N_{\p{\theta}})\) are Morita--Rieffel equivalent  if and only if \(\p{\theta} = g \cdot \theta\) for some \(g \in \SO(N,N\vert\bZ)\). On the one hand, since \(\ker p = \Skew(N,\bZ)\),
\[
	\forall \theta \in \Skew(N,\bR), \; \forall \nu \in \Skew(N,\bZ), \quad p(\nu \cdot \theta) = p(\theta+\nu) = p(\theta),
\]
so that for all \(\theta \in \Skew(N,\bR)\) and \(\nu \in \Skew(N,\bZ)\), the \Star-isomorphism \(C^\infty(\bT^N_{\nu \cdot \theta}) \cong C^\infty(\bT^N_\theta)\) is recovered by the \G-equivalence of Equation~\ref{equivariantex}. On the other hand,
\[
	\forall \theta \in \Skew(N,\bR), \; \forall R \in \GL(N,\bZ) \cong \Aut(\bT^N), \quad p(R \cdot \theta) = p(R \theta R^T) = \hat{R}^\ast p(\theta), 
\]
so that for all \(\theta \in \Skew(N,\bR)\) and \(R \in \GL(N,\bZ)\), the \Star-isomorphism \(C^\infty(\bT^N_{R \cdot \theta}) \cong C^\infty(\bT^N_\theta)\) is recovered by the non-equivariant unitary equivalence of Equation~\ref{nonequivariantex}.
\end{example}

\section{Connes--Landi deformations of commutative spectral triples}\label{sec:3}

At last, we turn to our main goal, the statement and proof of an extension of Connes's reconstruction theorem~\cite{Con13} to Connes--Landi deformations of \G-equivar\-i\-ant commutative spectral triples. In this section, given a Fr\'{e}chet \(G\)-\(\ast\)-algebra \(\cA\), let \[\cA^\fin \coloneqq \bigoplus_{\bx \in \dual{G}}^{\mathrm{alg}} \cA_\bx.\]

\subsection{A reconstruction theorem}

We begin by recalling the statement of the reconstruction theorem for commutative spectral triples, which one can interpret as a partial analogue for spectral triples of Gel'fand--Na{\u\i}mark duality.

First, recall that a \emph{Dirac-type operator} on a compact oriented Riemannian manifold \((X,g)\) is a first order differential operator \(D\) on a Hermitian vector bundle \(E \to X\), such that \([D,f]^2 = -g(\dif f,\dif f) \id_E\) for all \(f \in C^\infty(X)\), or equivalently, such that
\(
	D^2 = - \sum_{i,j} g^{ij} \tfrac{\partial}{\partial x^i} \tfrac{\partial}{\partial x^j} + \text{lower order terms}
\)
in local coordinates. In particular, if \(X\) is a compact oriented Riemannian manifold and \(D\) is an essentially self-adjoint Dirac-type operator on a Hermitian vector bundle \(E \to X\), then
\((C^\infty(X),L^2(X,E),D)\) defines a regular spectral triple with commutative algebra. One can check that every such spectral triple satisfies the following definition.

\begin{definition}[Connes~\cite{Con96,Con13}]
Let \((\cA,H,D)\) be a regular spectral triple, let \(p \in \bN\). One calls \((\cA,H,D)\) a \(p\)-dimensional \emph{commutative} spectral triple if \(\cA\) is commutative and if the following conditions all hold:
\begin{enumerate}
	\item \emph{Dimension}: The eigenvalues \(\set{\lambda_n}_{n \in \bN}\) of \((D^2+1)^{-1/2}\), counted with multiplicity and arranged in decreasing order, satisfy \(\lambda_n = O(n^{-1/p})\) as \(n \to +\infty\).
	\item \emph{Order one}: For all \(a\), \(b\ \in \cA\), \([[D,L(a)],L(b)] = 0\).
%	%\item \emph{Regularity}: For all \(a \in \cA\), \(a\), \([D,a] \in \cap_k \Dom (\ad \ds{D})^k\).
	\item \emph{Orientability}: Define \(\epsilon : \cA^{\otimes(p+1)} \to \cA^{\otimes(p+1)}\) by setting
\[
	 \forall a_0, a_1, \dotsc, a_p \in \cA, \quad \epsilon(a_0 \otimes a_1 \otimes \cdots \otimes a_p) \coloneqq  \frac{1}{p!} \sum_{\pi \in S_p} (-1)^\pi  a_0 \otimes a_{\pi(1)} \otimes \cdots \otimes a_{\pi(p)},
\]
and call \(\bc \in \cA^{\otimes(p+1)}\) \emph{antisymmetric} if \(\epsilon(\bc) = \bc\). Define \(\pi_D : \cA^{\otimes(p+1)} \to B(H)\) by \(\pi_D(a_0 \otimes a_1 \otimes \cdots \otimes a_p) \coloneqq L(a_0)[D,L(a_1)] \cdots [D,L(a_p)]\) for all \(a_0,a_1,\dotsc,a_p \in \cA\). There exists an antisymmetric \(\bc \in \cA^{\otimes(p+1)}\), such that \(\chi := \pi_D(\bc)\) is a self-adjoint unitary, satisfying
\[
	\forall a \in \cA, \quad L(a)\chi = \chi L(a), \quad [D,L(a)] \chi = (-1)^{p+1} \chi [D,L(a)].
\]
	\item \emph{Finiteness and absolute continuity}: The subspace \(\cH_\infty := \cap_k \Dom \abs{D}^k\) defines a finitely generated projective \(\cA\)-module admitting a Hermitian metric \(\hp{}{}\), such that \(\ip{\xi}{\eta} = \fint L\left(\hp{\xi}{\eta}\right)\) for all \(\xi\), \(\eta \in \cH_\infty\),  where \(\fint : B(H) \ni T \mapsto \Tr_\omega(T (D^2+1)^{-p/2})\) for \(\Tr_\omega\) a fixed Dixmier trace.
	\item \emph{Strong regularity}: For all \(T \in \End_{\cA}(\cH_\infty)\),  \(T \in \cap_k \Dom[\abs{D},\cdot]^k\).
\end{enumerate}
\end{definition}

\begin{remark}
Connes's own definition of commutative spectral triple and statement of the reconstruction theorem involves a stronger orientability condition, modelled on the case of twisted spin\(^\bC\) Dirac operators on compact spin\(^\bC\) manifolds. However, one can show that the condition given above is sufficient~\cite[Proof of Corollary 2.19]{Ca12}.
\end{remark}

The reconstruction theorem, which Connes conjectured in 1996~\cite{Con96} and then proved in 2008~\cite{Con13}, after a substantial attempt by Rennie and V\'arilly~\cite{RV} in 2006, can now be stated as follows.

\begin{theorem}[{Connes~\cite[Theorem 1.1]{Con13}}]
Let \((\cA,H,D)\) be a \(p\)-dimensional commutative spectral triple. There exist a compact oriented Riemannian \(p\)-manifold \(X\), a Hermitian vector bundle \(E \to X\), and an essentially self-adjoint Dirac-type operator \(D_E\) on \(E\), such that
\(
	(\cA,H,D) \cong (C^\infty(X),L^2(X,E),D_E).
\)
\end{theorem}

\begin{remark}\label{reconunique}
The reconstructed differential-geometric data \((X,E,D_E)\) is unique up to orientation-preserving isometric diffeomorphisms covered by unitary bundle morphisms intertwining the Dirac-type operators. Indeed, suppose we are given a unitary equivalence \((\psi,\Psi) : (C^\infty(X_1),L^2(X_1,E_1),D_1) \iso (C^\infty(X_2),L^2(X_2,E_2),D_2)\) of concrete commutative spectral triples. By a result of Mr\v{c}un~\cite{Mrcun}, there exists a unique diffeomorphism \(\phi : X_1 \to X_2\) such that \(\psi(f) = f \circ \inv{\phi}\) for all \(f \in C^\infty(X_1)\); since \(\Psi D_1 \Psi^\ast = D_2\), it follows, by the Serre--Swan theorem, that \(\Psi\) is induced by a unitary bundle endomorphism \(\Phi : E_1 \to E_2\) covering \(\phi\) and intertwining \(D_1\) and \(D_2\), and hence, that \(\phi\) is isometric and orientation-preserving.
\end{remark}

Now, once more, let \(G\) be a fixed compact Abelian Lie group. We define a \emph{compact oriented Riemannian \G-manifold} to be a compact oriented Riemannian manifold \(X\) together with a smooth action \(\sigma : G \to \operatorname{Diff}(X)\) by orientation-preserving isometries. If \((X,\sigma)\) is a compact oriented Riemannian \G-manifold, then we define a \emph{G-equivariant Hermitian vector bundle} over \(X\) to be a Hermitian vector bundle \(E \xrightarrow X\) together with a smooth action \(U : G \to \operatorname{Diff}(E)\) by fibrewise unitaries, such that \(p \circ U_t = U_t \circ p\) for all \(t \in G\). Finally, if \((E,U) \to (X,\sigma)\) is a \G-equivariant Hermitian vector bundle over a compact oriented Riemannian \G-manifold, then we define a \emph{G-invariant Dirac-type operator} on \(E\) to be a Dirac-type operator on \(E\) such that \(U_t D U_t^\ast = D\) for all \(t \in G\). 

Let \(X\) be a compact oriented Riemannian \G-manifold, let \(E \to X\) be a \G-equivariant Hermitian vector bundle, and let \(D\) be an essentially self-adjoint \G-invariant Dirac-type operator on \(E\). Then, in particular, \((C^\infty(X),L^2(X,E),D)\) defines a \G-equivariant concrete commutative spectral triple, to which one can apply Connes--Landi deformation---in the literature, the resulting spectral triples are often called \emph{toric noncommutative manifolds}~\cite{VSthesis}. We propose the following abstract definition for the Connes--Landi deformation of a \G-equivariant spectral triple by a cohomology class \(\theta \in H^2(\dual{G},\bT)\).

\begin{definition}\label{def:thetacomm}
Let \((\cA,H,D)\) be a \(G\)-equivariant regular spectral triple, let \(\theta \in H^2(\dual{G},\bT)\), and let \(p \in \bN\). We shall call \((\cA,H,D)\) a \(p\)-dimensional \emph{\(\theta\)-commutative} spectral triple if the following conditions all hold:
\begin{enumerate}\setcounter{enumi}{-1}
	\item \emph{Order zero}: The algebra \(\cA\) is \emph{\(\theta\)-commutative}, viz, \[\forall\bx,\by\in\dual{G}, \; \forall a_\bx \in \cA_\bx, \; \forall b_\by \in \cA_\by, \quad b_\by a_\bx = \e{\iot{\theta}(\bx,\by)} a_\bx b_\by,\] so that \(R \coloneqq \pi_{-\iot{\theta}} :  \cA^\op = \cA_{-\iot{\theta}} \to B(H)\) defines a \G-equivariant \Star-homomorphism, such that for all \(a\), \(b \in \cA\), \([L(a),R(b)] = 0\).
	\item \emph{Dimension}: The eigenvalues \(\set{\lambda_n}_{n \in \bN}\) of \((D^2+1)^{-1/2}\), counted with multiplicity and arranged in decreasing order, satisfy \(\lambda_n = O(n^{-1/p})\) as \(n \to +\infty\).
	\item \emph{Order one}: For all \(a\), \(b\ \in \cA\), \([[D,L(a)],R(b)] = 0\).
	%\item \emph{Regularity}: For all \(a \in \cA\), \(a\), \([D,a] \in \cap_k \Dom (\ad \ds{D})^k\).
	\item\label{item:orient} \emph{Orientability}: Define \(\epsilon_\theta : (\cA^\fin)^{\otimes(p+1)} \to (\cA^\fin)^{\otimes(p+1)}\) by
\begin{multline*}
	\epsilon_\theta(a_0 \otimes a_1 \otimes \cdots \otimes a_p)\\ \coloneqq \frac{1}{p!} \sum_{\pi \in S_p} \exp\left(2\pi i \mkern-18mu \sum_{\substack{i < j\\ \pi(i)>\pi(j)}} \mkern-18mu \iot{\theta}(\bx_{\pi(i)},\bx_{\pi(j)})\right)(-1)^\pi a_0 \otimes a_{\pi(1)} \otimes \cdots \otimes a_{\pi(p)}
\end{multline*}
for isotypic elements \(a_0 \in \cA_{\bx_0}\), \ldots, \(a_p \in \cA_{\bx_p}\) of \(\cA\), and say that \(\bc \in (\cA^\fin)^{\otimes(p+1)}\) is \emph{\(\theta\)-antisymmetric} if \(\epsilon_\theta(\bc) = \bc\). Define \(\pi_D : \cA^{\otimes(p+1)} \to B(H)\) by
\[
	\forall a_0,\; a_1, \dotsc, a_p \in \cA, \quad \pi_D(a_0 \otimes a_1 \otimes \cdots \otimes a_p) \coloneqq L(a_0)[D,L(a_1)] \cdots [D,L(a_p)].
\]
There exists a \(G\)-invariant \(\theta\)-antisymmetric \(\bc \in (\cA^\fin)^{\otimes(p+1)}\), such that \(\chi := \pi_D(\bc)\) is a self-adjoint unitary, satisfying
\[
	\forall a \in \cA, \quad L(a)\chi = \chi L(a), \quad [D,L(a)] \chi = (-1)^{p+1} \chi [D,L(a)].
\]
	\item \emph{Finiteness and absolute continuity}: The subspace \(\cH_\infty := \cap_k \Dom \abs{D}^k\) defines a \(G\)-equivariant finitely-generated projective\ right \(\cA\)-module admitting a \(G\)-equivariant Hermitian metric \(\hp{}{}\), such that \(\ip{\xi}{\eta} = \fint R\left(\hp{\xi}{\eta}\right)\) for all \(\xi\), \(\eta \in \cH_\infty\), where \(\fint : B(H) \ni T \mapsto \Tr_\omega(T (D^2+1)^{-p/2})\) for \(\Tr_\omega\) a fixed Dixmier trace.
	\item \emph{Strong regularity}: For all \(T \in \End_{\cA^\op}(\cH_\infty)\),  \(T \in \cap_k \Dom ([\abs{D},\cdot]^k)\).
\end{enumerate}
\end{definition}

\begin{remark}
By regularity, \(\cH_\infty\) is necessarily a left \(\cA\)-module, and hence, by the proof of Proposition~\ref{yamashita}, defines a right \(\cA\)-module via \(R \coloneqq \pi_{-\iot{\theta}}(\cA): \cA^\op \to B(H)\), so that the finiteness, absolute continuity, and strong regularity conditions are well-posed. The importance of such a bimodule structure for noncommutative Riemannian geometry via spectral triples, independent of any real structure, has already been observed by Lord, Rennie, and V\'arilly~\cite{LRV}.
\end{remark}

\begin{remark}\label{rmk:0commut}
The cohomology class \(\theta\) completely governs the noncommutativity of a \(\theta\)-commutative spectral triple; in particular, a \(0\)-commutative spectral triple is precisely a \(G\)-equivariant commutative spectral triple. \emph{See Proposition~\ref{prop:appx} in the appendix for a proof of the only non-trivial point, which concerns orientability.}
\end{remark}

\begin{remark}
This definition is robust with respect to \G-equivalence, for if a \G-equivariant spectral triple \((\cA,H,D)\) is \G-equivalent to a \(p\)-dimensional \(\theta\)-commu\-ta\-tive spectral triple, then \((\cA,H,D)\) itself defines a \(p\)-dimensional \(\theta\)-commutative spectral triple.
\end{remark}

\begin{remark}
By \(\theta\)-commutativity of \(\cA\), the standard proof that an antisymmetric Hoch\-schild \(p\)-chain for a commutative algebra yields a Hochschild \(p\)-cycle~\cite[Proposition 8.10]{GBVF}, \emph{mutatis mutandis}, shows that a \(\theta\)-antisymmetric Hochschild \(p\)-chain \(\bc \in (\cA^\fin)^{\otimes(p+1)}\) necessarily defines a Hochschild \(p\)-cycle for \(\cA\).
\end{remark}

Our main technical result is the following theorem, which simultaneously guarantees that our abstract definition describes Connes--Landi deformations of \G-equivariant commutative spectral triples, and facilitates the extension of Connes's reconstruction theorem to \(\theta\)-commutative spectral triples; we defer its proof to \S \ref{proofsection}.

\begin{theorem}\label{deform}
Let \((\cA,H,D)\) be a \(p\)-dimensional \(\theta_0\)-commutative spectral triple and let \(\theta \in H^2(\dual{G},\bT)\). Then \((\cA_\theta,H,D)\) defines a \(p\)-dimensional \((\theta_0+\theta)\)-commu\-tat\-ive spectral triple.
\end{theorem}

Since a \G-equivariant concrete commutative spectral triple is \(0\)-commutative, we immediately recover the basic results of Connes and Landi and of Connes and Dubois-Violette on noncommutative toric manifolds, except phrased in terms of our definition, instead of the notion of noncommutative spin geometry.

\begin{corollary}[{Connes--Landi~\cite[Theorem 6]{CL}, Connes--Dubois-Violette~\cite[Theorem 9]{CDV}}]
Let \(X\) be a compact \(p\)-dimensional oriented Riemannian \(G\)-manifold, let \(E \to X\) be a \G-equivariant Hermitian vector bundle, and let \(D\) be an essentially self-adjoint \G-invariant Dirac-type operator on \(E\). For any \(\theta \in H^2(\dual{G},\bT)\), \((C^\infty(X)_\theta,L^2(X,E),D)\) defines a \(p\)-dimensional \(\theta\)-commu\-ta\-tive spectral triple.
\end{corollary}

\begin{example}[{Gabriel--Grensing~\cite[Theorem 5.5]{GG}, cf.\ Rieffel~\cite[\S 4]{Rieffel98}}]
Let \(G\) be a compact Lie group, let \((A,\alpha)\) be a unital \GCstar-algebra such that \(\alpha\) is faithful and ergodic (i.e., \(A^G = \bC\)), and let \(\tau : A \to \bC\) be the unique \G-invariant trace. Gabriel and Grensing~\cite{GG}, following Rieffel~\cite{Rieffel98}, construct a spectral triple for \(A\) as follows. Let
\(
	\cA  \coloneqq \set{a \in A \given G \ni t \mapsto \alpha_t(S) \; \text{is norm-smooth}}
\)
and let \((H,U)\) be a covariant \Star-representation of \((A,\alpha)\), such that
\(
	\cH \coloneqq \set{\xi \in H \given G \ni t \mapsto U_t\xi \; \text{is norm-smooth}}
\)
defines a left Hermitian f.g.p.\ \G-\(\cA\)-module, satisfying \(\ip{\xi}{\eta} = \tau(\hp{\xi}{\eta})\) for all \(\xi\), \(\eta \in \cH\). Let \(\set{X_1,\dotsc,X_N}\) be an orthonormal basis for \(\fg \coloneqq \Lie(G)\), and let \(c : \operatorname{Cl}(\fg) \to B(\slashed{S}(\fg))\) be the spinor representation of the Clifford algebra \(\operatorname{Cl}(\fg)\) of \(\fg\), with the convention that \(c(X)c(Y) + c(Y)c(X) = -2\ip{X}{Y}\) and \(c(X)^\ast = -c(X)\) for all \(X\), \(Y \in \fg\). Finally, let \(D_{A,H} \coloneqq \sum_{k=1}^N (\alpha_\ast)_{X_k} \otimes c(X_k)\). Then \((\cA,H \otimes \slashed{S}(\fg),D_{A,H})\) defines an \(N^+\)-summable \G-equivariant spectral triple.

Now, suppose that \(G\) is Abelian. By results of Albeverio and H{\o}egh-Krohn~\cite{AHK} and of Olesen, Pedersen, and Takesaki~\cite{OPT}, there exists a unique class \(\theta \in H^2(\dual{G},\bT)\) such that \(\cA \cong C^\infty(G)_\theta\). The explicit construction above, together with the standard construction of the orientation cycle on noncommutative tori~\cite[Lemmata 12.5--6]{GBVF}, implies that \((\cA,H \otimes \slashed{S}(\fg),D_{A,H})\) is indeed an \(N\)-dimensional \(\theta\)-commutative spectral triple. In fact, using the results of \S~\ref{moduledeformsec}, one can check that
\(
	(\cA,H \otimes \slashed{S}(\fg),D_{A,H}) \cong_G (C^\infty(G)_\theta,L^2(G,\slashed{S} \otimes V_H),\slashed{D} \otimes \id),
\)
where \(\slashed{S} \coloneqq G \times \slashed{S}(\fg)\) is the spinor bundle and \(\slashed{D}\) is the spin Dirac operator  corresponding to the trivial spin structure on \(G\), whilst \(V_H\) is the finite-dimensional unitary representation of \(G\) such that \(\cH \cong (C^\infty(G) \otimes V_H)_{\Theta}\) as left Hermitian f.g.p.\ \GA-modules.
\end{example}

Finally, suppose that \((\cA,H,D)\) is a \(p\)-dimensional \(\theta\)-commutative spectral triple. Then, by Theorem~\ref{deform}, \((\cA_{-\theta},H,D)\) is now a \(p\)-dimensional \(0\)-commutative spectral triple, i.e., a \G-equivariant \(p\)-dimensional \emph{commutative} spectral triple. Hence, we may apply Connes's reconstruction theorem for commutative spectral triples to yield its extension to Connes--Landi deformations of commutative spectral triples.

\begin{theorem}[{cf.\ Connes~\cite[Theorem 1.1]{Con13}}]\label{thetareconstruct}
Let \((\cA,H,D)\) be a \(p\)-dimensional \(\theta\)-commutative spectral triple. There exist a compact \(p\)-dimensional oriented Riemannian \(G\)-manifold \(X\), a \(G\)-equivariant Hermitian vector bundle \(E \to X\), and an essentially self-adjoint \G-invariant Dirac-type operator \(D_E\) on \(E\) such that
\[
	(\cA,H,D) \cong_G (C^\infty(X)_\theta, L^2(X,E),D_E).
\]
\end{theorem}

\begin{remark}
As it stands, the class \(\theta \in H^2(\dual{G},\bT)\) is necessarily part of the data of a \(\theta\)-commutative spectral triple \((\cA,H,D)\) \emph{qua} abstract Connes--Landi deformation of a \G-equviariant commutative spectral triple. In general, it need not be unique, and it is only in the case where \(G\) acts ergodically on \(\cA\), in the sense that \(\cA^G = \bC\), that one can apply the classification results of Albeverio and H{\o}egh-Krohn~\cite{AHK} and of Olesen, Pedersen, and Takesaki~\cite{OPT} for ergodic \(W^\ast\)- and \Cstar-dynamical systems to extract a candidate deformation parameter \(\theta \in H^2(\dual{G},\bT)\).
\end{remark}

\subsection{The proof of Theorem~\ref{deform}}\label{proofsection}

We now give a proof of Theorem~\ref{deform}. Let \((\cA,H,D)\) be a \(p\)-dimensional \(\theta_0\)-commutative spectral triple and let \(\theta \in H^2(\dual{G},\bT)\). By Theorems~\ref{kleppner} and~\ref{equivariant}, let \(\Theta \in B(\dual{G})\) be a bicharacter representing \(\theta\). The claim is that \((\cA_\Theta,H,D)\) defines a \(p\)-dimensional \((\theta_0+\theta)\)-commutative spectral triple. 

By Theorem~\ref{yamashita}, we know that \((\cA_\Theta,H,D)\) defines a \G-equivariant spectral triple; in particular, since the operator \(D\) is untouched, the deformed spectral triple \((\cA_\Theta,H,D)\) trivially satisfies the same dimension condition as \((\cA,H,D)\). Thus, it remains to show that \((\cA_\Theta,H,D)\) satisfies the order zero, order one, orientability, finiteness and absolute continuity, and strong regularity conditions for a \((\theta_0 + \theta)\)-commutative spectral triple. For convenience, let \(L_\Theta \coloneqq \pi_\Theta \circ L : \cA_\Theta \to B(H)\) denote the left \Star-representation of \(\cA_\Theta\) on \(H\) defining \((\cA_\Theta,H,D)\).

\subsubsection{Order zero; order one}

First, since \(\cA\) is \(\theta_0\)-commutative and
\begin{equation}\label{opposites}
 \forall \bx,\;\by\in\dual{G},\quad -\Theta(\by,\bx) + \iot{\theta_0}(\bx,\by) = \iot{\theta_0+\theta}(\bx,\by) - \Theta(\bx,\by),
\end{equation}
it follows that \(b_\by \star_\Theta a_\bx = \e{\iot{\theta_0+\theta}(\bx,\by)}a_\bx \star_\Theta b_{\by}\) for all isotypic elements \(a_\bx \in \cA_\bx\), \(b_\by \in \cA_\by\) of \(\cA\), i.e., that \(\cA_{\Theta}\) is \((\theta_0+\theta)\)-commutative. Together with Equation~\ref{opposites}, this implies that the opposite algebra \((\cA_\Theta)^\op\) of \(\cA_\Theta\) is precisely \((\cA_\Theta)_{-\iot{\theta}} = \cA_{\Theta^t}\), which admits the \G-equivariant \Star-representation \(R_\Theta : (\cA_\Theta)^\op \to B(H)\) defined by
\[
	R_\Theta \coloneqq \pi_{-\iot{\theta_0+\theta}} \circ L_\Theta = \pi_{-\iot{\theta_0+\theta}+\Theta} \circ L = \pi_{\Theta^t -\iot{\theta_0}} \circ L = \pi_{\Theta^t} \circ R.
\]

Now, let \(a_\bx \in \cA_\bx\) and \(b_\by \in \cA_\by\) be isotypic elements of \(\cA\). On the one hand, recall that \([L(a_\bx),R(b_\by)] = 0\) by the order zero condition on \((\cA,H,D)\), and on the other hand, observe that
\[
	L_\Theta(a_\bx) = \pi_\Theta( L(a_\bx)) = L(a_\bx) U_{-\Theta(\bx,\cdot)}, \quad R_\Theta(b_\by) = \pi_{\Theta^t}(R(b_\by)) = R(b_\by)U_{-\Theta^t(\bx,\cdot)}
\]
since \(\Theta\) was chosen to be a bicharacter. Thus,
\begin{align*}
	&L_\Theta(a_\bx) R_\Theta(b_\by) - R_\Theta(b_\by) L_\Theta(a_\bx)\\
= &\e{\Theta(\bx,\by)} L(a_\bx)R(b_\by) U_{-\Theta(\bx,\cdot)-\Theta^T(\by,\cdot)} - \e{\Theta^T(\by,\bx)} R(b_\by) L(a_\bx) U_{-\Theta^T(\by,\cdot)-\Theta(\bx,\cdot)}
\\	=&0,
\end{align*}
so that the order zero condition is satisfied. Similarly, since \([[D,L(a_\bx)],R(b_\by)]=0\) by the order one condition on \((\cA,H,D)\), and,  by the proof of Theorem~\ref{yamashita}, \[[D,L_\Theta(a_\bx)] = \pi_\Theta([D,L(a_\bx)]) = [D,L(a_\bx)]U_{-\Theta(\bx,\cdot)},\] it follows that \([[D,L_\Theta(a_\bx)],R_\Theta(b_\by)] = 0\), so that the order one condition is satisfied.

\subsubsection{Orientability}\label{ssec:orientability}
For convenience, let the maps \(\pi_D : \cA^{\otimes(p+1)} \to B(H)\) and  \(\pi_{\Theta,D} : \cA_\Theta^{\otimes(p+1)} \to B(H)\) be given by
\begin{gather*}
	\pi_D(a_0 \otimes a_1 \otimes \cdots \otimes a_p) \coloneqq L(a_0)[D,L(a_1)] \cdots [D,L(a_p)],\\
	\pi_{\Theta,D}(a_0 \otimes a_1 \otimes \cdots \otimes a_p) \coloneqq L_\Theta(a_0)[D,L_\Theta(a_1)] \cdots [D,L_\Theta(a_p)],
\end{gather*}
for all \(a_0,\dotsc,a_p \in \cA\). Since \(\chi = \pi_D(\bc)\) is \G-invariant, we have \(\pi_\Theta(S)\chi = \pi_\Theta(S\chi)\) and \(\chi\pi_\Theta(S) = \pi_\Theta(\chi S)\) for all \(S \in B^\infty(H)\), and hence, by the proof of Theorem~\ref{yamashita},
\[
	\forall a \in \cA, \quad L_\Theta(a)\chi = \chi L_\Theta(a), \quad [D,L_\Theta(a)]\chi = (-1)^{p+1}\chi[D,L_\Theta(a)].
\]
Thus, it suffices to construct a \(G\)-invariant \(\iota(\theta_0+\theta)\)-antisymmetric \(\bc_\Theta \in (\cA_\Theta^\fin)^{\otimes (p+1)}\) such that \(\pi_{\Theta,D}(\bc_\Theta) = \pi_D(\bc) \eqqcolon \chi\). The key to doing so will be the following lemma.

\begin{lemma}\label{voldeform}
Define \( \Theta_\ast : (\cA^\fin)^{\otimes(p+1)} \to (\cA_\Theta^\fin)^{\otimes(p+1)}\) by setting
\[
	\Theta_\ast(a_0 \otimes a_1 \otimes \cdots \otimes a_p) \coloneqq \e{\sum_{i < j} \Theta(\bx_i,\bx_j)} a_0 \otimes a_1 \otimes \cdots \otimes a_p
\]
for isotypic \(a_0 \in \cA_{\bx_0},\dotsc,a_p \in \cA_{\bx_p}\). Then \(\Theta_\ast\) is \(G\)-equivariant and satisfies
\[
\epsilon_{\theta_0+\theta} \circ \Theta_\ast = \Theta_\ast \circ \epsilon_{\theta_0}, \quad
\rest{\pi_D}{((\cA^\fin)^{\otimes (p+1)})^G} = \pi_{\Theta,D} \circ \Theta_\ast\vert_{((\cA_\Theta^\fin)^{\otimes (p+1)})^G}.
\]
\end{lemma}

\begin{proof}
The map \(\Theta_\ast\) is \(G\)-equivariant by construction, so let us turn to the rest of the claim. We first check that \(\epsilon_{\theta_0+\theta} \circ \Theta_\ast = \Theta_\ast \circ \epsilon_{\theta_0}\). Without loss of generality, let
\(
	s := a_0 \otimes a_1 \otimes \cdots \otimes a_p
\)
for isotypic elements \(a_0 \in \cA_{\bx_0}\), \ldots, \(a_p \in \cA_{\bx_p}\) of \(\cA\), since finite linear combinations of such tensors span \((\cA^\fin)^{\otimes(p+1)}\). Then, by Equation~\ref{opposites},
\begin{multline*}
	 - \sum_{\substack{1 \leq i<j\\ \pi(i)>\pi(j)}} \iot{\theta_0+\theta}(\bx_{i},\by_{j}) + \sum_{i<j} \Theta(\bx_i,\bx_j)\\ = \sum_{j \geq 1} \Theta(\bx_0,\bx_{\pi(j)}) + \sum_{\substack{i,j \geq 1\\ \pi(i)<\pi(j)}}\Theta(\bx_{\pi(i)},\bx_{\pi(j)}) - \sum_{\substack{i<j\\ \pi(i)>\pi(j)}} \iot{\theta_0}(\bx_i,\bx_j)
\end{multline*}
for all permutations \(\pi \in S_p\), so that, indeed, \(\epsilon_{\theta_0+\theta}(\Theta_\ast s) = \Theta_\ast(\epsilon_{\theta_0}(s))\). 

Now, let us check that \(\rest{\pi_D}{((\cA^\fin)^{\otimes (p+1)})^G} = \pi_{\Theta,D} \circ \Theta_\ast\vert_{((\cA_\Theta^{\fin})^{\otimes (p+1)})^G}\). Without loss of generality, let
\(
	s := a_0 \otimes a_1 \otimes \cdots \otimes a_p
\)
for isotypic elements \(a_0 \in \cA_{\bx_0}\), \ldots, \(a_p \in \cA_{\bx_p}\) of \(\cA\) with \(\bx_0 + \bx_1 + \cdots + \bx_p = \bo\), since finite linear combinations of such tensors span \(((\cA^\fin)^{\otimes(p+1)})^G\). Since \(s\) is \G-invariant, so too is \(\pi_D(s)\), so that
\[
	\pi_{\Theta,D}(s) =  \pi_\Theta\left(\e{-\sum_i \Theta\left(\bx_i,\sum_{i<j}\bx_j\right)}\pi_D(s)\right) = \e{-\sum_{i<j} \Theta(\bx_i,\bx_j)}\pi_D(s),
\]
and hence, \(\pi_{\Theta,D}(\Theta_\ast s) = \e{\sum_{i<j}\Theta(\bx_i,\bx_j)}\pi_{\Theta,D}(s)= \pi_D(s)\).
\end{proof}

Now, set \(\bc_\Theta \coloneqq \Theta_\ast \bc\). By Lemma~\ref{voldeform}, since \(\bc\) is \(G\)-invariant, so too is \(\bc_\Theta\), and since \(\bc\) is \(\theta_0\)-antisymmetric,
\(
	\epsilon_{\theta_0+\theta}(\bc_\Theta) = \epsilon_{\theta_0+\theta}(\Theta_\ast(\bc)) = \Theta_\ast(\epsilon_{\theta_0}(\bc)) = \Theta_\ast(\bc) = \bc_\Theta,
\)
so that \(\bc_\Theta\) is \((\theta_0+\theta)\)-antisymmetric, and
\(
 \pi_{\Theta,D}(\bc_\theta) = \pi_{\sigma,D}(\Theta_\ast \bc) = \pi_D(\bc) = \chi.
\)

\begin{remark}
This construction can be viewed as the distillation of an argument of Connes and Dubois--Violette~\cite[Proposition 3, Proof of Theorem 9]{CDV} concerning deformed differential calculi and Hochschild homology. In addition, it recovers the standard construction of orientation cycles on noncommutative \(n\)-tori from the usual translation-invariant volume form on the flat \(n\)-torus~\cite[pp.\ 546--7]{GBVF}.
\end{remark}

\subsubsection{Finiteness and absolute continuity}

Let us first consider finiteness. By finiteness and absolute continuity for \((\cA,H,D)\), we know that \(\cH_\infty \coloneqq \cap_k \Dom \abs{D}^k\) defines a right Hermitian f.g.p\ \G-\(\cA\)-module, and so, by Proposition~\ref{fgpmod}, defines a Hermitian Fr\'echet \GA-module for the Fr\'echet topology induced by any inclusion of \(\cH_\infty\) as a complementable subspace of \(\cA^N\) for some \(N \in \bN\). On the other hand, the dense subspace \(\cH_\infty\) \emph{qua} smooth domain of \(D\) admits the Fr\'echet topology defined by the Sobolev seminorms \(\norm{}_k\), where \(\norm{\xi}_k \coloneqq \norm{\abs{D}^k\xi} = \sqrt{\ip{\xi}{D^{2k}\xi}}\) for all \(k \in \bN\) and \(\xi \in \cH_\infty\). To show that \((\cA_\Theta,H,D)\) satisfies finiteness, it suffices to show that these topologies on \(\cH_\infty\) coincide, since in that case, one can easily check on isotypic elements that \(\cH_\infty\) endowed with the right \Star-representation
\[
	R_\Theta \coloneqq \pi_{-\iot{\theta_0+\theta}} \circ L_\Theta = \pi_{\Theta^t} \circ R : (\cA_\Theta)^\op \to B(H)
\]
of \(\cA_\Theta\) is precisely the deformation \((\cH_\infty)_\Theta\) of \(\cH_\infty\) \emph{qua} Hermitian Fr\'echet \GA-module, and hence, by Theorem~\ref{moduledeform}, defines a Hermitian f.g.p.\ \G-\(\cA_\Theta\)-module.

\begin{lemma}[{Connes \cite[Proposition\ 2.3.(5)--(6)]{Con13}}]
There exists a bicontinuous \G-equivariant isomorphism \(\psi : \cH_\infty \iso e(V \otimes \cA)\) of Hermitian Fr\'echet \G-\(\cA\)-modules, where \(V\) is a finite-dimensional unitary representation of \(G\) and \(e \in (B(V) \otimes \cA)^G\) is an orthogonal projection. In particular, the maps
\[
	\cH_\infty \times \cA \to \cH_\infty, \quad (\xi,a) \mapsto R(a)\xi; \quad 
	\cH_\infty \times \cH_\infty \to \cA, \quad (\xi,\eta) \mapsto \hp{\xi}{\eta};
\]
are jointly continuous, and the \G-action on \(\cH_\infty\) is strongly smooth.
\end{lemma}

\begin{proof}
By Proposition~\ref{fgpmod}, let \(\psi : \cH_\infty \iso e(V \otimes \cA)\) be an algebraic \G-equivariant isomorphism of the desired form; we merely apply the proof of \cite[Proposition 2.3.(5)--(6)]{Con13}, \emph{mutatis mutandis}. Since the multiplication on \(\cA\) defines a jointly continuous map \(\cA \times \cA \to \cA\), \(e(V\otimes \cA)\) is a closed subspace of \(V \otimes \cA\), and thus is complete. Now, define \(\phi : V \otimes \cA \surj \cH_\infty\) by
\[
	\forall v \in V, \; \forall a \in \cA, \quad \phi(v \otimes a) := \inv{\psi}(e(v \otimes a)) = R(a)\inv{\psi}(e(v \otimes 1)),
\]
and observe that \(R(\cA) = \pi_{\iot{\theta}}(L(\cA)) \subset \cap_k \Dom([\abs{D},\cdot]^k\) by the proof of proposition~\ref{yamashita}. Since~\cite[Eq.\ 7]{Con13}
\begin{equation}\label{continuity}
	\forall T \in \cap_k \Dom[\abs{D},\cdot]^k, \; \forall \xi \in \cH_\infty, \quad \abs{D}^m T\xi = \sum_{k=0}^m \binom{m}{k} [\abs{D},\cdot]^k(T) \abs{D}^{m-k}\xi,
\end{equation}
the map \(\phi : V \otimes \cA \surj \cH_\infty\) is surjective and continuous, and hence, by the open mapping theorem, is open. Thus, \(\rest{\phi}{e(V \otimes \cA)} = \inv{\psi}\) is a bijective and continuous, and thus is bicontinuous, as was claimed; in particular, the action of \(G\) on \(\cH_\infty\) is smooth since the action of \(G\) on \(e(V \otimes \cA)\) is. The rest of the claim then follows from bicontinuity of \(\psi\) together with joint continuity of the multiplication \(\cA \times \cA \to \cA\).
\end{proof}

Let us now turn to absolute continuity. On the one hand, recall that by absolute continuity of \((\cA,H,D)\), for some fixed Dixmier trace  \(\Tr_\omega\) on \(B(H)\), we can write
\[
	\forall \xi, \eta \in \cH_\infty, \quad \ip{\xi}{\eta} = \fint R(\hp{\xi}{\eta}) \coloneqq \Tr_\omega \left( R(\hp{\xi}{\eta}) (D^2+1)^{-p/2}\right).
\]
On the other hand, since \(\cH_\infty\) \emph{qua} right f.g.p.\ \G-\(\cA_\Theta\)-module is simply the deformation \((\cH_\infty)_\Theta\) of \(\cH\) \emph{qua} right f.g.p.\ \GA-module, it admits the deformed \(G\)-equvariant Hermitian metric \(\ip{}{}_\Theta\) defined by
\[
	\forall \xi, \; \eta \in \cH_\infty, \quad \hp{\xi}{\eta}_\Theta = \sum_{\bx,\by \in \dual{G}} \e{\Theta(\bx-\by,\bx)}\hp{\hat{\xi}(-\bx+\by)}{\hat{\eta}(\by)},
\]
which satisfies, by \G-equivariance of \(R\) and \(R_\Theta\),
\[
	\forall \xi,\;\eta \in \cH_\infty, \quad F_\bo(R_\Theta(\hp{\xi}{\eta}_\Theta) = R\left(\sum_{\bx \in \dual{G}} \hp{\hat{\xi}(\bx)}{\hat{\eta}(\bx)}\right) = F_\bo(R(\hp{\xi}{\eta})).
\]
Thus, to prove absolute continuity with respect to \(\hp{}{}_\Theta\), viz,
\[
	\forall \xi, \eta \in \cH_\infty, \quad \ip{\xi}{\eta} = \fint R_\Theta(\hp{\xi}{\eta}_\Theta) = \Tr_\omega \left( R_\Theta(\hp{\xi}{\eta}_\Theta) (D^2+1)^{-p/2}\right),
\]
it suffices to prove the following lemma.

\begin{lemma}\label{dixmier}
Let \(\Tr_\omega\) be any Dixmier trace on \(B(H)\). Then
\begin{equation}
	\forall T \in B^c(H), \quad \Tr_\omega( T (D^2+1)^{-p/2}) = \Tr_\omega (F_\bo(T) (D^2+1)^{-p/2}).
\end{equation}
\end{lemma}

\begin{proof}
For convenience, let \(\fint : B(H) \to \bC\) be defined by \(\fint T \coloneqq \Tr_\omega (S (D^2+1)^{-p/2})\) for \(S \in B(H)\). Since the Dixmier trace \(\Tr_\omega : M_{1,\infty}(H) \to \bC\) is continuous on the Lorentz ideal \(M_{1,\infty}(H)\), which is symmetric, it follows that \(\fint\) is continuous for the uniform topology on \(B(H)\). Moreover, since \(D\) is \G-invariant and since Dixmier traces are unitarily invariant, it follows that \(\fint\) is \G-invariant for the action \(\beta : G \to \Aut(B(H))\) defined by \(\beta_t(T) = U_t T U_t^\ast\) for all \(t \in G\) and \(T \in B(H)\); in particular, for any \(T \in B^c(H)\), the map \(G \ni t \mapsto \beta_t(T)\) is continuous for the uniform topology on \(B(H)\). Thus, for any \(T \in B^c(H)\),
\[
	\fint \left(\int_G U_tTU_t^\ast \dif t \right) = \int_G \left(\fint U_tTU_t^\ast\right) \dif t = \int_G \left(\fint T\right) \dif t = \fint T
\]
by continuity of \(\fint : B(H) \to \bC\) and of \(G \ni t \mapsto \beta_t(T)\).
\end{proof}

\subsubsection{Strong regularity}

By our proof above of finiteness and absolute continuity, we can exactly identify \(\cH_\infty\) \emph{qua} right Hermitian f.g.p.\ \G-\(\cA_\Theta\)-module with the deformation \((\cH_\infty)_\Theta\) of \(\cH_\infty\) \emph{qua} right Hermitian f.g.p.\ \G-\(\cA\)-module, where the topology of \(\cH_\infty\) \emph{qua} smooth domain of \(D\) can be identified with the topology of \(\cH_\infty\) \emph{qua} f.g.p. \G-\(\cA\)- or \(\cA_\Theta\)-module. Hence, by Corollary~\ref{moduledeform}, we have a \G-equvariant topological \Star-isomorphism \(\pi_\Theta : \End_\cA(\cH_\infty)_\Theta \iso \End_{\cA_\Theta}(\cH_\infty)\) that coincides precisely with the restriction of the map \(\pi_\Theta : B^\infty(H)_\Theta \iso B^\infty(H)\) of Lemma~\ref{bigdeform}, so that
\begin{multline*}
	\End_{\cA_\Theta}(\cH_\infty) = \pi_\Theta(\End_\cA(\cE)) \subset \pi_\Theta\left(B^\infty(H) \cap \left(\cap_k \Dom [\abs{D},\cdot]^k\right)\right)\\ \subset B^\infty(H) \cap \left(\cap_k \Dom [\abs{D},\cdot]^k\right).
\end{multline*}
Thus, \((\cA_\Theta,H,D)\) satisfies the last remaining condition for a \((\theta_0+\theta)\)-commu\-ta\-tive spectral triple. \(\square\)

\section{Rational Connes--Landi deformations}\label{sec:4}

\subsection{A refined splitting homomorphism}

Finally, we consider the case of \emph{rational} Connes--Landi deformations of \G-equivariant commutative spectral triples, viz, Connes--Landi deformations by a deformation parameter \(\theta \in H^2(\dual{G},\bT)\) of finite order. Our main tool will be a certain refinement of the so-called splitting homomorphism of Connes and Dubois-Violette~\cite{CDV}, which will allow us to effect rational Connes--Landi deformations using only a finite subgroup of \(G\).

We begin by recalling the Connes--Dubois-Violette splitting homomorphism. For notational convenience, we make the following definition.

\begin{definition}[{Olesen--Pedersen--Takesaki~\cite[\S 2.2]{OPT}}]
If \((\cA,\alpha)\) is a nuclear Fr\'echet \Gstar-algebra and \((\cB,\beta)\) is a Fr\'echet \Gstar-algebra, then we define their \emph{\G-product} to be the Fr\'echet \Gstar-algebra \((\cA \boxtimes_G \cB,\alpha \boxtimes_G \beta)\), where
\begin{gather*}
	\cA \boxtimes_G \cB \coloneqq \set{w \in \cA \hat{\otimes} \cB \given \forall t \in G, \; (\alpha_t \otimes \id)(w) = (\id \otimes \beta_t)(w)},\\
	\forall t \in G, \; \forall w \in \cA \boxtimes_G \cB, \quad (\alpha \boxtimes_G \beta)_t(w) \coloneqq (\alpha \otimes \id)(w) = (\id \otimes \beta_t)(w).
\end{gather*}
Similarly, if \(U : G \to \U(H)\) and \(V : G \to \U(K)\) are unitary representations of \(G\), then we define their \emph{\G-product} to be the unitary representation \(U \boxtimes_G V : G \to \U(H \boxtimes_G K)\) of \(G\), where
\begin{gather*}
	H \boxtimes_G K \coloneqq \set{\xi \in H \otimes K \given \forall t \in G, \; (U_t \otimes 1)\xi = (1 \otimes V_t)\xi},\\
	\forall t \in G, \; \forall \xi \in H \boxtimes_G K, \quad (U \boxtimes_G V)_t \, \xi \coloneqq (U_t \otimes 1)\xi = (1 \otimes V_t)\xi.
\end{gather*}
\end{definition}

Now, let \((\cA,\alpha)\) be a Fr\'echet \Gstar-algebra and let \(L : \cA \to B(H)\) be a \G-equivariant continuous \Star-representation of \(\cA\) on a Hilbert space \(H\). As was already observed by Rieffel~\cite{Rieffel}, it is convenient to \G-equivariantly identify \(\cA^\infty\) with  \(C^\infty(G,\cA)^G \cong C^\infty(G) \boxtimes_G \cA\) via the \G-equivariant topological \Star-isomorphism \(\cA^\infty \iso C^\infty(G,\cA)\) defined by \(\cA \ni a \mapsto \alpha_\bullet(a)\), for this now induces an isomorphism \(\cA_\Theta \iso C^\infty(G) \boxtimes_G \cA\) for any \(\Theta \in Z^2(\dual{G},\bT)\); this is the \emph{splitting homomorphism} of Connes and Dubois-Violette~\cite[\S 11]{CDV}, \emph{sensu stricto}. However, one can also \G-equivariantly identify \(H\) with \(L^2(G,H)^G \cong L^2(G) \boxtimes_G H\) via the \G-equivariant unitary \(H \iso L^2(G,H)^G\) defined by \(H \ni \xi \mapsto U_\bullet \, \xi\), and hence \G-equivariantly identify \(L : \cA \to B(H)\) with the \G-equivariant continuous \Star-representation \(M \boxtimes_G L : C^\infty(G) \boxtimes_G \cA \to B(L^2(G) \boxtimes_G H)\) defined by \(M \boxtimes_G L(w) \coloneqq \rest{M \otimes L}{L^2(G) \boxtimes_G H}\) for all \(w \in C^\infty(G) \boxtimes_G \cA\). Thus, we can identify the deformed \Star-representation \(L_\Theta : \cA_\Theta \to B(H)\) underpinning Connes--Landi deformation with \(M_\Theta \boxtimes_G L : C^\infty(G)_\Theta \boxtimes_G \cA \to B(L^2(G) \boxtimes_G H)\), where \(M_\Theta : C^\infty(G) \to B(L^2(G))\) is defined by \(M_\Theta(f_1)f_2 \coloneqq f_1 \star_\Theta f_2\) for all \(f_1\), \(f_2 \in C^\infty(G)\). As a result, one obtains the splitting homomorphism as an alternative construction of Connes--Landi deformation.

\begin{theorem}[{Connes--Dubois-Violette~\cite[\S\S 11, 13]{CDV}}]
Let \((\cA,H,D)\) be a \G-equivariant spectral triple and let \(\Theta \in Z^2(\dual{G},\bT)\). For each \(\bx \in \dual{G}\), let \(U_\bx \coloneqq e \circ \bx \in C^\infty(G,\U(1))\). The maps \(\phi : \cA_\Theta \to C^\infty(G)_\Theta \boxtimes_G \cA\), \(\Phi : H \to L^2(G) \boxtimes_G H\) given by
\[
	\forall a \in \cA_\Theta, \enskip \phi(a) \coloneqq \sum_{\bx\in\dual{G}} U_\bx \otimes \hat{a}(\bx),\quad
	\forall \xi \in H, \enskip \Phi(\xi) \coloneqq \sum_{\bx\in\dual{G}} U_\bx \otimes P_\bx \xi,
\]
together define a \G-equivalence
\[
	(\phi,\Phi) : (\cA_\Theta,H,D) \iso \left(C^\infty(G)_\Theta \boxtimes_G \cA,L^2(G) \boxtimes_G H, \rest{1 \otimes D}{L^2(G) \boxtimes_G H}\right),
\]
\end{theorem}

For our purposes, however, we shall need a refinement of the splitting homomorphism, which will effect Connes--Landi deformation along a \G-action by a deformation parameter \(\theta \in H^2(\dual{G},\bT)\) in terms of a canonically associated subgroup \(\im\theta\) of \(G\) and a canonically associated pushforward \(\theta_\nd \in H^2(\dual{\im \theta},\bT)\), such that the algebra \(C^\infty(\im\theta)_{\theta_\nd}\) is necessarily \emph{simple}.

\begin{definition}
Let \(\theta \in H^2(\dual{G},\bT)\). We define the \emph{image} of \(\theta\) to be the Lie subgroup
\[
	\im\theta \coloneqq \overline{\set{\iot{\theta}(\bx,\cdot) \given \bx \in \dual{G}}} = \overline{\set{\iot{\theta}(\cdot,\bx)\given \bx \in \dual{G}}} \leq G,
\]
whilst we define the \emph{kernel} of \(\theta\) to be the subgroup
\[
	\Ker\theta \coloneqq \set{\bx \in \dual{G} \given \iot{\theta}(\bx,\cdot)=0} = \set{\bx \in \dual{G} \given  \iot{\theta}(\cdot,\bx) = 0} \leq \dual{G}.
\]
\end{definition}

Now, since \(\Ker\theta = \set{\bx \in \dual{G} \given \forall t \in \im\theta, \; \ip{\bx}{t} = 0}\) is the annihilator of \(\im\theta\) in \(\dual{G}\), it follows by Pontrjagin duality that \(\dual{\im\theta}=\dual{G}/\Ker\theta\). Hence, one can make the following definition, which is the key to the question of the simplicity of \(C^\infty(G)_\Theta\).

\begin{definition}
A class \(\theta \in H^2(\dual{G},\bT)\) is called \emph{nondegenerate} if and only if \(\Ker\theta = 0\), if and only if \(\im\theta = G\).
\end{definition}

\begin{theorem}[{Slawny~\cite[Theorem 3.7]{Slawny}}]\label{slawny}
Let \(\Theta \in Z^2(\dual{G},\bT)\). The algebra \(C^\infty(G)_\Theta\) is simple if and only if \([\Theta] \in H^2(\dual{G},\bT)\) is nondegenerate.
\end{theorem}

Now, since \(\dual{\im\theta}=\dual{G}/\Ker\theta\), the alternating bicharacter \(\iot{\theta} \in A(\dual{G})\) descends to the alternating bicharacter \(\omega \in A(\dual{\im\theta})\) defined by \(\omega([\bx],[\by]) \coloneqq \iot{\theta}(\bx,\by)\) for \(\bx\), \(\by \in \dual{G}\), 
which satisfies
\(
	\set{[\bx] \in \dual{\im\theta} \given \omega([\bx],\cdot) = 0} = \set{[\bx] \in \dual{\im\theta} \given \omega(\cdot,[\bx])=0} = 0.
\)
Hence, we can make the following definition.

\begin{definition}[{cf.\ Baggett--Kleppner~\cite[Theorem 3.1]{BK}}]
Let \(\theta \in H^2(\dual{G},\bT)\). We define the \emph{nondegenerate part} of \(\theta\) to be the unique nondegenerate class \(\theta_\nd \in H^2(\dual{\im\theta},\bT)\) such that \(\iot{\theta_\nd}([\bx],[\by]) = \iot{\theta}(\bx,\by)\) for all \(\bx\), \(by \in \dual{G}\).
\end{definition}

Note that by Slawny's theorem, for any representative \(\Theta_\nd \in Z^2(\dual{\im\theta},\bT)\) of \(\theta_\nd\), the algebra \(C^\infty(\im\theta)_{\theta_\nd}\) is guaranteed to be simple. Thus, we can now give our refinement of the Connes--Dubois-Violette splitting homomorphism.

\begin{theorem}\label{newsplit}
Let \((\cA,H,D)\) be a \G-equivariant spectral triple. Let \(\theta \in H^2(\dual{G},\bT)\), let \(\Theta_\nd \in Z^2(\dual{\im\theta},\bT)\) be a representative of \(\theta_\nd\), and let \(\Theta \in Z^2(\dual{G},\bT)\) be the pullback of \(\Theta_\nd\) to a representative of \(\theta\), viz, \(\Theta(\bx,\by) \coloneqq \Theta_\nd([\bx],[\by])\) for all \(\bx\), \(\by \in \dual{G}\). For each \([\bx] \in \dual{\im\theta}\), let \(U_{[\bx]} \coloneqq e \circ [\bx] \in C^\infty(\im\theta,\U(1))\). The maps \(\phi_\theta : \cA_\Theta \to C^\infty(\im\theta)_{\Theta_\nd} \boxtimes_{\im\theta} \cA\) and \(\Phi_\theta : H \to L^2(\im\theta) \boxtimes_{\im\theta} H\) given by
\begin{gather}
	\forall a \in \cA_\Theta, \quad \phi_\theta(a) \coloneqq \sum_{[\bx]\in\dual{\im\theta}} U_{[\bx]} \otimes \left( \sum_{\by\in[\bx]} \hat{a}(\bx)\right),\\
	\forall \xi \in H, \quad \Phi_\theta(\xi) \coloneqq \sum_{[\bx]\in\dual{\im\theta}} U_{[\bx]} \otimes \left(\sum_{\by\in[\bx]} P_\by \xi\right),
\end{gather}
together define an \(\im\theta\)-equivalence
\begin{multline*}
	(\phi_\theta,\Phi_\theta) : (\cA_\Theta,H,D) \iso\\ \left(C^\infty(\im\theta)_{\Theta_\nd} \boxtimes_{\im\theta} \cA,L^2(\im\theta) \boxtimes_{\im\theta} H, \rest{1 \otimes D}{L^2(\im\theta) \boxtimes_{\im\theta} H}\right).
\end{multline*}
\end{theorem}

\begin{proof}
The isotypic subspaces of \(\cA\) and \(H\) with respect to the action of \(\im\theta \leq G\) are given, for all \([\bx] \in \dual{\im\theta} = \dual{G}/\Ker\theta\), by
\begin{align*}
	\cA_{[\bx]} &\coloneqq \set{a \in \cA \given \forall t \in \im\theta, \; \alpha_t(a) = \e{\ip{[\bx]}{t}}a} = \overline{\oplus_{\by\in[\bx]}^{\mathrm{fin}}\cA_\bx},\\
	H_{[\bx]} &\coloneqq \set{\xi \in H \given \forall t \in \im\theta, \; U_t(\xi) = \e{\ip{[\bx]}{t}}\xi} = \oplus_{\by\in[\bx]} H_\bx,
\end{align*}
respectively, so that \((\phi_\theta,\Phi_\theta)\) is exactly the Connes--Dubois-Violette splitting homomorphism for the Connes--Landi deformation \((\cA_{\Theta_\nd},H,D)\) of \((\cA,H,D)\) \emph{qua} \(\im\theta\)-equivariant spectral triple by \(\Theta_\nd\). Therefore, it suffices to show that \((\cA_\Theta,H,D) = (\cA_{\Theta_\nd},H,D)\) as \(\im\theta\)-equivariant spectral triples. However, for all \(a\), \(b \in \cA\),
\begin{multline*}
	a \star_\Theta b = \sum_{\bx,\by\in\dual{G}} \e{-\Theta(\bx,\by)}\hat{a}(\bx)\hat{b}(\by)\\ = \sum_{[\bx],[\by]\in\dual{\im\theta}} \e{-\Theta_\nd([\bx],[\by])} \left(\sum_{\bz \in [\bx]}\hat{a}(\bz)\right)\left(\sum_{\bw\in[\by]}\hat{a}(\bw)\right) = a \star_{\Theta_\nd} b,
\end{multline*}
which also shows, \emph{mutatis mutandis}, that \(L_\Theta(a)\xi = L_{\Theta_\nd}(a)\xi\) for all \(a \in \cA\) and \(\xi\) in the dense subspace \(\oplus_{\bx \in \dual{G}}^{\mathrm{alg}} H_\bx \subset \oplus_{[\bx]\in\dual{\im\Theta}}^{\mathrm{alg}} H_{[\bx]}\) of \(H\).
\end{proof}

\subsection{Rational  \texorpdfstring{\(\theta\)}{theta}-commutative spectral triples}

At last, we turn to  \(\theta\)-commu\-ta\-tive spectral triples in the case where \(\theta\) is rational in the following precise sense.

\begin{definition}
We call \(\theta \in H^2(\dual{G},\bT)\) \emph{rational} if it has finite order.
\end{definition}

\begin{example}
In the case that \(G = \bT^N\), we have that \(\theta \in H^2(\bZ^N,\bT) \cong \bT^{N(N-1)/2}\) is rational if and only if \(\theta \in (\bQ/\bZ)^{N(N-1)/2} \leq \bT^{N(N-1)/2}\).
\end{example}

Now, if \(\theta\) is rational of order \(q\), then \(\im\theta\) is contained in the \(q\)-torsion of the compact Abelian Lie group \(H^2(\dual{G},\bT)\), and as such is finite. Thus, by Theorem~\ref{slawny}, for any representative \(\Theta_\nd \in Z^2(\dual{\im\theta},\bT)\) of \(\theta_\nd\), the Fr\'echet pre-\Cstar-algebra \(C^\infty(\im\theta)_{\Theta_\nd}\) is actually a finite-dimensional simple \Cstar-algebra, and hence is \Star-isomorphic to \(M_n(\bC)\) for some \(q\); in fact, by the following theorem, \(n = q\).

\begin{theorem}[{Olesen--Pedersen--Takesaki~\cite[Theorem 5.9]{OPT}}]\label{opt}
Let \(K\) be a compact Abelian Lie group, and let \(\theta \in H^2(\dual{K},\bT)\) be nondegenerate and rational  of order \(q \in \bN\). Tthere exist a finite group \(\Gamma\) of order \(q\) and an isomorphism \(\kappa : \Gamma \times \dual{\Gamma} \iso K\), such that \(\theta = [\hat{\kappa}^\ast \Omega]\), where \(\Omega \in B(\dual{\Gamma} \times \Gamma)\) is defined by
\[
	\forall (k_1,x_1), \; (k_2,x_2) \in \dual{\Gamma} \times \Gamma, \quad \Omega((k_1,x_1),(k_2,x_2)) \coloneqq \ip{k_1}{x_2},
\]
and where the pullback \(\hat{\kappa}^\ast \Omega \in B(\dual{K})\) of \(\Omega\) is defined by
\[
	\forall \bx,\;\by\in\dual{K},\quad \hat{\kappa}^\ast\Omega(\bx,\by) \coloneqq \Omega(\hat{\kappa}(\bx),\hat{\kappa}(\by)) = \Omega(\bx \circ \kappa,\by \circ \kappa).
\]
In particular, for any representative \(\Theta \in Z^2(\dual{K},\bT)\) of \(\theta\), there exists an equivariant isomorphism \(C^\infty(K)_\Theta \iso C^\infty(\Gamma \times \dual{\Gamma})_\Omega \cong M_q(\bC)\).
\end{theorem}

Now, let \(\theta \in H^2(\dual{G},\bT)\) be rational of order \(q\) and let \((\cA,H,D)\) be a \(\theta\)-commutative spectral triple; by Theorems~\ref{thetareconstruct} and~\ref{equivariant}, without loss of generality, we can take \((\cA,H,D) \cong (C^\infty(X)_\Theta,L^2(X,E),D)\), where \(X\) is a compact oriented Riemannian \G-manifold, where \(E \to X\) is a \G-equivariant Hermitian vector bundle, where \(D\) is an essentially self-adjoint \G-invariant Dirac-type operator on \(E\), and where \(\Theta \in B(\dual{G})\) is the pullback of a bicharacter \(\Theta_\nd \in \dual{\im\theta}\) representing the nondegenerate core of \(\theta\). Then, by Theorems~\ref{newsplit} and~\ref{opt}, we can write
\begin{align*}
	&(C^\infty(X)_{\Theta},L^2(X,E),D) \\&\quad\cong_{\im\theta} \left(M_q(\bC) \boxtimes_{\im\theta} C^\infty(X), M_q(\bC) \boxtimes_{\im\theta} L^2(X,E), \rest{\id \otimes D}{M_q(\bC) \boxtimes_{\im\theta} L^2(X,E)}\right)\\
&\quad \cong_{\im\theta} \left(C^\infty(X,M_q(\bC))^{\im\theta},L^2(X,M_q(\bC) \otimes E)^{\im\theta},\rest{\id \otimes D}{L^2(X,M_q(\bC) \otimes E)^{\im\theta}}\right),
\end{align*}
where we have identified the \Star-representation \(\Lambda_{\Theta_\nd} : C^\infty(\im\theta)_{\Theta_\nd} \to B(L^2(\im\theta))\) defined by \(\Lambda_{\Theta_\nd}(f)g \coloneqq f \star_{\Theta_\nd} g\) for all \(f \in C^\infty(\im\theta)\) and \(g \in L^2(\im\theta)=C^\infty(\im\theta)\), with left matrix multiplication of \(M_q(\bC)\) on \(M_q(\bC)\) endowed with the Frobenius inner product. 

Suppose, moreover, that \(\im\theta \leq G\) acts freely and properly on \(X\), so that \(X \surj B \coloneqq X/\im\theta\) defines a compact oriented Riemannian principal \(\im\theta\)-bundle. Then, 
\[
	(C^\infty(X)_\Theta,L^2(X,E),D) \cong_{\im\theta} (C^\infty(B,A),L^2(B,A \otimes E^B),(D^B)_A),
\] 
where \(A \coloneqq X \times_{\im\theta} M_q(\bC)\), where \(E^B \coloneqq E/\im\theta\), and where \((D^B)_A\) is the twisting by the trivial flat connection on \(A\) of the restriction \(D^B\) of \(D\) to \(E^B\), to an essentially self-adjoint \(\im\theta\)-invariant Dirac-type operator on \(A \otimes E^B\). In other words, \((C^\infty(X)_\Theta,L^2(X,E),D)\) is \emph{almost-commutative} in the more general, topologically non-trivial sense first proposed by the author~\cite{Ca12,Ca13} and studied by Boeijink and Van Suijlekom~\cite{BVS} and by Boeijink and Van den Dungen~\cite{BVD}. Let us recall the relevant definitions.

\begin{definition}
Let \(X\) be a compact oriented Riemannian manifold. A \emph{self-adjoint Clifford module} on \(X\) is a Hermitian vector bundle \(E \to X\) together with a bundle morphism \(c : T^\ast X \to \End(E)\), the \emph{Clifford action}, satisfying
\begin{gather*}
	\forall \omega_1, \; \omega_2 \in C^\infty(X,T^\ast X), \quad c(\omega_1)c(\omega_2) + c(\omega_2)c(\omega_1) = -2g(\omega_1,\omega_2)\id_E,\\
	\forall \omega \in C^\infty(X,T^\ast X), \; \forall \xi_1, \; \xi_2 \in C^\infty(X,E), \quad \hp{c(\omega)\xi_1}{\xi_2} + \hp{\xi_1}{c(\omega)\xi_2} = 0.
\end{gather*}
If \(E \to X\) is a self-adjoint Clifford module, then we call an essentially self-adjoint Dirac-type operator \(D\) on \(E\) \emph{compatible} if \([D,f] = c(\dif f)\) for all \(f \in C^\infty(X)\).
\end{definition}

\begin{definition}[{cf.~\cite[Definition 1.16]{Ca13}}]
Let \(X\) be a compact oriented Riemannian manifold. We define an \emph{algebra bundle} on \(X\) to be a locally trivial bundle of finite-dimensional \Cstar-algebras on \(X\). If \(A \to X\) is an algebra bundle, we define an \emph{\(A\)-module} to be a Hermitian vector bundle \(E \to X\) together with a monomorphism \(L : A \to \End(E)\) of locally trivial bundles of finite-dimensional \Cstar-algebras on \(X\); in particular, we define a \emph{Clifford \(A\)-module} to be a self-adjoint Clifford module \(E \to X\) together with a monomorphism \(L : A \to \End(E)\) of locally trivial bundles of finite-dimensional \Cstar-algebras on \(X\), such that
\[
	\forall a \in C^\infty(X,A), \; \forall \omega \in C^\infty(X,T^\ast X), \quad [L(a),c(\omega)] = 0.
\]
\end{definition}

\begin{definition}[{\'C.~\cite[Definition 2.3]{Ca12}, cf.\ Boeijink--Van Suijlekom~\cite{BVS}, Boeijink--Van den Dungen~\cite{BVD}}]
A \emph{(concrete) almost-commutative spectral triple} is a spectral triple of the form \((C^\infty(X,A),L^2(X,E),D)\), where \(X\) is a compact oriented Riemannian manifold, called the \emph{base}, \(A \to X\) is an algebra bundle, \(E \to X\) is a Clifford \(A\)-module, and \(D\) is a compatible Dirac-type operator on \(E\).
\end{definition}

We can now summarise our discussion of rational \(\theta\)-commutative spectral triples.

\begin{theorem}\label{rational}
Let \(X\) be a compact oriented Riemannian \G-manifold, let \(E \to X\) be a \G-equivariant Hermitian vector bundle, and let \(D\) be an essentially self-adjoint \G-invariant Dirac-type operator on \(E\). Let \(\Theta \in  Z^2(\dual{G},\bT)\) be such that \([\Theta]\) is rational of order \(q \in \bN\), and suppose that \(\im[\Theta] \leq G\) acts properly on \(X\); hence, let \(B \coloneqq X/\im[\Theta]\), let \(E^B \coloneqq E/\im[\Theta]\), and let \(D^B\) denote the restriction of \(D\) to an essentially self-adjoint \G-invariant Dirac-type operator on \(E^B\). Let \(\Theta_\nd \in Z^2(\dual{\im[\Theta]},\bT)\) be a representative for the nondegenerate core of \([\Theta]\); hence, let \(A \coloneqq X \times_{\im[\Theta]} M_q(\bC)\) for \(M_q(\bC) \cong C^\infty(\im[\Theta])_{\Theta_\nd}\) \emph{qua} finite-dimensional \Cstar-algebra, so that \(A\) defines an \(A\)-module when endowed with the Hermitian metric
\[
	\forall a_1,a_2 \in C^\infty(B,A) \cong C^\infty(X,M_q(\bC))^{\im[\Theta]}, \quad \hp{a_1}{a_2} \coloneqq \Tr(a_1^\ast a_2),
\]
and the monomorphism \(L : A \to \End(A)\) of algebra bundles defined by
\[
	\forall a_1,a_2 \in C^\infty(B,A), \quad L(a_1)a_2 \coloneqq a_1a_2,
\]
Finally, let \((D^B)_A\) be the twisting of \(D^B\) by the trivial connection on \(A\) to a compatible \G-invariant Dirac-type operator on \(A \otimes E^B\). There exists an \(\im[\Theta]\)-equivalence
\[
	(C^\infty(X)_\Theta,L^2(X,E),D) \cong_{\im[\Theta]} (C^\infty(B,A),L^2(B,A \otimes E^B),(D^B)_A);
\]
in particular, \((C^\infty(X)_\Theta,L^2(X,E),D)\) is almost-commutative with base \(B\). 
\end{theorem}

In fact, by the work of Boeijink and Van den Dungen~\cite[\S 3.3]{BVD}, we can interpet this result as giving a factorisation in unbounded \(KK\)-theory of rational \(\theta\)-commutative spectral triples satisfying the hypothesis of Theorem~\ref{rational} on the action of \(\im\theta\).

\begin{corollary}[{cf.\ Boeijink--Van den Dungen~\cite[Corollary 3.9]{BVD}}]
Under the hypotheses of Theorem~\ref{rational}, \((C^\infty(X)_\Theta,L^2(X,E),D)\) \emph{qua} unbounded \(K\)-cycle for the \Cstar-completion \(C(X)_\Theta\) of \(C^\infty(X)_\Theta\), factors as an internal Kasparov product
\[
	(L^2(X,E),D) \cong (C(B,A),0,\mathrm{d}) \otimes_{C(B)} (L^2(B,E^B),D^B);
\]
In particular, the data \((C(B,A),0,\mathrm{d})\) define a morphism \((C(X)_\Theta,L^2(X,E),D) \to (C(B),L^2(B,E^B),D^B)\) in Mesland's \(KK\)-theoretic category of spectral triples~\cite{Mesland}.
\end{corollary}

\begin{remark}
Still more is true if \(H^3(B,\bZ)\) has no \(q\)-torsion, for then, the Dixmier--Douady class of the finite rank Azumaya bundle \(M_q(\bC) \to A \to B\) necessarily vanishes, and hence \(A = \End(F) = F \otimes F^\ast\) for some Hermitian vector bundle \(F \to B\), so that \((C^\infty(X)_\Theta,L^2(X,E),D)\), with respect to its \(C^\infty(X)_\Theta\)-bimodule structure, is Morita equivalent to \((C^\infty(B),L^2(B,E^B),D^B)\).
\end{remark}

Moreover, by existing results concerning the spectral action on finite normal Riemannian coverings~\cite[\S 3]{CMT}, Theorem~\ref{rational} even has implications for the spectral action on rational \(\theta\)-commutative spectral triples.

\begin{corollary}[{cf.\ \'C.--Marcolli--Teh~\cite[Theorem 3.6]{CMT}}]
Let the \(\theta\)-commutative spectral triple \((C^\infty(X)_\Theta,L^2(X,E),D)\) satisfy the hypotheses of Theorem~\ref{rational}. Let \(f \coloneqq \mathcal{L}[\phi] \circ (x \mapsto x^2) : \bR \to \bC\) for \(\mathcal{L}[\phi]\) the Laplace transform of a rapidly decreasing function \(\phi : [0,\infty) \to \bC\). For all \(\bA = \bA^\ast \in \End_{C^\infty(X)_\Theta^\op}(C^\infty(X,E))\),
\[
	\Tr f\left(\frac{1}{\Lambda}(D+\bA)\right) = \frac{1}{q^2} \Tr f \left(\frac{1}{\Lambda}(\id \otimes D + \Phi_\theta\bA\Phi_\theta^\ast)\right) + O(\Lambda^{-\infty}), \quad 0 < \Lambda \to +\infty
\]
\end{corollary}

\begin{remark}
In particular, this implies that \((C^\infty(X)_\Theta,L^2(X,E),D)\), endowed with a suitably deformed \G-invariant real structure and the (asymptotic) spectral action, can be interpreted as defining a topologically non-trivial Yang--Mills theory with gauge group \(\operatorname{PSU}(q)\) in the sense of Boeijink and Van Suijlekom~\cite{BVS}.
\end{remark}

Let us conclude by applying Theorem~\ref{rational} to the canonical example of the rational noncommutative \(2\)-torus.

\begin{example}[H{\o}egh-Krohn--Skjelbred~\cite{HKS}, cf.\ Rieffel~\cite{Rieffel83}]
Let \(\theta \in H^2(\bZ^2,\bT) \cong \bT\) be rational of order \(q \in \bN\), so that \(\theta = \exp(p/q)\) for \(p \in \bZ\) coprime with \(q\). Since \(\iot{\theta}(\bx,\by) = \theta(x_1 y_2 - x_2 y_1)\) for all \(\bx\), \(\by \in \bZ^2\), it follows that 
\[
	\im\theta = (q^{-1}\bZ^2)/\bZ^2 \cong \bZ_q \times \bZ_q, \quad \Ker\theta = q\bZ^2, \quad \dual{\im\theta} = \bZ^2/q\bZ^2 \cong \bZ_q \times \bZ_q.
\]
One can then check that the representative \(\Theta \in B(\bZ^2)\) of \(\theta\) defined by \(\Theta(\bx,\by) \coloneqq \theta x_1 y_2\) for all \(\bx\), \(\by \in \bZ^2\) descends to a representative \(\Theta_\nd \in B(\dual{\im\theta})\) of the nondegenerate core of \(\theta\), and that the isomorphism \(C^\infty(\im\theta)_{\Theta_\nd} \iso M_q(\bC)\) of Theorem~\ref{opt} is given by \(U_{[(1,0)]} \mapsto U_q\), \(U_{[(0,1)]} \mapsto V_q\), where the unitaries \(U_q\), \(V_q \in M_q(\bC)\) act on the standard ordered basis \(\set{e_k}_{k=0}^{q-1}\) of \(\bC^q\) by
\[
	\forall 0 \leq k \leq q-1, \quad U_q e_k \coloneqq e^{2\pi i k \theta} e_k, \quad V_q e_k \coloneqq e_{k-1 \bmod q};
\]
in particular, the translation action of \(\im\theta\) on \(C^\infty(\im\theta)_{\Theta_\nd}\) translates to the following action of \(\bZ_q \times \bZ_q\) on \(M_q(\bC)\):
\[
	\forall \bx,\; \by \in \bZ^2, \quad ([x_1],[x_2]) \cdot U^{y_1}_q V^{y_2}_q \coloneqq e^{2\pi i \theta(x_1 y_2 - x_2 y_1)} U^{y_1}_q V^{y_2}_q.
\]
Since \(\im\theta \leq \bT^2\) acts freely and properly on \(\bT^2\) by translation, and since \(\bT^2/\im\theta \cong \bT^2\), it follows that \(C^\infty(\bT^2)_\Theta \cong C^\infty(\bT^2,A_\theta)\), where \(A_\theta \coloneqq \bT^2 \times_{\bZ_q \times \bZ_q} M_q(\bC)\); in fact, since \(H^3(\bT^2,\bZ) = 0\), it follows that \(A_\theta \cong \End(F_\theta)\) for some Hermitian vector bundle \(F_\theta \to \bT^2\), which can also be constructed explicitly. 
\end{example}

\section*{Acknowledgments}
The author would like to thank Matilde Marcolli, his thesis advisor at the Max Planck Institute for Mathematics and at the California Institute of Technology, and Guoliang Yu, his postdoctoral mentor at Texas A\&M University, for their help, guidance, and support, and Jan Jitse Venselaar for numerous helpful conversations. The author would also like to thank Masoud Khalkhali, Andrzej Sitarz, Zhizhang Xie, and Makoto Yamashita for helpful conversations and correspondence. The author gratefully acknowledges the financial and administrative support of the Departments of Mathematics at the California Institute of Technology and at Texas A\&M University and the hospitality of the Fields Institute and of the Departments of Mathematics at the California Institute of Technology and the University of Western Ontario. Moreover, the author gratefully acknowledges the support of an AMS--Simons Travel Grant as well as earlier travel support from NSF grant DMS 1266158 in the context of the Focus Program on Noncommutative Geometry and Quantum Groups in honour of Marc A.\ Rieffel.

\appendix
\section{Corrected orientability in the \texorpdfstring{\(0\)}{0}-commutative case}

In this new appendix, we show that our correction to Definition~\ref{def:thetacomm}.\ref{item:orient} is still satisfied by a \(G\)-equivariant concrete commutative spectral triple.
Again, given a Fr\'{e}chet \(G\)-\(\ast\)-algebra \(\cA\), let \[\cA^\fin \coloneqq \bigoplus_{\bx \in \dual{G}}^{\mathrm{alg}} \cA_\bx.\]

\begin{proposition}\label{prop:appx}
	Let \((X,g)\) be a \(p\)-dimensional compact oriented Riemannian \(G\)-manifold,  \(E \to X\) a \(G\)-equivariant Hermitian	vector bundle, and \(D\) a \(G\)-invariant self-adjoint Dirac-type operator on \(E\).
	Then \((C^\infty(X),L^2(X,E),D)\) defines a \(p\)-dimensional \(0\)-commutative spectral triple with respect to the corrected definition.
\end{proposition}

\begin{lemma}[{Proof of Lemma~\ref{fgpmod}}]\label{lemma1}
	Let \(\cA\) be a Fr\'{e}chet pre-\(G\)-\(C^\ast\)-algebra, \(\cE\) a Hermitian f.g.p.\ \(G\)-\(\cA\)-module, and \(\set{\xi_1,\dotsc,\xi_n} \subset \cE\) a finite algebraic generating set for the right \(\cA\)-module \(\cE\).
	If \(\set{\eta_1,\dotsc,\eta_n} \subset \cE\) satisfies \(\norm{\xi_i-\eta_i} < \tfrac{1}{n}\) for each \(i \in \set{1,\dotsc,n}\), then it algebraically generates the right \(\cA\)-module \(\cE\).
\end{lemma}

\begin{lemma}\label{lemma2}
	Let \(X\) be a compact \(G\)-manifold.
	There exist \(\bx_1,\dotsc,\bx_m \in \dual{G}\) and non-zero \(b_1 \in C^\infty(X)_{\bx_1}\),\,\ldots, \(b_m \in C^\infty(X)_{\bx_m}\) such that \(\set{\du{b_1},\dotsc,\du{b_m}}\) algebraically generates the \(C^\infty(X)\)-module \(\Omega^1(X)\).
\end{lemma}

\begin{proof}
	First, using a finite atlas for \(X\) together with a subordinate smooth partition of unity, construct \(a_1,\dotsc,a_n \in C^\infty(X)\) such that \(\set{\du{a_1},\dotsc,\du{a_n}}\) algebraically generates the \(C^\infty(X)\)-module \(\Omega^1(X)\).
	Next, fix a \(G\)-invariant Riemannian metric \(g\) on \(X\), thereby making \(\Omega^1(X)\) into a Hermitian f.g.p.\ \(G\)-\(C^\infty(X)\)-module suitably topologized by a countable family of norms including the pre-Hilbert \(C^\infty(X)\)-module norm \(\norm{\cdot}\) on \(\Omega^1(X)\) induced by \(g\), see \S\S~\ref{ssec:peterweyl} and \ref{moduledeformsec}.
	Now, for each \(i \in \set{1,\dotsc,n}\), since the Fourier expansion \(\du{a_i} = \sum_{\bx \in \dual{G}}\widehat{\du{a_i}}(\bx)\) is, in particular, absolutely convergent with respect to \(\norm{\cdot}\), there exist finite \(F_i \subset \dual{G}\) such that \(\norm{\du{a_i} - \sum_{\bx \in F_i}\widehat{\du{a_i}}(\bx)} < \tfrac{1}{n}\).
	Thus, by \(G\)-equivariance and Fr\'{e}chet-continuity of \(\du : C^\infty(X) \to \Omega^1(X)\), it follows, for each \(i \in \set{1,\dotsc,n}\), that \(\du{a_i} = \du{a_i^\prime}\), where \(a_i^\prime \coloneqq  \sum_{\bx \in F_i}\widehat{a_i}(\bx) \in C^\infty(X)^\fin\).
	By Lemma~\ref{lemma1}, it now follows that \(\set{\du{a_1^\prime},\dotsc,\du{a_m^\prime}}\) algebraically generates the \(C^\infty(X)\)-module \(\Omega^1(X)\).	
	Finally, take \(\set{b_1,\dotsc,b_m}\) to be the distinct non-zero elements of the finite set \(\set{\widehat{a_i}(\bx) \given 1 \leq i \leq n,\,\bx \in \dual{G}}\), whence, for each \(k \in \set{1,\dotsc,m}\), we determine \(\bx_k \in \dual{G}\) by the inclusion \(b_k \in C^\infty(X)_{\bx_k}\).
\end{proof}

\begin{proof}[Proof of Proposition~\ref{prop:appx}]
	The only remaining issue is proving that the corrected orientability condition is satisfied.
	Let \(\star\) be the \(G\)-invariant Hodge star operator on \((M,g)\) with its given orientation.
	Let \(\pi_\wedge : C^\infty(X)^{\otimes(p+1)} \to \Omega^p(X)\) denote the surjective \(G\)-equivariant left \(C^\infty(X)\)-linear map given by
	\[
		\pi_\wedge(c_0 \otimes c_1 \otimes \cdots \otimes c_p) \coloneqq c_0 \cdot \du{c_1} \wedge \cdots \wedge \du{c_p}
	\]
	for all \(c_0,c_1,\dotsc,c_p \in C^\infty(X)\).
	By~\cite[Proof of Thm 11.4]{Con13}, \emph{mutatis mutandis}, it therefore suffices to find \(G\)-invariant \(0\)-antisymmetric \(\bc \in (C^\infty(X)^\fin)^{\otimes(p+1)}\) such that \(\pi_\wedge(\bc) = \star(1)\).
	
	By Lemma~\ref{lemma2}, there exist characters \(\bx_1,\dotsc,\bx_m \in \dual{G}\) and non-zero isotypical vectors \(b_1 \in C^\infty(X)_{\bx_1}\),\,\ldots ,\(b_m \in C^\infty(X)_{\bx_m}\) such that \(\set{\du{b_1},\dotsc,\du{b_m}}\) algebraically generates the \(C^\infty(X)\)-module \(\Omega^1(X)\).
	Given \(i_1,\dotsc,i_p \in \set{1,\dotsc,m}\), since
	\[
		\star(\du{b_{i_1}} \wedge \cdots \wedge \du{b_{i_p}})(t \cdot x) = \exp\mleft(2\pi i\sum_{k=1}^p \langle \bx_{i_k}, t \rangle\mright) \star(\du{b_{i_1}} \wedge \cdots \wedge \du{b_{i_p}})(x)
	\]
	for all \(x \in X\), and \(t \in G\), it follows that the zero locus of \(\du{b_{i_1}} \wedge \cdots \wedge \du{b_{i_p}} \in \Omega^p(X)\) is a \(G\)-invariant subset of \(X\).
	Thus, the compact \(G\)-manifold \(X\) admits a finite cover \(\set{U_1,\dotsc,U_q}\) by \(G\)-invariant open sets, where, for each \(1 \leq j \leq q\), there exist indices \(1 \leq i_{j;1} < \cdots < i_{j;p} \leq m\) such that \(\du{b_{i_{j;1}}} \wedge \cdots \wedge \du{b_{i_{j;p}}}\) is nowhere vanishing on \(U_j\); for convenience, set \(\by_j \coloneqq \sum_{k=1}^p \bx_{i_{j;k}}\) for each \(j \in \set{1,\dotsc,q}\).
	Finally, fix a \(G\)-invariant smooth partition of unity \(\set{\rho_1,\dotsc,\rho_q}\) for \(X\) subordinate to \(\set{U_1,\dotsc,U_q}\), and observe that
	\(
		\star(1) = \sum_{j=1}^q a_j \cdot \du{b_{i_{j;1}}} \wedge \cdots \wedge \du{b_{i_{j;p}}},
	\)
	where
	\[
		a_j \coloneqq \rho_j \cdot \star(\du{b_{i_{j;1}}} \wedge \cdots \wedge \du{b_{i_{j;p}}})^{-1} \in C^\infty_c(U_j) \cap C^\infty(X)_{-\by_j}.
	\]
	for each \(j \in \set{1,\dotsc,q}\).
	Thus, at last, we may simply take
	\[
		\bc \coloneqq \sum_{j=1}^q \frac{1}{p!} \sum_{\sigma \in S_p} (-1)^\sigma a_j \otimes b_{i_{j;\sigma(1)}} \otimes \cdots \otimes b_{i_{j;\sigma(p)}}. \qedhere
	\]
\end{proof}

\bibliography{refs}{}
\bibliographystyle{amsplain}

\end{document}